\documentclass[10pt,twocolumn,twoside]{IEEEtran}
\pdfoutput=1
\usepackage{anysize}
\marginsize{1cm}{1cm}{1cm}{1cm}
\usepackage{amsmath}
\usepackage{amsfonts}
\usepackage{math}
\usepackage{xcolor}
\usepackage{graphicx}
\usepackage{booktabs}
\usepackage{float}
\usepackage{multirow}
\usepackage{url}
\usepackage{ifdraft}

\definecolor{darkgreen}{rgb}{0 0.6 0.6}
\definecolor{darkred}{rgb}{0.7 0 0}
\definecolor{darkblue}{rgb}{0 0 0.5}
\usepackage{hyperref}
\hypersetup{linkcolor=darkred,citecolor=darkgreen,colorlinks=true,urlcolor=darkblue}

\usepackage{todonotes}

\usepackage[small,bf]{caption}
\usepackage{subfig}

\usepackage{algorithm}
\usepackage{algorithmic}

\usepackage{soul}

\newcommand{\tr}{\mathsf{T}}
\newcommand{\E}{\textrm{E}}
\newcommand{\refeq}[1]{(\ref{#1})}
\renewcommand{\th}{^\textrm{th}}

\newcommand{\refsec}[1]{Section~\ref{#1}}
\newcommand{\reffig}[1]{Figure~\ref{#1}}
\newcommand{\refalg}[1]{Algorithm~\ref{#1}}
\newcommand{\reftab}[1]{Table~\ref{#1}}

\newcommand{\mb}[1]{\textbf{#1}}

\newcounter{td}
\newcommand{\note}[2]{%
\refstepcounter{td}%
\label{cm:#1}%
\todo[linecolor=darkred,backgroundcolor=white,bordercolor=white]{%
\textcolor{darkred}{\textbf{\footnotesize C\thetd}}%
}%
\textcolor{darkred}{\textbf{#2}}%
}
\newcommand{\coverletter}{\input{CoverLetter-R1}}

 \renewcommand{\note}[2]{#2}
 \renewcommand{\st}[1]{}
 \renewcommand{\coverletter}{}

\newlength{\histwidth}
\newcommand{\ih}[1]{\includegraphics[width=\histwidth]{Histograms/#1}}

\begin{document}
\coverletter

\title{A Geometric Approach to Sound Source Localization from Time-Delay Estimates}
\author{Xavier Alameda-Pineda and Radu Horaud
\thanks{INRIA Grenoble Rh\^one-Alpes and  Universit\'e de Grenoble.}\thanks{This work was supported by the EU project
HUMAVIPS FP7-ICT-2009-247525.}}
\maketitle

\begin{abstract}
This paper addresses the problem of sound-source localization from time-delay estimates using arbitrarily-shaped
non-coplanar microphone arrays. A novel geometric formulation is proposed, together with a thorough algebraic analysis and a global optimization solver. The proposed model is thoroughly described and evaluated.
The geometric analysis, stemming from the direct acoustic propagation model, leads to necessary and sufficient conditions for a
set of time delays to correspond to a unique position in the source space. Such sets of time delays are referred to as {\em feasible sets}.
We formally prove that every feasible set corresponds to exactly one position in the source space, whose value can be
recovered using a closed-form localization mapping. Therefore we seek for the optimal feasible set of time delays given, as input, the
received microphone signals. This time delay estimation problem is naturally cast into a programming task, constrained by the
feasibility conditions derived from the geometric analysis. A global branch-and-bound optimization technique is proposed
to solve the problem at hand, hence estimating the best set of feasible time delays and, subsequently, localizing the
sound source. Extensive experiments with both simulated and real data are reported; we compare our methodology to four
state-of-the-art techniques. This comparison \note{r1:4}{\st{clearly}} shows that the proposed method combined with the
branch-and-bound algorithm outperforms existing methods. These in-depth geometric understanding, practical algorithms,
and encouraging results, open several opportunities for future work.
\end{abstract}

\section{Introduction}
\label{sec:intro}
For the past decades, source localization has been a fruitful research topic. Sound source localization (SSL) in
particular, has become an important application, because many speech, voice and event recognition systems
assume the knowledge of the sound source position. Time delay estimation (TDE) has proven to be a high-performance
methodological framework for SSL, especially when it is combined with training \cite{Deleforge12a}, statistics
\cite{So2008} or geometry \cite{Canclini13,Alameda-EUSIPCO-12}. We are interested in the development of a
general-purpose TDE-based method for SSL, i.e., TDE-SSL, and we are particularly interested in indoor environments. This is
extremely challenging for several reasons: (i) there may be several sound sources and their number varies over time,
(ii) regular rooms are echoic, thus leading to reverberations, and (iii) the microphones are often embedded in devices
(for example: robot heads and smart phones) generating high-level noise.

In this context, we focus on arbitrarily shaped non-coplanar microphone arrays, because of three main reasons. First,
microphone arrays working on real (mobile) platforms may need to accommodate very restrictive design criteria, for
which array geometries that have been traditionally studied, e.g., linear, circular, or spherical, are not well suited. 
We are particularly interested in embedding microphones into a robot head, such as the humanoid robot
NAO\footnote{\texttt{http://www.aldebaran-robotics.com/}} which possesses four microphones in a
tetrahedron-like shape. There are robot design constraints that are not compatible with a particular type of microphone
array. Moreover, solving for the most general non-coplanar microphone configuration opens the door to dynamically
reconfigurable microphone arrays in arbitrary layouts. Such methods have been already studied in the specific case of
spherical arrays \cite{Li2005}. Nevertheless, the most general case is worthwhile to be studied, since non-coplanar
arrays include an extremely wide
range of specific configurations. 

This paper has the following original contributions:
\begin{itemize}
 \item The \textit{geometric analysis} of the microphone array. We are able to characterize those time delays that
correspond to a position in the source space. Such time delays will be called \textit{feasible} and the derived
necessary and sufficient conditions will be called \textit{feasibility conditions}.
 \item A \textit{closed-form solution} for SSL. Indeed, we formally prove that every feasible set corresponds to exactly
one position in the source space. Moreover, a localization mapping is built to recover, unambiguously, the sound source position from
any set of feasible time delays.
 \item A \textit{programming framework} in which the TDE-SSL problem is cast. More precisely, we propose a criterion
for multichannel TDE designed to deal with microphone arrays in general configuration. The feasibility conditions
derived from the geometric analysis constrain the optimization of the criterion, ensuring that the final TDEs
correspond to a position in the source space.
 \item A branch-and-bound\textit{global optimization method} solving the TDE-SSL task. Once the algorithm converges,
the close-form localization mapping is used to recover the sound source position. We state and prove that the
sound source position is unique.
 \item An \textit{extensive set of experiments} benchmarking the proposed technique to the state-of-the-art. Our method
is compared to four existing methods using both simulated and real data.
\end{itemize}

The remaining \note{r1:m:2}{part} of the paper is organized as follows. \refsec{sec:sota} describes the related work.
\refsec{sec:models} briefly summarizes the signal and propagation models.
\refsec{sec:geometry} presents the full geometric analysis, together with the formal proofs. \refsec{sec:tde} casts the
TDE-SSL task into a constrained optimization task. \refsec{sec:b&b_opt} describes the branch-and-bound global
optimization technique. The proposed SSL-TDE method is evaluated and compared to the state-of-the-art in \refsec{sec:validation}.
Finally, conclusions and a discussion for future work are provided in \refsec{sec:conclusions}.

\section{Related Work}
\label{sec:sota}
The task of localizing a sound source from time delay estimates has received a lot of attention in the past; recent
reviews can be found in \cite{Seco2009,Pertila2009,Chen2006}.

One group of approaches (referred to as \textit{bichannel SSL}) requires one pair of microphones. For example
\cite{Liu10,Mandel07,Viste03,Woodruff12}  estimate the azimuth from the interaural time difference. These methods assume
that the sound source is placed in front of the microphones and it lies in a horizontal plane. Consequently, they are
intrinsically limited to one-dimensional localization. Other methods either guess both the azimuth and
elevation \cite{Kullaib09,Deleforge12b} or track them \cite{Keyrouz2006,Keyrouz2007}. These methods are based on
estimating the impulse response function, which is a combination of the head related transfer function (HRTF) and the
room impulse response (RIR). In order to guarantee the adaptability of the system, the intrinsic properties of the
recording device encompassed in the HRTF must be estimated separately from the acoustic properties of the environment,
modeled by the RIR. Furthermore, these methods lead to localization techniques which do not yield closed form
expressions, thus increasing the computational complexity. Moreover, the dependency on both HRTF and RIR of the
source position is extremely complex, hence, it is difficult to model or to estimate this dependency. In conclusion, these
methods suffer from two main drawbacks. On one side, large training sets and complex learning procedures are required.
On the other side, the estimated parameters correspond to one particular combination of environment and microphone-pair
position and orientation. Since estimating the parameters for all such possible combinations is unaffordable, these
methods are hardly adaptable to unknown environments.

A second group of methods (referred to as \textit{multilateration}) performs SSL from TDE with more than two
microphones, by first estimating pairwise time delays, followed by localizing the source from these estimates. We
note that the time delays are estimated independently of the location. In other words, the two
steps are decoupled. Moreover, the TDEs do not incorporate the geometry of the array, that is, the estimates are
computed regardless of the microphones' position. This is problematic because the existence of a source point
consistent with all TDEs is not guaranteed. In order to illustrate this potential conflict, we consider
a three-microphone linear array. Let $t_m$ be the time-of-arrival at the $m\th$ microphone, and $t_{m,n}=t_n-t_m$ be
the time delay associated to the microphone pair $(m,n)$. In the particular set up of a three-microphone linear array,
the case $t_{1,2}>0$ and $t_{3,2}>0$ is not physically possible. Indeed, this is equivalent to say that the acoustic
wave reaches the first and the third microphones \textit{before} reaching the middle one, which is inconsistent with
the propagation path of the acoustic wave. In order to overcome this issue, multilateration is formulated either as
maximum likelihood (ML) \cite{Chen2003,Sheng2005,So2008,Urruela2004,Strobel1999,Zhang2007,Zhang2008}, as least squares
(LS) \cite{Smith1987,Brandstein1997,Brandstein1997a,Friedlander1987,Huang2001,Canclini13} or as global coherence fields
(CFG) \cite{Omologo-MLMI-2006,Mori-AcademicPress-1997,Brutti-Interspeech-2005,Brutti-Interspeech-2006}. Multilateration
methods posses the advantage of being able to evaluate different TDE and SSL techniques. This allows for a better
understanding of the interactions between TDE and SSL. Unfortunately, even if the ML/LS/GCF frameworks are able to
discard TDE outliers, they can neither prevent nor reduce their occurrence. Consequently, the performance of these
methods drops dramatically when used in highly reverberant environments.

A third group of methods (referred to as \textit{multichannel SSL}) estimates all time delays at once, thus
ensuring their mutual consistency. Multichannel SSL can be further split into two sub-groups. The first sub-group
performs SSL using the TDEs extracted from the acoustic impulse responses
\cite{Doclo2003,Salvati13,Huang03,Lim13,Nakamura-IROS-2011}\note{r2:4}{}.
These responses are directly estimated from the raw data, which is very challenging. As with bichannel SSL, large
training sets and complex learning procedures are necessary. Moreover, the estimated impulse responses correspond to the
acoustic signature of the environment associated with one particular microphone-array position and orientation.
Therefore, such methods suffer from low adaptability to a changing environment. The second sub-group exploits the
redundancy among the received signals. In \cite{Chen2003a} a multichannel criterion based on cross-correlation is
proposed. Even if the method is based on pair-wise cross-correlation functions, the estimation of the time delays is
performed at once. \cite{Chen2003a} has been extended  using temporal prediction \cite{He13a} and has also proven to be
equivalent to two information-theoretic criteria \cite{He13b,Benesty07}, under some statistical assumptions.  However,
all these methods were specifically designed for \textit{linear} microphone arrays. Indeed, the line geometry is directly
embedded in the proposed criterion and in the associated algorithms. Likewise, some methods were designed for other
array geometries, such as \textit{circular} \cite{Pavlidi13} or \textit{spherical}
\cite{Sasaki12,Rafaely-Springer-2010,Rafaely-TASL-2005,Sun-ICASSP-2008} arrays. Again, the geometry is directly embedded
in the methods in both cases. Hence, all these methods cannot be generalized to microphone arrays owing a general
geometric configuration.

Recently, we addressed multichannel TDE-SSL in the case of 
arbitrary arrays, thus guaranteeing the system's adaptability \cite{Alameda-EUSIPCO-12}. TDE-SSL was modelled
as a non-linear programming task, for which a gradient-based local optimization technique was proposed. However, this
method has several drawbacks. First, the geometric analysis is incomplete. Indeed, the reported model is not valid
for arrays with more than four microphones, thus limiting its generality. Second, the local optimization algorithm needed
to be initialized on a grid. Consequently, the resulting procedure is prohibitively slow. Third, the evaluation was
carried out in scenarios with almost no reverberations and only on simulated data. Last, no complexity analysis was
performed.

Unlike most of the existing approaches on multichannel TDE, we did not embed the geometry of the array in the criterion.
Instead, the geometry of the array is incorporated as two feasibility constraints. Furthermore, our approach has
several interesting features: (i) \textit{generality}, since it is not designed for a particular array geometry and it
may accommodate several microphones, (ii) \textit{adaptability}, because the method neither constrains nor estimates the
acoustic signature of the environment,  (iii) \textit{intuitiveness}, since the entire approach is built on a simple
signal model and the geometry is derived from the direct-path propagation model, (iv) \textit{soundness}, due to the
thorough mathematical formalism underpinning the approach and (v) \textit{robustness} and \textit{reliability}, as shown
by the extensive experiments and comparisons with state-of-the-art methods on both simulated and real data.

\section{Signal and Propagation Models}
\label{sec:models}

In this Section we describe the sound propagation model and the signal model. While the first one is exploited to
geometrically relate the time delays to the sound source position (see \refsec{sec:geometry}), the second one is used
to derive a multichannel SSL criterion (see \refsec{sec:tde}). We introduce the following notations: the position of the
sound source $\Svect\in\mathbb{R}^N$, the number of microphones $M$, as well as their positions,
$\{\Mvect_m\}_{m=1}^{m=M}\in\mathbb{R}^N$. Let  $x(t)$ be the signal emitted by the source. The signal received at the
$m\th$ microphone writes:
\begin{equation}
 x_m(t) = x(t-t_m) + n_m(t),
 \label{eq:signal_model}
\end{equation}
where $n_m$ is the noise associated with the $m\th$ microphone and $t_m$ is the time-of-arrival from the source to that
microphone.  The microphones' noise signals are assumed to be  zero-mean independent Gaussian
random processes. Throughout the article, constant sound propagation speed, $\nu$, and direct propagation path are
assumed. Hence we write $t_m = \|\Svect-\Mvect_m\|/\nu$. Using this model, the expression for the time delay between the
$m\th$ and the $n\th$ microphones, $t_{m,n}$, is expressed as:
\begin{equation}
 t_{m,n} = t_n-t_m = \frac{\|\Svect-\Mvect_n\|-\|\Svect-\Mvect_m\|}{\nu}.
\label{eq:tde_model}
\end{equation}

\begin{figure}
\centering
\ifCLASSOPTIONonecolumn \includegraphics[width=0.35\linewidth]{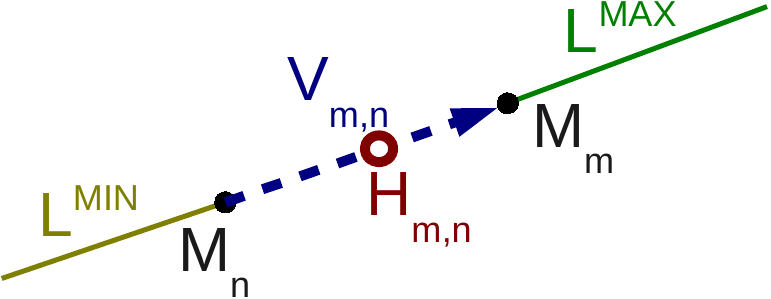}
\else \includegraphics[width=0.7\linewidth]{images/geometry.pdf}
\fi
\caption{The geometry associated with the two microphones case, located at $\Mvect_m$ and $\Mvect_n$ (see
Lemma~\ref{lem:hyperb}). $H_{m,n}$ is the mid-point of the microphones (in red) and $V_{m,n}$  designates the vector
$\Mvect_m-\Mvect_n$ (in dashed-blue). $L_{m,n}^\textrm{MAX}$ and $L_{m,n}^\textrm{MIN}$ are the two half lines
drawn in green and yellow respectively.}
\label{fig:two-microphones}
\end{figure}

\section{Geometric Sound Source Localization}
\label{sec:geometry}

We recall that the task is to localize the sound source from the TDE. In this Section we state the main theoretical
results. Firstly, we describe under which conditions a set of time delays correspond to a sound source position -- when
a sound source can be localized. Such sets will be called \textit{feasible} and the conditions, \textit{feasibility
constraints}. Secondly, we prove the uniqueness of the sound source positions for any feasible time delay set. Finally,
we provide a closed-formula for sound source localization from any feasible set of time delays. Even if, in practice
the problem is set in the source space, $\mathbb{R}^3$, the theory presented here is valid in
$\mathbb{R}^N$,$\,N\geq2$. In the following, \refsec{sec:2-mic} describes the geometry of the problem for the
two microphones case, and \refsec{sec:m-mic} delineates the geometry \note{r1:m:3}{associated} to the $M$-microphone
case in general
position.

\subsection{The Two-Microphone Case}
\label{sec:2-mic}
We start by characterizing the locus of sound-source locations corresponding to a particular time
delay estimate $\hat{t}_{m,n}$, namely $\Svect$  satisfying $t_{m,n}(\Svect)=\hat{t}_{m,n}$. Since~\refeq{eq:tde_model}
defines a hyperboloid in $\mathbb{R}^N$, this equation embeds
the \textit{hyperbolic geometry} of the problem. For completeness, we state the following lemma:
\begin{lemma}
The space of sound-source locations $\Svect\in\mathbb{R}^N$ satisfying
$t_{m,n}(\Svect)=\hat{t}_{m,n}$ is:
\renewcommand{\theenumi}{(\roman{enumi}}
\begin{enumerate}
\item the empty set if $|\hat{t}_{m,n}|>t_{m,n}^*$, where $t_{m,n}^*=\|\Mvect_m-\Mvect_n\|/\nu$;
\item the half line $\Lvect_{m,n}^\textrm{MAX}$ (or $\Lvect_{m,n}^\textrm{MIN}$), if $\hat{t}_{m,n}=t^*_{m,n}$
 (or if $\hat{t}_{m,n}=-t^*_{m,n}$), where 
\[\Lvect_{m,n}^\textrm{MAX} = \{\Hvect_{m,n} + \mu\Vvect_{m,n}\}_{\mu\geq 1/2},\]
\[\Lvect_{m,n}^\textrm{MIN} = \{\Hvect_{m,n} - \mu\Vvect_{m,n}\}_{\mu\geq 1/2},\]
$\Hvect_{m,n}=(\Mvect_m+\Mvect_n)/2$ is the microphones' middle point and $\Vvect_{m,n}=\Mvect_m-\Mvect_n$ is the
microphones' vectorial baseline (see \reffig{fig:two-microphones});
\item the hyperplane passing through $\Hvect_{m,n}$ and perpendicular to $\Vvect_{m,n}$, if $\hat{t}_{m,n}=0$; or
\item one sheet of a two-sheet hyperboloid with foci $\Mvect_m$ and $\Mvect_n$ for other values of $\hat{t}_{m,n}$.
\vspace{-0.3cm}
\end{enumerate}%
\renewcommand{\theenumi}{\arabic{enumi}}%
\label{lem:hyperb}%
\end{lemma}
\begin{proof} Using the triangular inequality, it is easy to see $-t^*_{m,n}\leq t_{m,n}(\Svect) \leq t^*_{m,n}$,
$\forall\Svect\in\mathbb{R}^N$, which proves (i). (ii) is proven by rewriting $\Svect= \Hvect_{m,n} + \mu_1\Vvect_{m,n}
+ \sum_{k=2}^N\mu_k \Wvect_k$, where $(\Vvect_{m,n},\Wvect_2,\ldots,\Wvect_N)$ is an orthogonal basis of $\mathbb{R}^
N$, and then taking the derivatives with respect to the $\mu_i$'s. In order to prove (iii) and (iv) and without loss of
generality, we can assume $\Mvect_m=\evect_1$, $\Mvect_n=-\evect_1$ and $\nu=1$, where $\evect_1$ is the first element
of the canonical basis of $\mathbb{R}^N$. Hence, $t^*_{m,n}=4$. Equation \refeq{eq:tde_model} rewrites:
\begin{equation}
(\hat{t}_{m,n})^2 + 4x_1 = -2\hat{t}_{m,n}\left((x_1+1)^2+\sum_{k=2}^N x_k^2\right)^{\frac{1}{2}},
 \label{eq:tde_particular}
\end{equation}
where $(x_1,\ldots,x_N)^\tr$ are the coordinates of $\Svect$. By squaring the previous equation we obtain:
\begin{equation}
 a(4-a) + 4a\sum_{k=2}^N x_k^2 - 4(4-a)x_1^2 = 0,
\label{eq:normalized_hyperboloid}
\end{equation}
where $a=\left(\hat{t}_{m,n}\right)^2$. Notice that if $\hat{t}_{m,n}=0$, \refeq{eq:normalized_hyperboloid} is
equivalent to $x_1=0$, which corresponds to the statement in (iii). For the rest of values of $a$, that is
$0<a<(t_{m,n}^*)^2=4$, equation \refeq{eq:normalized_hyperboloid} represents a two-sheet hyperboloid because all the
coefficients are strictly positive except the coefficient of $x_1^2$, which is strictly negative. In addition, we can
rewrite \refeq{eq:normalized_hyperboloid} as:
\begin{equation}
 x_1^2 = \frac{a(4-a) + 4a\sum_{k=2}^N x_k^2}{4(4-a)}
\label{eq:x_1}
\end{equation}
and notice that $x_1^2 >0$. We observe that the solution space of~\refeq{eq:normalized_hyperboloid} can be split into
two subspaces
$\mathcal{S}^+_{m,n}$ and $\mathcal{S}^-_{m,n}$ parametrized by $(x_2,\ldots,x_N)$, corresponding to the two solutions
of~\refeq{eq:x_1}. These two subspaces are the two sheets of the hyperboloid defined
in~\refeq{eq:normalized_hyperboloid}. Moreover, one can easily verify that $t_{m,n}\left(\mathcal{S}_{m,n}^+\right)
= -t_{m,n}\left(\mathcal{S}_{m,n}^-\right)$, so either $t_{m,n}\left(\mathcal{S}_{m,n}^+\right) = \hat{t}_{m,n}$ or
$t_{m,n}\left(\mathcal{S}_{m,n}^-\right) = \hat{t}_{m,n}$, but both equalities cannot hold simultaneously. Hence the set
of points $\Svect$ satisfying $t_{m,n}(\Svect)=\hat{t}_{m,n}$ is either $\mathcal{S}_{m,n}^+$ or $\mathcal{S}_{m,n}^-$:
one sheet of a two-sheet hyperboloid.
\end{proof}

We remark that the solutions of \refeq{eq:normalized_hyperboloid} are $\mathcal{S}_{m,n}^+\cup\mathcal{S}_{m,n}^-$.
However, the solutions of \refeq{eq:tde_particular} are either $\mathcal{S}_{m,n}^+$ or $\mathcal{S}_{m,n}^-$.
This occurs because \refeq{eq:normalized_hyperboloid} depends only on $a=(\hat{t}_{m,n})^2$, and not on
$\hat{t}_{m,n}$. Consequently, changing the sign of $\hat{t}_{m,n}$ does not modify the solutions of
\refeq{eq:normalized_hyperboloid}. In other words, the solutions of \refeq{eq:normalized_hyperboloid} contain not only
the \textit{genuine} solutions (those of \refeq{eq:tde_particular}), but also a set of \textit{artifact} solutions.
More precisely, this set corresponds to the solutions of \refeq{eq:tde_particular}, replacing $\hat{t}_{m,n}$ by
$-\hat{t}_{m,n}$. Geometrically, the solutions of \refeq{eq:tde_particular} are one sheet of a two-sheet hyperboloid,
and the solutions of \refeq{eq:normalized_hyperboloid} are the entire hyperboloid. Notice that, because $t_{m,n}({\cal
S}_{m,n}^+)=-t_{m,n}({\cal S}_{m,n}^-)$, we are always able to disambiguate the \textit{genuine} solutions from the
\textit{artifact} ones.


\subsection{The Case of $M$ Microphones in General Position}
\label{sec:m-mic}
We now consider the case of $M$ microphones in \textit{general position}, i.e., the microphones do not lie in a
hyperplane of $\mathbb{R^N}$. Firstly, we remark that, if a set of time delays
$\hat{\tvect}=\{\hat{t}_{m,n}\}_{m=1,n=1}^{m=M,n=M}\subset\mathbb{R}^{M^2}$ satisfies \refeq{eq:tde_model} $\forall
m,n$, then these time delays also satisfy the following constraints:
\begin{eqnarray*}
\hat{t}_{m,m} = 0 & & \forall m,\\
\hat{t}_{m,n} = -\hat{t}_{n,m} & & \forall m,n, \\
\hat{t}_{m,n} = \hat{t}_{m,k} + \hat{t}_{k,n} & & \forall m,n,k.
\end{eqnarray*}
As a consequence of these three equations we can rewrite any $\hat{t}_{m,n}$ in terms of
$(\hat{t}_{1,2},\ldots,\hat{t}_{1,M})$:
\begin{equation} \hat{t}_{m,n} = -\hat{t}_{1,m} + \hat{t}_{1,n} \qquad \forall m,n. \end{equation} 

\begin{figure}
\centering
\ifCLASSOPTIONtwocolumn\includegraphics[width=0.8\linewidth]{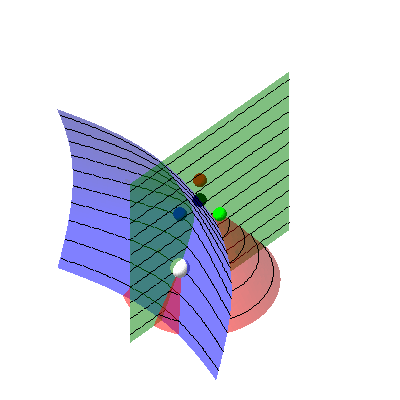}
\else\includegraphics[width=0.4\linewidth]{images/hyperboloidlocalization.png}
\fi
\caption{Localization of the source using four microphones. Their position is shown in black ($\Mvect_1$), blue
($\Mvect_2$), red ($\Mvect_3$) and green ($\Mvect_4$). The blue hyperboloid corresponds to $\hat{t}_{1,2}$, the red to
$\hat{t}_{1,3}$ and the green to $\hat{t}_{1,4}$. The intersection of the hyperboloids corresponds to the sound source
position (white marker).}
\label{fig:hyperboloid_localization}
\end{figure}

This can be written as a vector $\hat{\tvect} = (\hat{t}_{1,2} \ldots \hat{t}_{1,M})\tp$ that lies in an
$(M-1)$-dimensional vector subspace $\mathcal{W}\subset\mathbb{R}^{M^2}$. In other words, there are only $M-1$
linearly independent equations of the form~\refeq{eq:tde_model}. We remark that these $M-1$ linearly independent
equations are still coupled by the sound source position $\Svect$. Geometrically, this is equivalent to seek the
intersection of $M-1$ hyperboloids in $\mathbb{R}^N$ (see \reffig{fig:hyperboloid_localization}). Algebraically, this is
equivalent to solve a system of $M-1$ non-linear equations in $N$ unknowns. In general, this leads to finding the
roots of a high-degree polynomial. However, in our case the hyperboloids share one focus, namely $\Mvect_1$. As it will
be shown below, in this case the problem reduces to solving a second-degree polynomial plus a linear system of
equations. The $M-1$ equations \refeq{eq:tde_model} write:
\begin{equation}
 \left\{\begin{array}{ccc}
  \nu \hat{t}_{1,2} &=& \|\Svect-\Mvect_2\|-\|\Svect-\Mvect_1\|\\
  &\vdots&\\
  \nu \hat{t}_{1,M} &=& \|\Svect-\Mvect_M\|-\|\Svect-\Mvect_1\|\\
 \end{array}\right.
 \label{eq:system}
\end{equation}
Because the $M$ microphones are in general position (they do not lie in a hyperplane of $\mathbb{R}^N$),
\note{r1:m:4}{$M\geq N+1$ and the} number of equations is greater or equal than the number of unknowns. 

We now provide the conditions on $\hat{\tvect}$ under which \refeq{eq:system} yields a real and unique solution for
$\Svect$. More precisely, firstly, we provide a necessary condition on $\hat{\tvect}$ for~\refeq{eq:system} to have
real solutions, secondly, we prove the uniqueness of the solution and build a mapping to recover the solution $\Svect$,
and thirdly, we provide a necessary and sufficient condition on $\hat{\tvect}$ for~\refeq{eq:system} to have a real
and unique solution.

Notice that each equation in~\refeq{eq:system} is equivalent to
$(\nu\hat{t}_{1,m} + \|\Svect-\Mvect_1\|)^2 = \|\Svect-\Mvect_m\|^2,$ from which we obtain
$-2(\Mvect_1-\Mvect_m)^\tr\Svect + p_{1,m}\|\Svect-\Mvect_1\|+ q_{1,m} =
0,$ where $p_{1,m} = 2\nu \hat{t}_{1,m}$ and $q_{1,m} = \nu^2(\hat{t}_{1,m})^2+\|\Mvect_1\|^2-\|\Mvect_m\|^2$. Hence,
\refeq{eq:system} can now be written in matrix form:
\begin{equation}
\Mmat\Svect + \Pvect\|\Svect-\Mvect_1\| + \Qvect = 0,
\label{eq:linear_system}
\end{equation}
where $\Mmat\in\mathbb{R}^{(M-1)\times N}$ is a matrix with its $m\th$ row, $1\leq m\leq M-1$, equal to
$(\Mvect_{m+1}-\Mvect_1)^\tr$, $\Pvect = (p_{1,2},\ldots,p_{1,M})^\tr$ and $\Qvect = (q_{1,2},\ldots,q_{1,M})^\tr$.
Notice that $\Pvect$ and $\Qvect$ depend on $\hat{\tvect}$.

Without loss of generality and because the points $\Mvect_1,\ldots,\Mvect_M$ do not lie in the same hyperplane, we
assume that $\Mmat$ can be written as a concatenation of an invertible matrix $\Mmat_L\in\mathbb{R}^{N\times N}$ and a
matrix $\Mmat_E\in\mathbb{R}^{(M-N-1)\times N}$ such that $\Mmat =\left(\begin{array}{c}
\Mmat_L\\\Mmat_E\\\end{array}\right)$. \note{r3:1}{We can easily accomplish this by renumbering the microphones such
that the first $N+1$ microphones do not lie in the same hyperplane. This implies that the first $N$ rows of $\Mmat$ are
linearly independent, and therefore $\Mmat_L$ is invertible.} Similarly we have $\Pvect =
\left(\begin{array}{c}\Pvect_L \\
\Pvect_E\end{array}\right)$ and $\Qvect =\left(\begin{array}{c}\Qvect_L \\ \Qvect_E\end{array}\right)$. Thus,
\refeq{eq:linear_system} rewrites:
\begin{eqnarray}
\Mmat_L\Svect + \Pvect_L\|\Svect-\Mvect_1\| + \Qvect_L &=& 0\label{eq:split-system-loc},\\
\Mmat_E\Svect + \Pvect_E\|\Svect-\Mvect_1\| + \Qvect_E &=& 0\label{eq:split-system-ext},
\end{eqnarray}
where $\Pvect_L$, $\Qvect_L$ are vectors in $\mathbb{R}^N$ and $\Pvect_E$, $\Qvect_E$
 are vectors in $\mathbb{R}^{M-N-1}$. If we also decompose $\hat{\tvect}$ into $\hat{\tvect}_L$ and $\hat{\tvect}_E$, we
observe that $\Pvect_L$ and $\Qvect_L$ depend only on $\hat{\tvect}_L$ and that $\Pvect_E$ and $\Qvect_E$ depend only
on $\hat{\tvect}_E$. Notice that~\refeq{eq:system} is strictly equivalent
to~\refeq{eq:split-system-loc}-\refeq{eq:split-system-ext}. In the following, \refeq{eq:split-system-loc}
will be used for defining the necessary conditions on $\hat{\tvect}$ as well as localizing the sound source. The study
of~\refeq{eq:split-system-ext} is reported further on. 
By introducing a scalar variable
$w$, \refeq{eq:split-system-loc} can be written as:
\begin{eqnarray}
\Mmat_L\Svect + w\Pvect_L + \Qvect_L &=& 0,\label{eq:sub-system-lin}\\
\|\Svect-\Mvect_1\|^2 - w^2 &=& 0.\label{eq:sub-system-qua}
\end{eqnarray}

We remark that the system~\refeq{eq:sub-system-lin}-\refeq{eq:sub-system-qua} is defined in the $(\Svect,w)$
space. Notice that~\refeq{eq:sub-system-lin} a straight line and~\refeq{eq:sub-system-qua} represents a two-sheet
hyperboloid. Because two-sheet hyperboloids are not ruled surfaces, \refeq{eq:sub-system-qua} cannot contain the
straight line in~\refeq{eq:sub-system-lin}. Hence \refeq{eq:sub-system-lin} and \refeq{eq:sub-system-qua} intersect in
two (maybe complex) points.

In order to solve~\refeq{eq:sub-system-lin}-\refeq{eq:sub-system-qua}, we first rewrite~\refeq{eq:sub-system-lin} as
\begin{equation}
\Svect = \Avect w + \Bvect,
\label{eq:S(w)}
\end{equation}
where $\Avect = -\Mmat_L^{-1}\Pvect_L$ and $\Bvect = -\Mmat_L^{-1}\Qvect_L$, and then
substitute $\Svect$ from~\refeq{eq:S(w)} into~\refeq{eq:sub-system-qua} obtaining:
\begin{equation}
(\|\Avect\|^2-1)w^2 + 2\left<\Avect,\Bvect-\Mvect_1\right>w + \|\Bvect-\Mvect_1\|^2 = 0.
\label{eq:w}
\end{equation}

We are interested in the real
solutions, that is, $\Svect\in\mathbb{R}^N$. Because $\Avect,\Bvect\in\mathbb{R}^N$, the
solutions of~\refeq{eq:sub-system-lin}-\refeq{eq:sub-system-qua} are real, if and only if, the solutions
to~\refeq{eq:w} are real too. Equivalently, the discriminant of~\refeq{eq:w} has to be non-negative. Hence the
solutions to~\refeq{eq:sub-system-lin}-\refeq{eq:sub-system-qua} are real if and only if $\hat{\tvect}$
satisfies:
\begin{equation}
 \Delta(\hat{\tvect}) := \left<\Avect,\Bvect-\Mvect_1\right>^2 - \|\Bvect-\Mvect_1\|^2(\|\Avect\|^2-1)\geq0.
 \label{eq:delta}
\end{equation}
The previous equation is a \textit{necessary condition}
for~\refeq{eq:sub-system-lin}-\refeq{eq:sub-system-qua} to have real solutions. Albeit, we are interested in
the solutions of~\refeq{eq:split-system-loc}. Obviously, if $\Svect$ is a solution
of~\refeq{eq:split-system-loc}, then $(\Svect,\|\Svect-\Mvect_1\|)$ is a solution
of~\refeq{eq:sub-system-lin}-\refeq{eq:sub-system-qua}. However, the reciprocal is not true; these two systems
are not equivalent. Indeed, since $\Delta(\hat{\tvect})=\Delta(-\hat{\tvect})$, one of the solutions
of~\refeq{eq:sub-system-lin}-\refeq{eq:sub-system-qua} is the solution of~\refeq{eq:split-system-loc} and
the other is the solution of~\refeq{eq:split-system-loc} replacing $\hat{\tvect}$ by $-\hat{\tvect}$. In other
words, the two solutions of~\refeq{eq:sub-system-lin}-\refeq{eq:sub-system-qua}, namely $(\Svect^+,w^+)$ and
$(\Svect^-,w^-)$, satisfy
\begin{equation*}
\text{ either }
\left\{\begin{array}{l}
\tvect(\Svect^+)=\hat{\tvect}\\
\tvect(\Svect^-)=-\hat{\tvect}
\end{array}\right.\text{ or }
\left\{\begin{array}{l}
\tvect(\Svect^+)=-\hat{\tvect}\\
\tvect(\Svect^-)=\hat{\tvect}        
\end{array}\right..
\end{equation*}
Notice that this situation has been already encountered on equations \refeq{eq:tde_particular} and
\refeq{eq:normalized_hyperboloid}, where the same disambiguation reasoning has been used.
To summarize, the solution to~\refeq{eq:split-system-loc} is \textit{unique}. Moreover, we can use~\refeq{eq:S(w)} to
define the following \textit{localization mapping}, which retrieves the sound-source position from a feasible
$\hat{\tvect}$:
\begin{equation}
L(\hat{\tvect}):=\left\{\begin{array}{ll}
\Svect^+=\Avect w^+ + \Bvect & \text{if } \tvect(\Svect^+)=\hat{\tvect}\\
\Svect^-=\Avect w^- + \Bvect & \text{otherwise}.
\end{array}\right.
\label{eq:L}
\end{equation}
Until now we provided the condition under which
equation~\refeq{eq:split-system-loc} yields real solutions, the uniqueness of the solution and a localization
mapping. However, the original system includes also equation~\refeq{eq:split-system-ext}. In fact,
\refeq{eq:split-system-ext} adds $M-N-1$ constraints onto $\hat{\tvect}$.
Indeed, if $L(\hat{\tvect})$ is a solution of~\refeq{eq:split-system-loc}, then in
order to be a solution of~\refeq{eq:split-system-loc}-\refeq{eq:split-system-ext}, it has to satisfy:
\begin{equation}
{\cal E}(\hat{\tvect}):=\Mmat_E L(\hat{\tvect}) + \Pvect_E\|L(\hat{\tvect})-\Mvect_1\| + \Qvect_E = 0. 
\label{eq:E}
\end{equation}
Moreover, the reciprocal is true. Summarizing, the system~\refeq{eq:split-system-loc}-\refeq{eq:split-system-ext} has
a unique real solution $L(\hat{\tvect})$ if and only if $\Delta(\hat{\tvect})\geq0$ and ${\cal E}(\hat{\tvect})=0$.

It is interesting to discuss these findings from three different perspectives: geometric, algebraic, and computational:
\begin{enumerate}
\item \textit{The differential geometry point of vew.}
The set of feasible time delays, 
\[\mathcal{T} = \left\{ \hat{\tvect}\in\mathcal{W}, \Delta(\hat{\tvect})\geq 0 \text{ and } {\cal
E}(\hat{\tvect})=0\right\}, \]
is a bounded $N$-dimensional manifold with boundary lying in a $(M-1)$-dimensional vector subspace of
$\mathbb{R}^{M^2}$, Indeed, because ${\cal E}$ is a $(M-N-1)$-dimensional vector-valued function, $\mathcal{T}$ has
dimension $N$. The boundary of $\mathcal{T}$ is the set
\[\partial\mathcal{T}=\left\{\hat{\tvect}\in\mathcal{W}\middle|\Delta(\hat{\tvect})\geq0 \text{ and } {\cal
E}(\hat{\tvect})=0 \right\}.\]
In this context, the localization mapping must be seen as a smooth bijection from $\mathcal{T}$ to $\mathbb{R}^N$,
i.e., an isomorphism between the manifolds.
\item \textit{The algebraic point of view.}
 $\Delta$ and ${\cal E}$ characterize the time delays corresponding to a position in the source space,
$\mathbb{R}^N$. That is to say that $\Delta$ and ${\cal E}$ represent the feasibility constraints, the necessary and
sufficient conditions for the existence of $\Svect$. Under this conditions $\Svect$ is unique and given by the
closed-form localization mapping $L$.
\item \textit{The computational point of view.}
The mappings $\Delta (\hat{\tvect})$, ${\cal E}(\hat{\tvect}) $ and $L(\hat{\tvect})$  which are computed from
\refeq{eq:delta}, \refeq{eq:E}, and \refeq{eq:L} are expressed in closed-form and they only depend on the microphone
locations. The most time-consuming part of these computations is the inversion of the \textit{microphone matrix}
$\Mmat_L$, which can be performed off-line. Consequently, the use of these three mappings is intrinsically efficient.
\end{enumerate}

To conclude, we highlight that $\Delta (\hat{\tvect})$ and ${\cal E} (\hat{\tvect})$, i.e., \refeq{eq:delta} and
\refeq{eq:E}, provide the conditions under which the  $M-1$ time delays correspond to a
\textit{valid point} in $\mathbb{R}^N$. \textit{If these conditions are satisfied, the problem yields a unique location
for the source $\Svect$}. Moreover, the mapping $L (\hat{\tvect})$ defined by \refeq{eq:L} is a closed-form sound-source
localization solution
for any set of  \textit{feasible} time delays $\hat{\tvect}$:
\begin{equation}
 \Svect = L(\hat{\tvect}).
\label{eq:localization}
\end{equation}

\section{Time Delay Estimation}
\label{sec:tde}
In the previous section we described how to characterize the feasible sets of time delays and how to localize a sound
source from them. We now address the problem of how to obtain an optimal set of time delays given the perceived acoustic
signals. In the following, we delineate a criterion for multichannel time delay estimation
(\refsec{sec:mtde_criterion}), which will subsequently be used in \refsec{sec:task} to cast the TDE-SSL problem into a
non-linear multivariate constrained optimization task. Indeed, the multichannel TDE criterion is a non-linear cost
function allowing to choose the best value for $\hat{\tvect}$. The feasibility constraints derived in the previous
section are used to constrain the optimization problem, thus seeking the optimal feasible value for $\hat{\tvect}$.

\subsection{A Criterion for Multichannel TDE}
\label{sec:mtde_criterion}
The criterion proposed in \cite{Chen2003a} is built from the
theory of linear predictors and presented in the framework of \textit{linear} microphone arrays. Following a similar
approach, we propose to generalize this criterion to arrays owing a general microphone configuration. Given the $M$
perceived signals $\{x_m(t)\}_{m=1}^{m=M}$, we would like to estimate the time delays between them. As explained before,
only $M-1$ of the delays are linearly independent. Without loss of generality we choose, as above, the delays
$t_{1,2},\ldots,t_{1,m},\ldots,t_{1,M}$. We select $x_1(t)$ as the reference signal and set the following prediction
error:
\begin{equation}
e_{\cvect,\tvect}(t) = x_1(t) - \sum_{m=2}^M c_{1,m}\,x_m(t+t_{1,m}),
\label{eq:prediction_error}
\end{equation}
where $\cvect=\left(c_{1,2},\ldots,c_{1,m},\ldots,c_{1,M}\right)^\tr$ is the vector of the prediction coefficients and
$\tvect=\left(t_{1,2},\ldots,t_{1,m},\ldots,t_{1,M}\right)^\tr$ is the vector of the prediction time delays. Notice also
that, when $\tvect$ takes the true value, the signals $x_m(t+t_{1,m})$ and $x_n(t+t_{1,n})$ are on phase. The criterion
to minimize is the expected energy of the prediction error \refeq{eq:prediction_error}, leading to an unconstrained
optimization problem:
\begin{equation*}
 (\cvect^*,\tvect^*) = \arg\min_{\cvect,\tvect} \E\left\{e_{\cvect,\tvect}^2(t)\right\}.
\end{equation*}
In addition, it can be shown (see \cite{Chen2003a}) that this problem is equivalent to:
\begin{equation}
 \tvect^* = \arg\min_{\tvect}J(\tvect),
\label{eq:criterion}
\end{equation}
with
\begin{equation}
J(\tvect)=\det\left(\Rmat(\tvect)\right),
\label{eq:J}
\end{equation}
$\Rmat(\tvect)\in\mathbb{R}^{M\times M}$ being the real matrix of
normalized cross-correlation functions evaluated
at $\tvect$. That is $\Rmat(\tvect) = \left[\rho_{m,n}(t_{m,n})\right]_{m,n}$ with:
\begin{equation}
\rho_{m,n}(t_{m,n}) = \frac{\E\left\{x_m(t+t_{1,m})x_n(t+t_{1,n})\right\}}{\sqrt{E_m E_n}},
\label{eq:rho}
\end{equation}
where $E_m=R_{m,m}(0)=\E\left\{x_m^2(t)\right\}$ is the energy of the $m\th$ signal.

Importantly, the criterion $J$ in \refeq{eq:J} is designed to deal with microphones in a general configuration and not
for a specific microphone-array geometry. Hence, this guarantees the generality of the proposed approach. Because the
array's geometry is not embedded in $J$ (as it is done in \cite{Chen2003a}), $J$ is a multivariate function. In the next
Section, the feasibility constraints previously derived are combined with this cost function to set up a constrained
multivariate optimization task. 

\subsection{The Constrained Optimization Formulation}
\label{sec:task}
So far we characterized the \textit{feasible} values of $\tvect$, i.e., those corresponding to a sound source position
(\refsec{sec:geometry}) and introduced a criterion to choose
the best value for $\hat{\tvect}$ (\refsec{sec:mtde_criterion}). The next step is to look for the best value among the
feasible ones. This will be referred to as the
\textit{geometrically-constrained time delay estimation} problem, which naturally casts into the following
\textit{non-linear multivariate constrained optimization} problem:
\begin{equation}
\left\{
\begin{array}{l}
 \displaystyle \min_{\tvect} J(\tvect),\vspace{0.2cm}\\
 \textrm{s.t.} \quad \tvect\in{\cal W}\cap{\cal B}, \quad \Delta\left(\tvect\right)\geq0,\quad {\cal
E}\left(\tvect\right)=0,
\end{array}
\right.
\label{eq:const-optim}
\end{equation}
where ${\cal W}$, $\Delta$ and ${\cal E}$ are defined in \refsec{sec:geometry} and ${\cal B}$ is a compact set defined
as:
\begin{equation}
{\cal B} = \left\{\tvect\in\mathbb{R}^{M^2} \middle| |t_{m,n}|\leq t_{m,n}^*, \forall m,n=1,\ldots,M\right\}.
\label{eq:B_set}
\end{equation}

It is worth noticing that, in practice, the dimension of the optimization task is $M-1$. Indeed, since all time delays
can be expressed as a function of $(t_{1,2},\ldots,t_{1,M})$, the optimization is done with respect to these $M-1$
variables. We also remark that the optimization variables lay in a bounded space (as described in
\refsec{sec:2-mic}). The equality constraint is trivial when $M=N+1$. In other words, in a real world scenario, this
constraint does not exist when the array consists on $4=3+1$ microphones. When using five or more microphones, the
condition could be relaxed to $\|{\cal E}(\tvect)\|^2\leq\epsilon$, which is often more adapted to the existing
optimization algorithms. 

\note{r1:1}{We would like to highlight that all $M$ microphone signals are used in the estimation procedure. In a sense,
all received signals affect the estimation of all the time delays. This is why there is one $(M-1)$-dimensional
optimization task and not several one-dimensional optimization tasks. The localization is carried out immediately after
the time delay estimation thanks to a closed-form solution \refeq{eq:localization}, thus with no other estimation
procedure. The power of the proposed method relies on the intrinsic relation between signals and time delay estimates
combined with the use of the geometric constraints given by the microphones position. By adding these constraints to the
estimation procedure, we do not discard any infeasible sets, but we prevent them to be the outcome of our algorithm. In
other words, the estimation procedure will always provide a set of time delays corresponding to a position in the
sound source space.} Next Section describes the branch \& bound global optimization technique proposed to solve
\refeq{eq:const-optim}.

\section{Branch \& Bound Optimization}
\label{sec:b&b_opt}
Global optimization is, in most cases, an extremely challenging task. Nevertheless, the optimization
of \refeq{eq:const-optim} is well suited for a global optimizer. \note{r3:2}{Indeed, $J$ is continuously differentiable
on ${\cal B}$, therefore $\nabla J$ is continuous. This implies that $\nabla J$ is bounded on any compact set,
in particular on ${\cal B}$. Hence, by means of theorem 9.5.1 in \cite{Searcoid-Springer-2006}, $J$ is Lipschitz on
${\cal B}$.} Subsequently, a branch \& bound (B\&B) type of algorithm is well suited.

Such optimization techniques were initially proposed for linear mixed-integer programming \cite{Dakin1965} and extended
later on to the non-linear case \cite{Leyffer2001}. They alternate between the branch and bound procedures in order to
recursively seek the potential regions where the global minimum is. While the branch step splits the potential regions
into smaller pieces, the bound step estimates the lower and upper bounds of each potential region. After the bounding,
the \textit{discarding threshold} is set to the minimum of the upper bounds. Then, all regions whose lower bound is
bigger than the discarding threshold are discarded (since they cannot contain the global minimum).

The B\&B algorithm that we propose maintains two lists of regions: ${\cal P}$ containing the potential regions and
${\cal D}$ containing the discarded regions (see \refalg{alg:bab}). The B\&B inputs are the initial list of potential
regions and the Lipschitz constant $L$. The outputs are the two maintained lists, ${\cal P}$ and ${\cal D}$ after
convergence. Each of the regions in ${\cal P}$ and in ${\cal D}$ represents a $(M-1)$-dimensional cube $(\tvect,s)$,
where $\tvect$ is the cube's centre and $s$ is its side's length. The Branch routine splits each of the cubes in ${\cal
P}$ into $2^{M-1}$ smaller cubes of size $s/2$. Next, the bound routine (see \refalg{alg:bound}) estimates the upper
($u$) and the lower ($l$) bounds of all regions in ${\cal P}$. The discarding threshold $\tau$ is then set to the
minimum of all upper bounds. All sets in ${\cal P}$ with lower bound higher than $\tau$ are moved to the discarded list 
${\cal D}$. One prominent feature of the optimization task in~\refeq{eq:const-optim} us that we seek for the minimum
on the set ${\cal B}$. The set of potential sets, ${\cal P}$ is naturally initialized to the set ${\cal B}$.
Consequently, the B\&B procedure does not require a grid-based optimization.

\begin{algorithm}
\caption{Branch and Bound}
\label{alg:bab}
\begin{algorithmic}[1]
\STATE \textbf{Input:} The Lipschitz constant $L$ and the initial list of potential regions ${\cal P}$.
\STATE \textbf{Output:} The list of potential solutions ${\cal P}$ and the list of discarded regions ${\cal D}$.
\REPEAT
\STATE \textbf{(a) } ${\cal P}$ = Branch(${\cal P}$)
\STATE \textbf{(b) } [${\cal P}$, ${\cal RD}$] = Bound(${\cal P}$,$L$)
\STATE \textbf{(c) } ${\cal D} = {\cal D} \cup {\cal RD}$
\UNTIL{Convergence}
\end{algorithmic}
\end{algorithm}

\begin{algorithm}
\caption{Bound routine of \refalg{alg:bab}}
\label{alg:bound}
\begin{algorithmic}[1]
\STATE \textbf{Input:} The Lipschitz constant $L$ and the list of potential regions ${\cal P}$.
\STATE \textbf{Output:} The list of potential regions updated ${\cal P}$ and the list of recently discarded
regions ${\cal RD}$.
\FOR{$i=1,\ldots,|{\cal P}|$}
\STATE \textbf{(a) } $l^{(i)}=J(\tvect^{(i)})-s^{(i)}L$.
\STATE \textbf{(b) } $u^{(i)}=J(\tvect^{(i)})+s^{(i)}L$.
\STATE \textbf{(c) } $\tau = \min_{i=1,\ldots,|{\cal P}|} u^{(i)}$.
\ENDFOR
\FOR{$i=1,\ldots,|{\cal P}|$}
\IF{$l^{(i)}>\tau$}
\STATE Move $(\tvect^{(i)},s^{(i)})$ from ${\cal P}$ to ${\cal RD}$.
\ENDIF
\ENDFOR
\end{algorithmic}
\end{algorithm}

The branch and bound routines are alternated until convergence. Many criteria could be used to stop the algorithm.
In order to guarantee an accurate solution, we may force a maximum size for the potential regions. Also, in case
we would like to guarantee the stability of the solution, we may track the variation of the smallest of the lower
bounds. Either way, once the algorithm has converged, we select the best region in  ${\cal P}$ among those satisfying
the constraints $\Delta$ and ${\cal E}$. If there is no such region in ${\cal P}$\note{r3:3}{\footnote{In order to
decide whether a region satisfies the constraints or not, we test its centre. This approximation is justified by the
fact that, at this stage of the algorithm, the regions are extremely small (since we force a maximum region size).}},
the B\&B algorithm is run again
providing ${\cal D}$ as the initial list of potential regions.

\section{Experimental Validation}
\label{sec:validation}

\subsection{Experimental Setup}
\label{sec:set_up}
In order to validate the proposed model and the associated estimation technique, we used an evaluation protocol with
simulated and real data. In both cases, the environment was a room of approximately $4\times4\times4$ meters ($N=3$),
with an array of $M=4$ microphones placed at (in meters)  $\Mvect_1=(2.0,2.1,1.83)^\tr$, $\Mvect_2=(1.8,2.1,1.83)^\tr$,
$\Mvect_3=(1.9,2.2,1.97)^\tr$ and $\Mvect_4=(1.9,2.0,1.97)^\tr$. \textit{The microphones are the vertices of a
tetrahedron, resulting in a non-coplanar configuration}. The sound source was placed on a sphere of $1.7$~m radius
centred at the microphone array. More precisely, the source was placed at $21$ different azimuth values, between
$-160^\circ$ and $160^\circ$, and at $9$ different elevation values between $-60^\circ$ and $60^\circ$, hence at $189$
different directions. The speech fragments emitted by the source \note{r1:m:5}{were} randomly chosen from a publicly
available data set~\cite{Garofolo93}. One hundred millisecond cuts of these sounds were used as input of the evaluated
methods.

In the simulated case, we controlled two parameters. Firstly, the value of $T_{60}$, which is a parameter of the
image-source model~\cite{Lehmann2008} (available at~\cite{LehmanISMCode}), controlling the amount of reverberations.
More precisely, $T_{60}$ measures the time needed for the emitted signal to decay $60$~dB. The higher the $T_{60}$, the
larger the amount of reverberations and their energy. In our simulations, $T_{60}$ took the following values (in
seconds): $0$, $0.1$, $0.2$, $0.4$ and $0.6$. Secondly, we controlled the amount of white noise added to the
received signals by setting the signal-to-noise ratio ($SNR$) to $-10$, $-5$, or $0$~dB. 

In the real case, we used a slightly modified version of the acquisition protocol defined in~\cite{Deleforge12a}. This
protocol was designed to automatically gather sound signals coming from different directions, using the motor system of
a robotic platform placed in a regular indoor room. \note{r2:2}{Such realistic environment inherently limits the quality
of the recordings. First of all, the noise of the acquisition device (the computer's fans) is also recorded. Second,
ambient noise associated with the room's location (between a corridor and a server room) has also a negative, but very
realistic effect on the data.} \note{t60}{We roughly estimated the acoustic characteristics of the real recordings,
namely $T_{60}\approx0.5\,\textrm{s}$ and $SNR\approx0\,\textrm{dB}$.} In our case, we replaced the dummy head used
in~\cite{Deleforge12a}, by a tetrahedron-shaped microphone array. The motor platform has two degrees of freedom (pan and
tilt), and was designed to guarantee the repeatability of the movements. A loud-speaker was placed $1.7\,\textrm{m}$
away from the array to emulate the sound source. We recorded sound waves coming from $189$ different directions. 
Consequently, the performed tests cover an extensive range of realistic situations.

\subsection{Implementation}
\label{sec:implementation}
We implemented several methods and compared them within the previously described set up. \textit{b\&b} stands for the
branch \& bound method proposed in this paper. \textit{unc}, \textit{d-lb} and \textit{s-lb}, stemmed from
\cite{Alameda-EUSIPCO-12}, represent the state-of-the-art in multichannel TDE-SSL for arbitrarily shaped microphone
arrays. \textit{dm} stands for ``direct multichannel'', which is the generalization of \cite{Chen2003a} to arrays with
arbitrary configuration. \textit{n-mult}, \textit{t-mult}, and \textit{f-mult} are variants of \cite{Canclini13}, which
represent the state-of-the-art in multilateration methods. \textit{pi} is a straightforward multilateration algorithm.
We have chosen to implement all nine methods to (i) compare the proposed algorithm to the state-of-the-art and (ii)
push the limits of existing TDE-SSL algorithms of both multilateration and multichannel SSL disciplines.

\begin{itemize} 
\item \textit{b\&b} corresponds to the branch \& bound algorithm described in \refsec{sec:b&b_opt}. The list of
potential regions is naturally initialized as ${\cal P} = \{{\cal B}\}$. The Lipschitz constant $L$ is estimated by
computing the maximum slope among one thousand point pairs randomly drawn inside the feasibility domain. 

\item \textit{unc}, \textit{d-lb} and \textit{s-lb} are directly derived from the procedure described
in \cite{Alameda-EUSIPCO-12}. All of them perform multichannel TDE to further localize the sound source. \textit{unc}
and \textit{s-lb} are proposed here to push the limits of the base method, \textit{d-lb}. These local optimization
techniques are initialized on the unconstrained (${\cal G}_u$), constrained (${\cal G}_c$) and sparse (${\cal G}_s$)
grids respectively. The details are given in \refsec{sec:details_eusipco}.

\item \textit{dm} is the straightforward generalization of \cite{Chen2003a} to arbitrarily-shaped microphone arrays.
$J$, defined in \refeq{eq:J} is evaluated on ${\cal G}_c$, and the minimum over the grid is selected. The difference
between \textit{dm} and \textit{d-lb} is that in the former no local minimization is carried out.

\item  \textit{n-mult}, \textit{t-mult} and \textit{f-mult} are implementations of the method described in
\cite{Canclini13}. In this case the time delay estimates, $\hat{\tvect}$ are computed independently (using
\cite{Knapp1976}), and the sound source position, $\Svect$, is chosen to be as close as possible to the hyperboloids
associated with $\hat{\tvect}$. Because the algorithm was designed for distributed sensor networks and not for
egocentric arrays, we had to modify it. Further explanations are given in \refsec{sec:details_canclini}.

\item \textit{pi} corresponds to pair-wise independent time delay estimation based on cross-correlation
\cite{Knapp1976}. That is, $t_{1,j}$ is the maximum of the function $\rho_{1,j}(\tau)$; This is the simplest
multilateration algorithm one can think of.
\end{itemize}
Except for \textit{n-mult}, \textit{t-mult}, and \textit{f-mult}, which provide $\Svect$ directly, all other algorithms
provide a time delay estimate. If this estimate is feasible, $\Svect$ is recovered using \refeq{eq:localization}.

\subsubsection{\emph{Methods} unc, d-lb \emph{and} s-lb}
\label{sec:details_eusipco}
In \cite{Alameda-EUSIPCO-12}, the constrained problem is converted into an unconstrained problem with a different cost
function. The intuition is that the cost function is modified to penalize those points that are closer to the
feasibility border. In practice, the inequality constraint is added to the cost by means of a log-barrier function
\begin{equation}
\left\{
\begin{array}{l}
 \displaystyle \min_{\tvect} J(\tvect) - \mu\log(\Delta(\tvect)),\vspace{0.2cm}\\
 \textrm{s.t.} \quad \tvect\in{\cal W}\cap{\cal B}, \quad {\cal E}\left(\tvect\right)=0,
\end{array}
\right.
\label{eq:dip-optim}
\end{equation}
where $\mu\geq0$ is a regularizing parameter.

Consequently, the original task \refeq{eq:const-optim} is converted into a sequence of tasks indexed by $\mu$.
Each of the problems has an optimal solution $\hat{\tvect}_\mu$. It can be proven (see \cite{Boyd04}) that
$\hat{\tvect}_\mu\rightarrow\hat{\tvect}$ when $\mu\rightarrow0$. Log-barrier methods are gradient-based techniques,
which decrease the value of $\mu$ with the iterations, thus converging to the closest feasible local minimum of $J$.
Therefore, it is recommended to provide the analytic derivatives in order to increase both the convergence speed and
the accuracy (see Appendices~\ref{sec:criterion_grad_hess} and \ref{sec:constraint_grad_hess} for the expressions of
the gradients and Hessians of the cost function and the constraints, respectively).

Unfortunately, log-barrier methods are designed for convex problems. In other words, these methods find the local
minimum closest to the initialization point. Hence, in order to find the global minimum, the algorithm must be
multiply initialized from points lying on a grid ${\cal G}$. After convergence, the minimum among all the local minima
found is assumed to be the global solution of the problem. \textit{d-lb} (dense-log-barrier) corresponds to the method
in \cite{Alameda-EUSIPCO-12}, hence solving for \refeq{eq:const-optim}, initialized on a grid ${\cal G}_c$ of 352
feasible points. \textit{unc} solves for the unconstrained problem, i.e., \refeq{eq:criterion}, and it is initialized on
a grid ${\cal G}_u$. The difference between the two grids is that while ${\cal G}_c$ contains just feasible points,
${\cal G}_u$ contains unfeasible points as well. In practice, ${\cal G}_u$ contains 456 points. The rationale of
implementing \textit{unc} is to better assess and quantify the role played by the feasibility constraints, $\Delta$ and
${\cal E}$. \textit{s-lb} (sparse-log-barrier) corresponds to the same log-barrier method initialized on a
sparse grid ${\cal G}_s$. We conjecture that the global minimum of $J$ corresponds to one of the local maxima of
$\rho_{1,m}$ in \refeq{eq:rho} for $m=2,\ldots,M$. For each microphone pair $(1,m)$ we extract $K=3$ local maxima of
$\rho_{1,m}$. ${\cal G}_s$ consists of all possible combinations of these values, thus containing $K^{M-1}=27$ points
(in the case of $M=4$ microphones). \textit{s-lb} is implemented to assess the robustness towards initialization of the
local optimization technique. Both, \textit{d-lb} and \textit{s-lb} are reimplementations of the publicly available
MATLAB log-barrier dual interior-point method \cite{IPSolverCode}.

\subsubsection{\emph{Methods} n-mult, t-mult \emph{and} f-mult}
\label{sec:details_canclini}
As already mentioned, we implemented \cite{Canclini13} with some modifications. Indeed, the method was designed for
\textit{distributed} microphone arrays. With such a setup, the sound source position lies inside the volume defined by
the microphone positions in the room. In the case of an egocentric array the sound source is necessarily located outside
the volume delimited by the microphone array. The method  described in \cite{Canclini13} seeks the locations the closest
to the hyperboloids given by the independently estimated time delays $\hat{\tvect}$. More precisely, the following
criterion is minimized:
\begin{equation}
 H(\Svect) = \sum_{1\leq m<n\leq M}\left(h_{m,n}(\Svect)\right)^2,
\label{eq:criterion_canclini}
\end{equation}
where $h_{m,n}$ is the equation of a two-sheet hyperboloid (i.e., equation \refeq{eq:normalized_hyperboloid}) with foci
$\Mvect_m$ and $\Mvect_n$ and differential value $\hat{t}_{m,n}$. We will call
$H$ the \textit{multilateration cost function}, to distinguish is from the cost function $J$. If the estimated time
delays are feasible, the minimization of \refeq{eq:criterion_canclini} is equivalent
to solve the system $\left\{h_{m,n}(\Svect)=0\right\}_{1\leq m<n\leq M}$ or to compute $L(\hat{\tvect})$, otherwise the
methods seeks the value of
$\Svect$ that best explains the TDE. We have experimentally observed that, in most of
the cases, one solution is ``inside'' the microphone array and the other one is ``outside'' the array. However, the cost
function behaves differently around these two solutions. Indeed, $H$ is much
sharper around the solution inside the microphone array. Hence, this is usually the one found by the optimization
procedure. That is why we had to modify the cost function in order to bias the
optimization:
\begin{equation}
 \tilde{H}(\Svect) = H(\Svect) + \lambda\left(\|\Svect\|^2-r\right)^2,
 \label{eq:final_canclini}
\end{equation}
where $r$ is the desired radius of the solution and $\lambda$ is a regularization parameter. This way of constraining
the optimization problem is justified by the well-known fact that the distance to
the source  is very difficult to estimate with egocentric arrays. We tested three different values for $r$:
\textit{n-mult} (near-multilateration) corresponds to $r=0.9$, \textit{t-mult} (true-multilateration) corresponds to
$r=1.7$ (which is the actual distance from the array to the source), and \textit{f-mult} (far-multilateration)
corresponds to $r=2.5$ (all measures are in meters). In all cases the optimization procedure was
initialised on a grid of 200 directions, ${\cal G}_d$, and, in order to increase the accuracy and convergence speed, the
analytic derivatives were computed (see Appendix~\ref{sec:derivatives_H}).

\subsection{Results and Discussion}
\label{sec:results}
\begin{table*}
 \caption{Results obtained with both simulated data and real data (last column). The first row shows the SNR ratio in dB.
The second row shows the values of $T_{60}$ in seconds. The remaining rows show the results with the methods outlined in
Section~\ref{sec:implementation}. For each SNR--$T_{60}$ combination and for each method, we display three values: (i)~the
proportion of inliers (the angular error is less than $30^\circ$), (ii) the mean angular error of inliers, and  (iii) their
standard deviation. Column wise, the best results are shown in \textbf{bold}. Notice that \textit{x-mult} methods often
yield the best mean angular error of inliers. However, \textit{x-mult} requires an estimate of the array-to-source
distance.}
  \label{tab:results}
\centering
\ifCLASSOPTIONtwocolumn
{\scriptsize
\begin{tabular}{c|ccccc|ccccc|ccccc||c}
\toprule
SNR& \multicolumn{5}{c|}{$0$} & \multicolumn{5}{c|}{$-5$} & \multicolumn{5}{c||}{$-10$} & Real data \\
\midrule
$T_{60}$ & 0 & 0.1 & 0.2 & 0.4 & 0.6 & 0 & 0.1 & 0.2 & 0.4 & 0.6 & 0 & 0.1 & 0.2 & 0.4 & 0.6 & 0.5 \\
\midrule
 \multirow{3}{*}{\textit{b\&b}}  & \mb{82.1\%}  & \mb{82.8\%}  & \mb{73.8\%}  & \mb{48.3\%}  & \mb{35.7\%} &
\mb{84.1\%} 
& \mb{82.7\%} & \mb{68.6\%} & \mb{41.1\%}  & \mb{29.8\%}  & \mb{77.5\%}  & \mb{66.6\%}  & \mb{44.5\%}  & \mb{24.8\%}  &
\mb{19.0\%} & \mb{27.5\%} \\
		  & 9.59  & 10.49  & 12.65  & 14.99  & 16.10  & 10.46  & 11.58  & 13.91  & 16.07  & 16.97  & 13.45  &
14.35  & 16.53  & 18.29  & 18.63 & 16.04\\
		  & \mb{3.66}  & \mb{4.47}  & \mb{6.14}  & \mb{7.09}  & 7.30  & \mb{4.64}  & \mb{5.45}  & \mb{6.75}  &
7.44 & 7.35 & \mb{6.56} & \mb{6.85} & 7.36  & 7.29  & 7.26 & 7.55\\
\midrule\midrule
 \multirow{3}{*}{\textit{unc}}  & 38.3\%  & 36.2\%  & 33.3\%  & 23.3\%  & 19.2\%  & 37.5\%  & 37.0\%  & 31.5\%  &
21.4\%  & 16.7\%  & 33.4\%  & 29.2\%  & 21.6\%  & 14.3\%  & 11.6\% & 14.1\% \\
		  & 15.89  & 16.15  & 17.01  & 17.94  & 18.60  & 16.76  & 16.92  & 17.71  & 18.48  & 18.74  & 17.28  &
17.75  & 18.67  & 18.95  & 19.54 & 18.93\\
		  & 7.47  & 7.30  & 7.46  & 7.30  & 7.69  & 7.38  & 7.36  & 7.41  & 7.30  & \mb{7.26}  & 7.51  & 7.62  &
7.34  & \mb{7.21}  & 7.03 & 7.13\\
\midrule
 \multirow{3}{*}{\textit{d-lb}}  & 75.3\%  & 75.4\%  & 67.5\%  & 44.6\%  & 33.4\%  & 80.4\%  & 77.9\%  & 61.9\%  &
36.8\%  & 28.0\%  & 66.6\%  & 56.5\%  & 36.3\%  & 20.6\%  & 15.8\% & 22.3\%\\
		  & 10.54  & 11.55  & 13.54  & 15.53  & 16.25  & 11.74  & 12.99  & 14.74  & 16.54  & 17.11  & 14.69  &
16.01  & 17.27  & 18.28  & 18.61 & 17.51\\
		  & 4.57  & 5.26  & 6.51  & 7.21  & 7.22  & 5.52  & 6.35  & 6.91  & 7.63  & 7.38  & 6.90  & 7.19  & 7.37
 & 7.30  & 7.20 & 7.53\\
\midrule
 \multirow{3}{*}{\textit{s-lb}}  & 46.9\%  & 46.7\%  & 40.9\%  & 27.9\%  & 22.2\%  & 39.3\%  & 40.8\%  & 34.4\%  &
23.6\%  & 18.7\%  & 31.0\%  & 29.7\%  & 22.1\%  & 14.5\%  & 13.2\% & 13.2\%\\
		  & 11.63  & 12.58  & 14.79  & 17.05  & 17.67  & 13.41  & 14.58  & 16.60  & 17.76  & 18.19  & 17.13  &
17.85  & 18.46  & 19.17  & 19.65 & 18.80\\
		  & 5.54  & 6.17  & 6.98  & 7.33  & 7.13  & 6.41  & 6.74  & 7.21  & 7.33  & 7.28  & 7.36  & 7.31  &
\mb{7.00} & 7.41  & 7.26 & 7.09\\
\midrule
 \multirow{3}{*}{\textit{dm}}  & 80.3\%  & 77.4\%  & 62.4\%  & 41.2\%  & 30.2\%  & 78.3\%  & 74.0\%  & 57.7\%  &
34.8\%  & 26.9\%  & 60.4\%  & 51.3\%  & 35.6\%  & 21.4\%  & 16.3\% & 16.7\%\\
		  & 15.75  & 15.94  & 16.49  & 17.04  & 17.76  & 15.80  & 16.09  & 16.48  & 17.53  & 17.82  & 16.65  &
16.90  & 17.86  & 18.76  & 19.03 & 19.34\\
		  & 7.11  & 7.18  & 7.38  & 7.50  & 7.48  & 7.12  & 7.15  & 7.35  & 7.55  & 7.45  & 7.30  & 7.31  & 7.33
 & 7.29  & 7.34 & 7.00\\
\midrule\midrule
 \multirow{3}{*}{\textit{n-mult}}  & 60.8\%  & 54.6\%  & 40.5\%  & 22.6\%  & 16.0\%  & 49.7\%  & 46.3\%  & 35.1\% 
& 18.6\%  & 13.7\%  & 41.9\%  & 36.0\%  & 23.0\%  & 12.8\%  & 10.1\% & 13.2\%\\
		  & 7.99  & 9.00  & 11.18  & \mb{14.11}  & 15.46  & \mb{9.23}  & \mb{10.56}  & 13.18  & \mb{15.82}  &
16.79 & \mb{13.11}  & \mb{14.28}  & 16.26  & \mb{17.84}  & 18.23 & 17.54\\
		  & 5.45  & 6.23  & 7.43  & 8.11  & 7.97  & 6.30  & 6.98  & 7.51  & 7.93  & 7.51  & 7.32  & 7.46  & 7.68
 & 7.40  & \mb{7.01} & 7.27\\
\midrule
 \multirow{3}{*}{\textit{t-mult}}  & 61.0\%  & 53.1\%  & 39.4\%  & 22.0\%  & 15.5\%  & 50.0\%  & 45.5\%  & 34.1\% 
& 18.3\%  & 13.6\%  & 41.3\%  & 35.1\%  & 22.4\%  & 12.3\%  & 9.9\% &17.38\% \\
		  & \textbf{7.81}  & \mb{8.75}  & \mb{11.06}  & 14.14  & \mb{15.41}  & 9.42  & 10.72  & \mb{13.09}  &
16.02 & \mb{16.52} & 13.65 & 14.50  & \mb{16.22}  & 18.15  & \mb{18.17} & 16.1\\
		  & 5.09  & 6.03  & 7.22  & 8.10  & 7.74  & 6.14  & 7.00  & 7.43  & 7.88  & 7.36  & 7.34  & 7.38  & 7.39
 & 7.45  & 7.22 & 7.80\\
\midrule
 \multirow{3}{*}{\textit{f-mult}}  & 59.6\%  & 52.5\%  & 38.6\%  & 21.7\%  & 15.0\%  & 48.8\%  & 44.5\%  & 34.0\% 
& 17.9\%  & 13.5\%  & 40.6\%  & 34.7\%  & 21.9\%  & 12.0\%  & 9.8\% & 17.91\%\\
		  & 8.38  & 9.31  & 11.51  & 14.27  & 15.69  & 9.92  & 11.22  & 13.54  & 16.21  & 16.92  & 13.87  &
14.75  & 16.34  & 17.87  & 18.31 & \mb{15.3}\\
		  & 5.35  & 6.16  & 7.21  & 7.90  & 7.66  & 6.24  & 7.03  & 7.33  & 7.70  & 7.37  & 7.27  & 7.41  & 7.36
 & 7.41  & 7.21 & 7.78\\
\midrule
 \multirow{3}{*}{\textit{pi}}  & 53.7\%  & 53.3\%  & 44.3\%  & 30.3\%  & 23.6\%  & 41.4\%  & 41.5\%  & 32.7\%  &
21.9\%  & 16.9\%  & 29.6\%  & 28.9\%  & 20.8\%  & 14.7\%  & 12.5\% & 12.9\%\\
		  & 11.31  & 12.47  & 14.60  & 16.81  & 17.67  & 13.24  & 14.23  & 16.50  & 18.12  & 18.08  & 17.04  &
17.76  & 18.92  & 18.98  & 19.25 & 19.03\\
		  & 5.55  & 6.17  & 6.92  & 7.33  & 7.60  & 6.11  & 6.71  & 7.09  & \mb{7.15}  & 7.43  & 7.19  & 7.17  &
7.49  & 7.58  & 7.22 & \mb{6.88}\\
\bottomrule
\end{tabular}}
\else
{\tiny
\begin{tabular}{c|ccccc|ccccc|ccccc||c}
\toprule
SNR& \multicolumn{5}{c|}{$0$} & \multicolumn{5}{c|}{$-5$} & \multicolumn{5}{c||}{$-10$} & Real data \\
\midrule
$T_{60}$ & 0 & 0.1 & 0.2 & 0.4 & 0.6 & 0 & 0.1 & 0.2 & 0.4 & 0.6 & 0 & 0.1 & 0.2 & 0.4 & 0.6 & 0.5 \\
\midrule
 \multirow{3}{*}{\textit{b\&b}}  & \mb{82.1\%}  & \mb{82.8\%}  & \mb{73.8\%}  & \mb{48.3\%}  & \mb{35.7\%} &
\mb{84.1\%} 
& \mb{82.7\%} & \mb{68.6\%} & \mb{41.1\%}  & \mb{29.8\%}  & \mb{77.5\%}  & \mb{66.6\%}  & \mb{44.5\%}  & \mb{24.8\%}  &
\mb{19.0\%} & \mb{27.5\%} \\
		  & 9.59  & 10.49  & 12.65  & 14.99  & 16.10  & 10.46  & 11.58  & 13.91  & 16.07  & 16.97  & 13.45  &
14.35  & 16.53  & 18.29  & 18.63 & 16.04\\
		  & \mb{3.66}  & \mb{4.47}  & \mb{6.14}  & \mb{7.09}  & 7.30  & \mb{4.64}  & \mb{5.45}  & \mb{6.75}  &
7.44 & 7.35 & \mb{6.56} & \mb{6.85} & 7.36  & 7.29  & 7.26 & 7.55\\
\midrule\midrule
 \multirow{3}{*}{\textit{unc}}  & 38.3\%  & 36.2\%  & 33.3\%  & 23.3\%  & 19.2\%  & 37.5\%  & 37.0\%  & 31.5\%  &
21.4\%  & 16.7\%  & 33.4\%  & 29.2\%  & 21.6\%  & 14.3\%  & 11.6\% & 14.1\% \\
		  & 15.89  & 16.15  & 17.01  & 17.94  & 18.60  & 16.76  & 16.92  & 17.71  & 18.48  & 18.74  & 17.28  &
17.75  & 18.67  & 18.95  & 19.54 & 18.93\\
		  & 7.47  & 7.30  & 7.46  & 7.30  & 7.69  & 7.38  & 7.36  & 7.41  & 7.30  & \mb{7.26}  & 7.51  & 7.62  &
7.34  & \mb{7.21}  & 7.03 & 7.13\\
\midrule
 \multirow{3}{*}{\textit{d-lb}}  & 75.3\%  & 75.4\%  & 67.5\%  & 44.6\%  & 33.4\%  & 80.4\%  & 77.9\%  & 61.9\%  &
36.8\%  & 28.0\%  & 66.6\%  & 56.5\%  & 36.3\%  & 20.6\%  & 15.8\% & 22.3\%\\
		  & 10.54  & 11.55  & 13.54  & 15.53  & 16.25  & 11.74  & 12.99  & 14.74  & 16.54  & 17.11  & 14.69  &
16.01  & 17.27  & 18.28  & 18.61 & 17.51\\
		  & 4.57  & 5.26  & 6.51  & 7.21  & 7.22  & 5.52  & 6.35  & 6.91  & 7.63  & 7.38  & 6.90  & 7.19  & 7.37
 & 7.30  & 7.20 & 7.53\\
\midrule
 \multirow{3}{*}{\textit{s-lb}}  & 46.9\%  & 46.7\%  & 40.9\%  & 27.9\%  & 22.2\%  & 39.3\%  & 40.8\%  & 34.4\%  &
23.6\%  & 18.7\%  & 31.0\%  & 29.7\%  & 22.1\%  & 14.5\%  & 13.2\% & 13.2\%\\
		  & 11.63  & 12.58  & 14.79  & 17.05  & 17.67  & 13.41  & 14.58  & 16.60  & 17.76  & 18.19  & 17.13  &
17.85  & 18.46  & 19.17  & 19.65 & 18.80\\
		  & 5.54  & 6.17  & 6.98  & 7.33  & 7.13  & 6.41  & 6.74  & 7.21  & 7.33  & 7.28  & 7.36  & 7.31  &
\mb{7.00} & 7.41  & 7.26 & 7.09\\
\midrule
 \multirow{3}{*}{\textit{dm}}  & 80.3\%  & 77.4\%  & 62.4\%  & 41.2\%  & 30.2\%  & 78.3\%  & 74.0\%  & 57.7\%  &
34.8\%  & 26.9\%  & 60.4\%  & 51.3\%  & 35.6\%  & 21.4\%  & 16.3\% & 16.7\%\\
		  & 15.75  & 15.94  & 16.49  & 17.04  & 17.76  & 15.80  & 16.09  & 16.48  & 17.53  & 17.82  & 16.65  &
16.90  & 17.86  & 18.76  & 19.03 & 19.34\\
		  & 7.11  & 7.18  & 7.38  & 7.50  & 7.48  & 7.12  & 7.15  & 7.35  & 7.55  & 7.45  & 7.30  & 7.31  & 7.33
 & 7.29  & 7.34 & 7.00\\
\midrule\midrule
 \multirow{3}{*}{\textit{n-mult}}  & 60.8\%  & 54.6\%  & 40.5\%  & 22.6\%  & 16.0\%  & 49.7\%  & 46.3\%  & 35.1\% 
& 18.6\%  & 13.7\%  & 41.9\%  & 36.0\%  & 23.0\%  & 12.8\%  & 10.1\% & 13.2\%\\
		  & 7.99  & 9.00  & 11.18  & \mb{14.11}  & 15.46  & \mb{9.23}  & \mb{10.56}  & 13.18  & \mb{15.82}  &
16.79 & \mb{13.11}  & \mb{14.28}  & 16.26  & \mb{17.84}  & 18.23 & 17.54\\
		  & 5.45  & 6.23  & 7.43  & 8.11  & 7.97  & 6.30  & 6.98  & 7.51  & 7.93  & 7.51  & 7.32  & 7.46  & 7.68
 & 7.40  & \mb{7.01} & 7.27\\
\midrule
 \multirow{3}{*}{\textit{t-mult}}  & 61.0\%  & 53.1\%  & 39.4\%  & 22.0\%  & 15.5\%  & 50.0\%  & 45.5\%  & 34.1\% 
& 18.3\%  & 13.6\%  & 41.3\%  & 35.1\%  & 22.4\%  & 12.3\%  & 9.9\% &17.38\% \\
		  & \textbf{7.81}  & \mb{8.75}  & \mb{11.06}  & 14.14  & \mb{15.41}  & 9.42  & 10.72  & \mb{13.09}  &
16.02 & \mb{16.52} & 13.65 & 14.50  & \mb{16.22}  & 18.15  & \mb{18.17} & 16.1\\
		  & 5.09  & 6.03  & 7.22  & 8.10  & 7.74  & 6.14  & 7.00  & 7.43  & 7.88  & 7.36  & 7.34  & 7.38  & 7.39
 & 7.45  & 7.22 & 7.80\\
\midrule
 \multirow{3}{*}{\textit{f-mult}}  & 59.6\%  & 52.5\%  & 38.6\%  & 21.7\%  & 15.0\%  & 48.8\%  & 44.5\%  & 34.0\% 
& 17.9\%  & 13.5\%  & 40.6\%  & 34.7\%  & 21.9\%  & 12.0\%  & 9.8\% & 17.91\%\\
		  & 8.38  & 9.31  & 11.51  & 14.27  & 15.69  & 9.92  & 11.22  & 13.54  & 16.21  & 16.92  & 13.87  &
14.75  & 16.34  & 17.87  & 18.31 & \mb{15.3}\\
		  & 5.35  & 6.16  & 7.21  & 7.90  & 7.66  & 6.24  & 7.03  & 7.33  & 7.70  & 7.37  & 7.27  & 7.41  & 7.36
 & 7.41  & 7.21 & 7.78\\
\midrule
 \multirow{3}{*}{\textit{pi}}  & 53.7\%  & 53.3\%  & 44.3\%  & 30.3\%  & 23.6\%  & 41.4\%  & 41.5\%  & 32.7\%  &
21.9\%  & 16.9\%  & 29.6\%  & 28.9\%  & 20.8\%  & 14.7\%  & 12.5\% & 12.9\%\\
		  & 11.31  & 12.47  & 14.60  & 16.81  & 17.67  & 13.24  & 14.23  & 16.50  & 18.12  & 18.08  & 17.04  &
17.76  & 18.92  & 18.98  & 19.25 & 19.03\\
		  & 5.55  & 6.17  & 6.92  & 7.33  & 7.60  & 6.11  & 6.71  & 7.09  & \mb{7.15}  & 7.43  & 7.19  & 7.17  &
7.49  & 7.58  & 7.22 & \mb{6.88}\\
\bottomrule
\end{tabular}}
\fi
\end{table*}

\reftab{tab:results} summarizes the results obtained with the evaluated methods on simulated as well as on real data. While
 columns 2 to 16 correspond to simulated data, the last column corresponds to real data. In the simulated case,
the first row displays the $SNR$ in dB and the second row displays $T_{60}$ in seconds. The next rows display the
performance of the evaluated methods, and are split into three groups by double lines: the proposed \textit{b\&b}
method (top), state-of-the-art multichannel TDE methods (middle), and state-of-the-art multilateration methods (bottom).
For each SNR--$T_{60}$ combination and for each method we provide three numbers: (i) the percentage of inliers
(angular error $<30^\circ$), (ii) the average and  (iii) the standard deviation of the angular error over the inliers. These
quantities are computed on $100$~ms long signals received from the $189$ different sound source positions, and the best
values are shown in bold. On an average, each entry of the table roughly corresponds to $3,000$ localization trials.
Overall, we performed more than $400,000$ localizations.

\begin{table*}[t]
 \caption{This table displays the histograms associated with the localization error, organized in the same way as \reftab{tab:results}.
 The histogram abscissae start at $0^\circ$ (no error) and span to $180^\circ$ (maximum error).}
  \label{tab:histograms}
\centering
{\scriptsize
\begin{tabular}{c|ccccc|ccccc|ccccc||c}
\toprule
SNR& \multicolumn{5}{c|}{$0$} & \multicolumn{5}{c|}{$-5$} & \multicolumn{5}{c||}{$-10$} & Real data \\
\midrule
$T_{60}$ & 0 & 0.1 & 0.2 & 0.4 & 0.6 & 0 & 0.1 & 0.2 & 0.4 & 0.6 & 0 & 0.1 & 0.2 & 0.4 & 0.6 & 0.5 \\
\midrule
 \textit{b\&b} &\ih{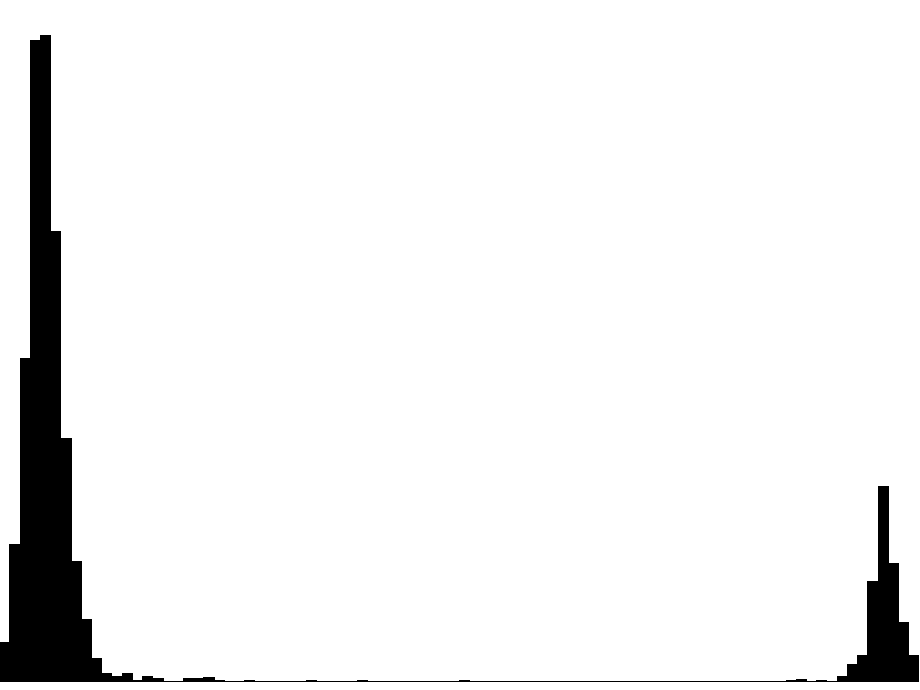}& \ih{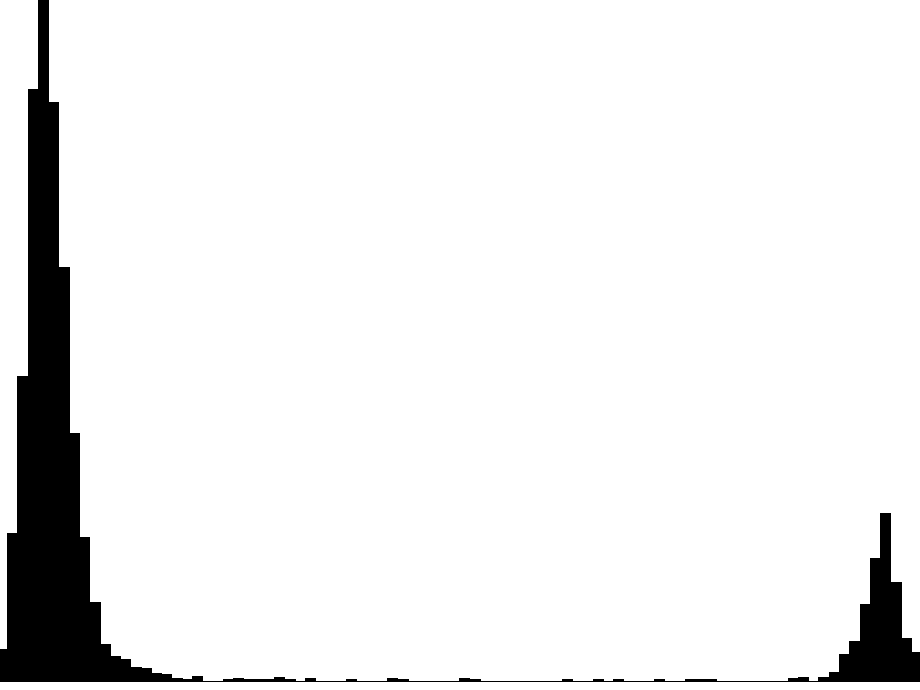} & \ih{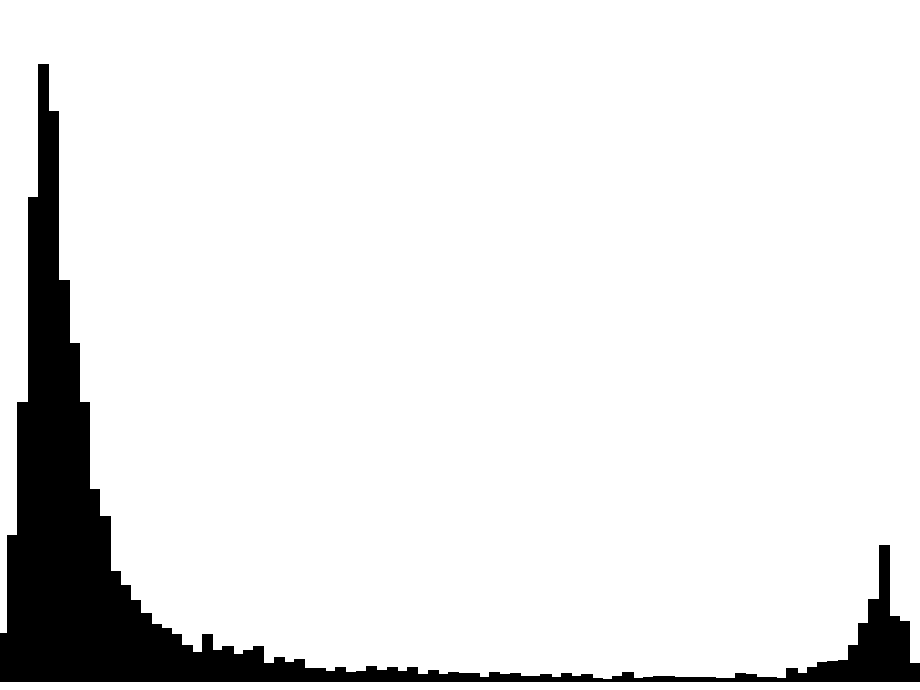} & \ih{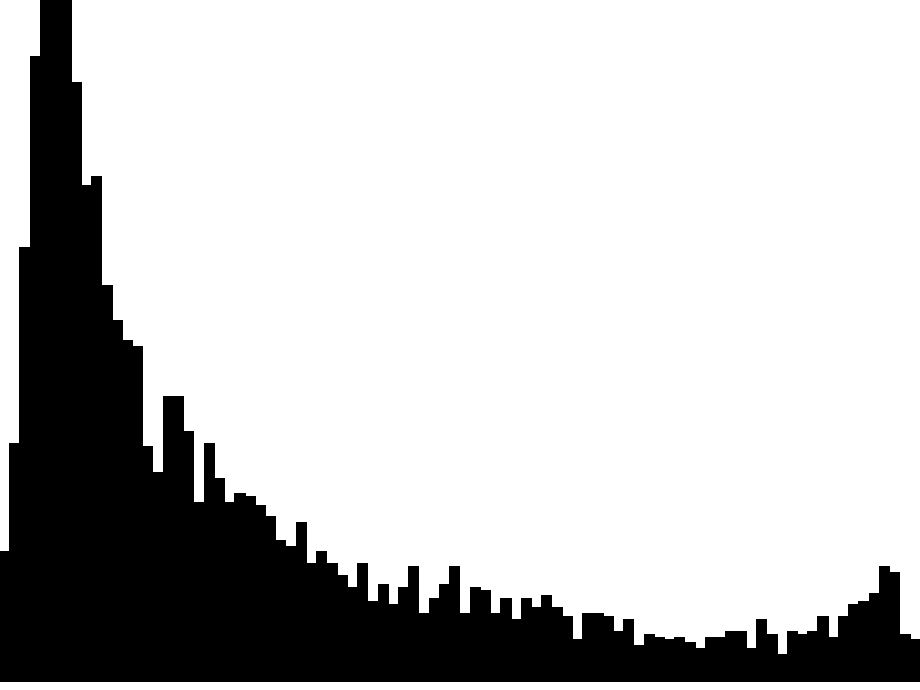} & \ih{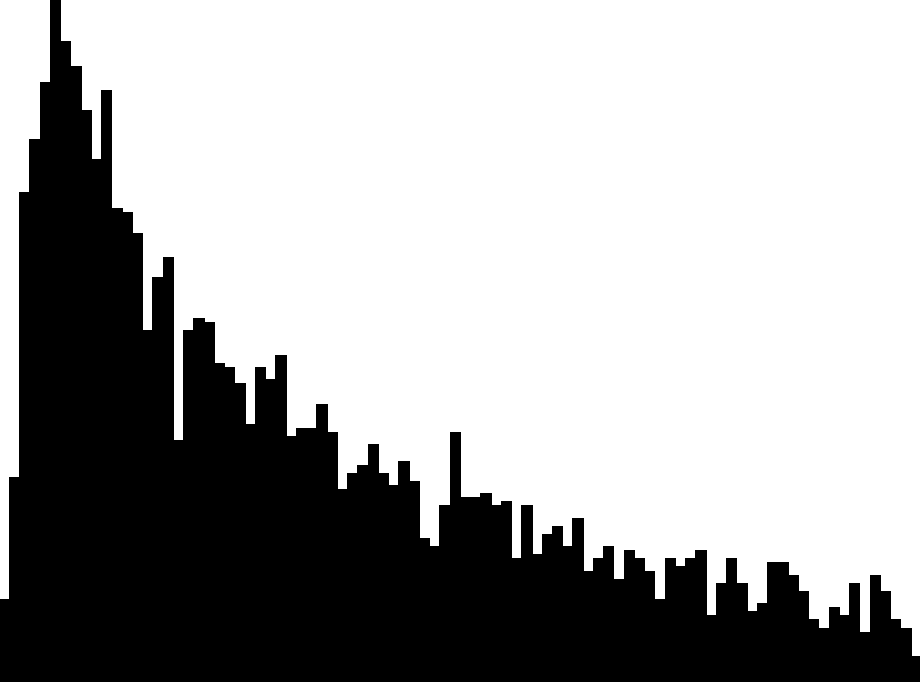} &
\ih{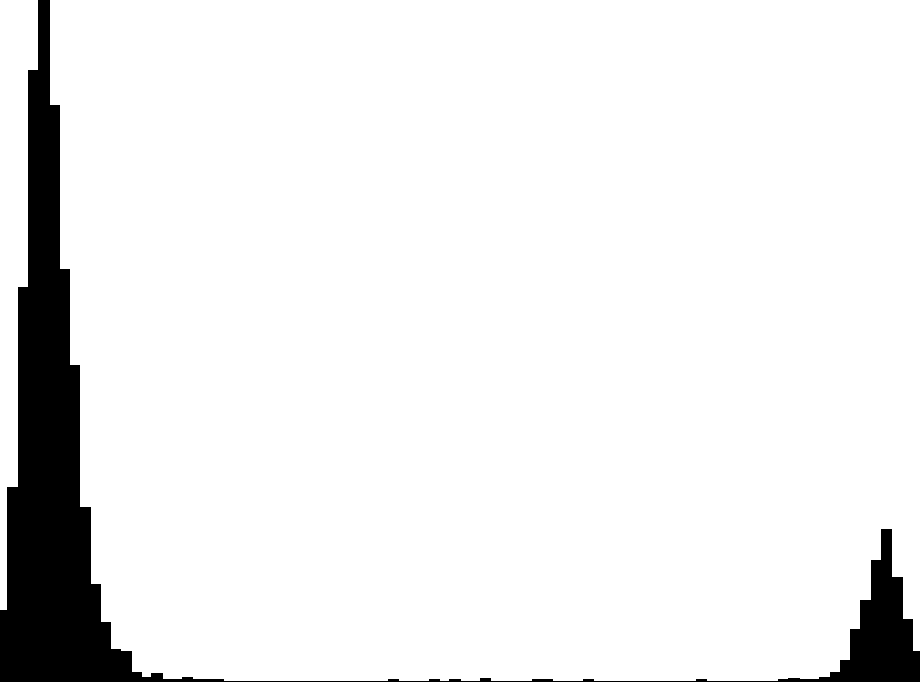} & \ih{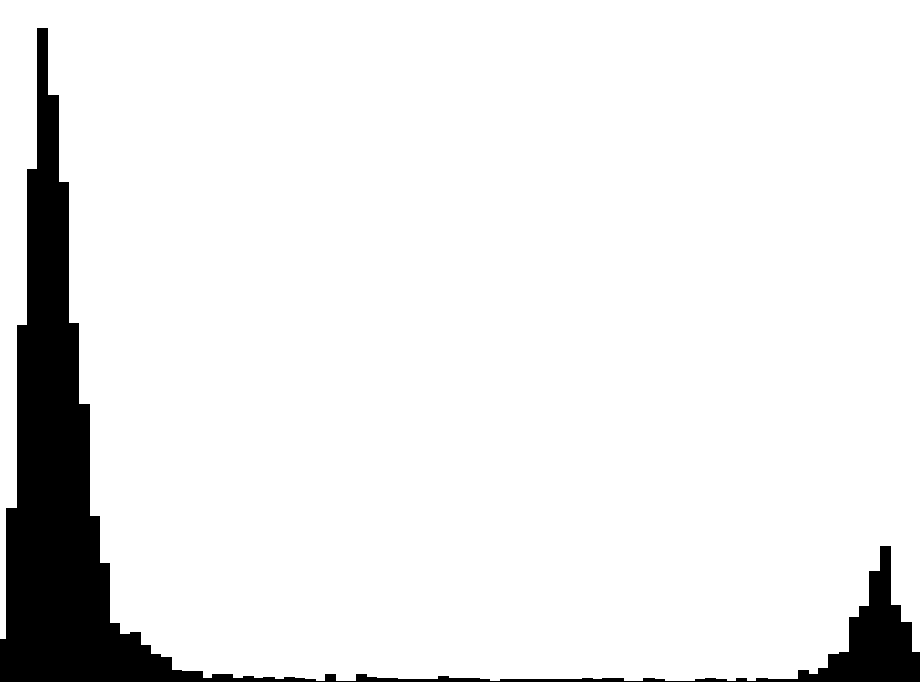} & \ih{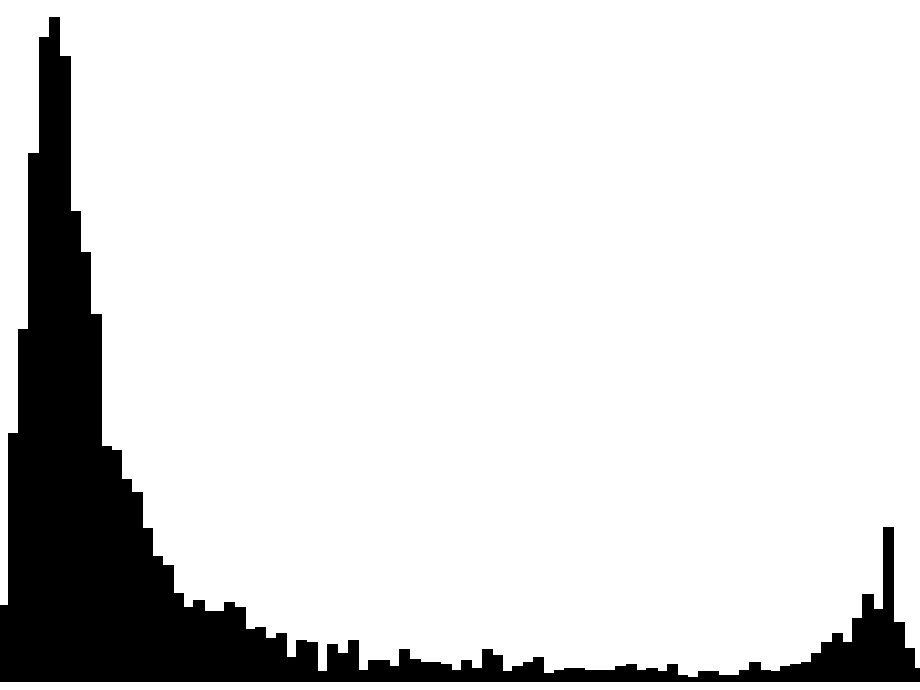} & \ih{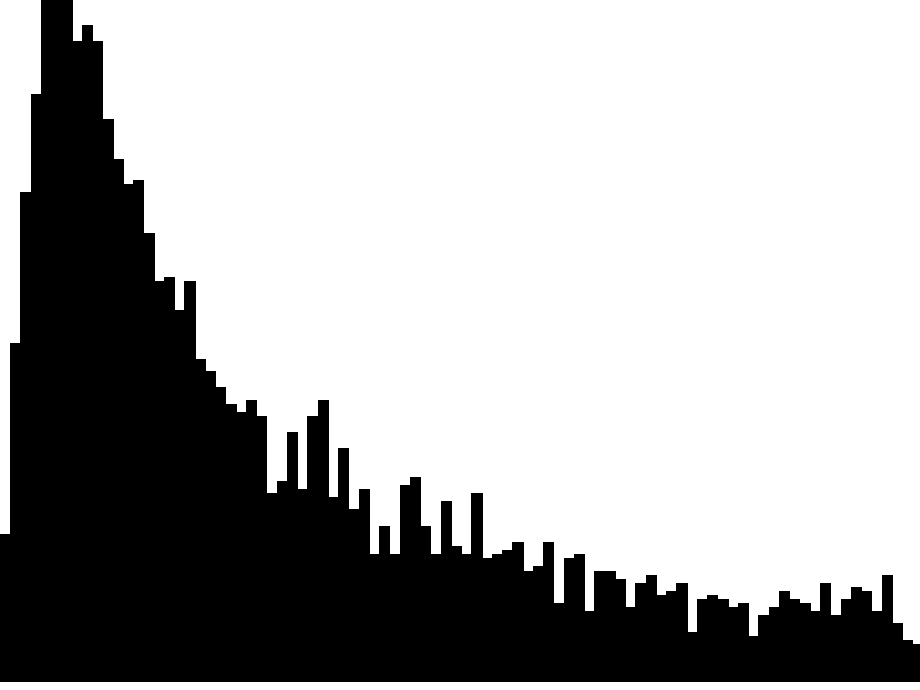} & \ih{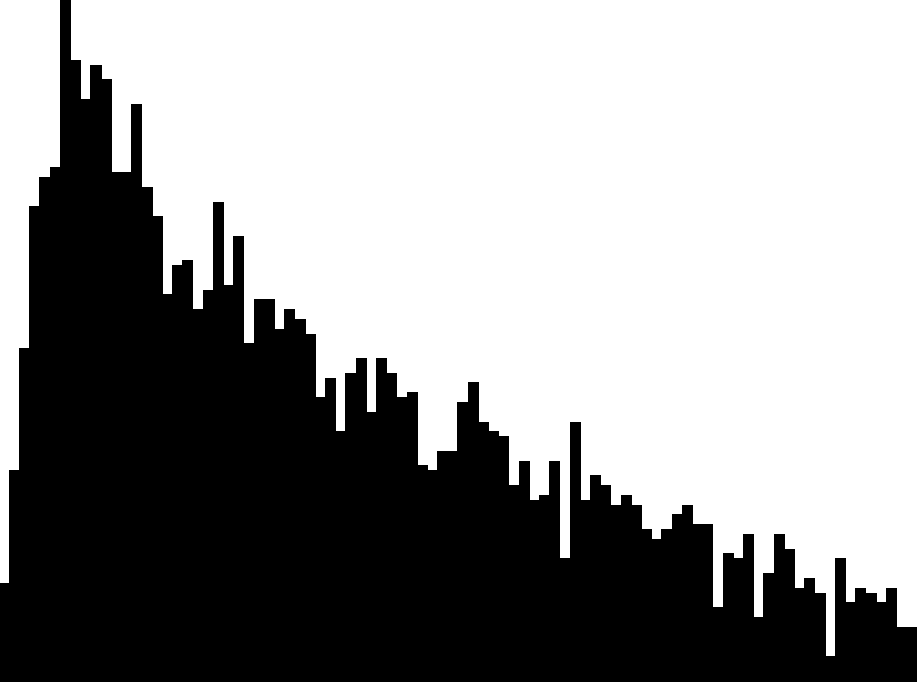} & \ih{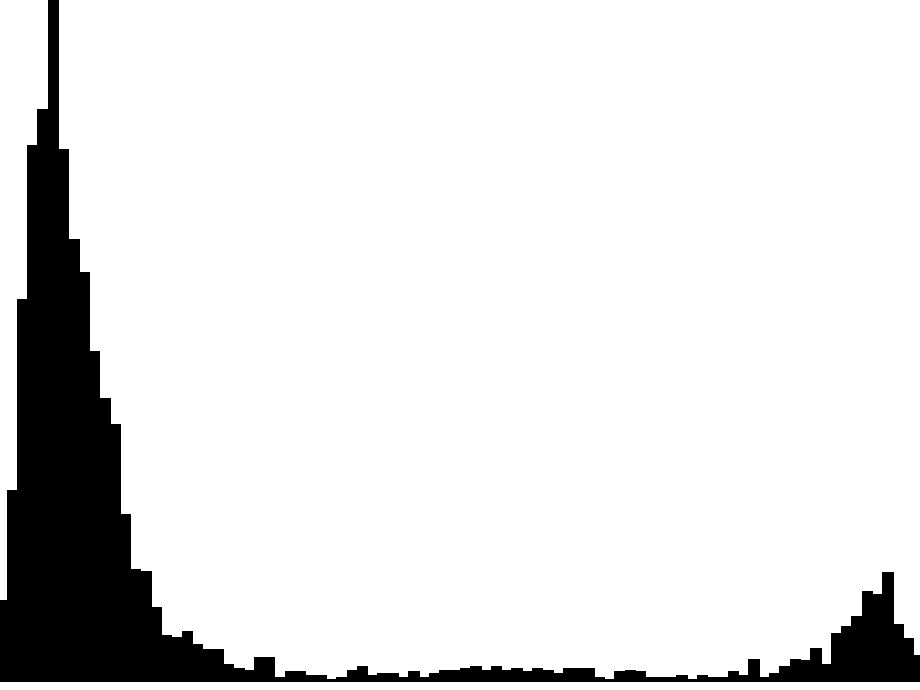} &
\ih{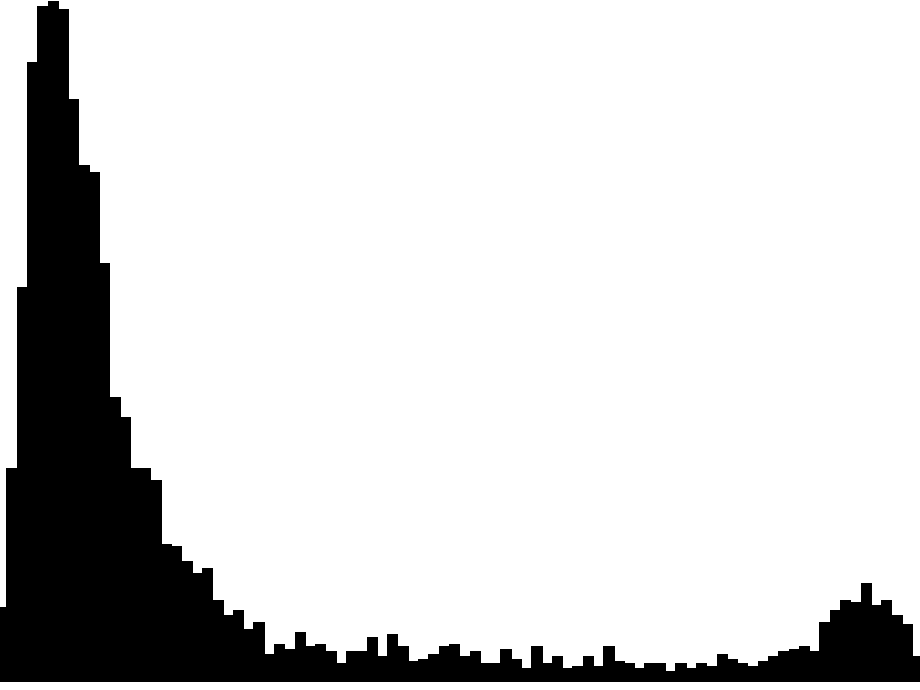} & \ih{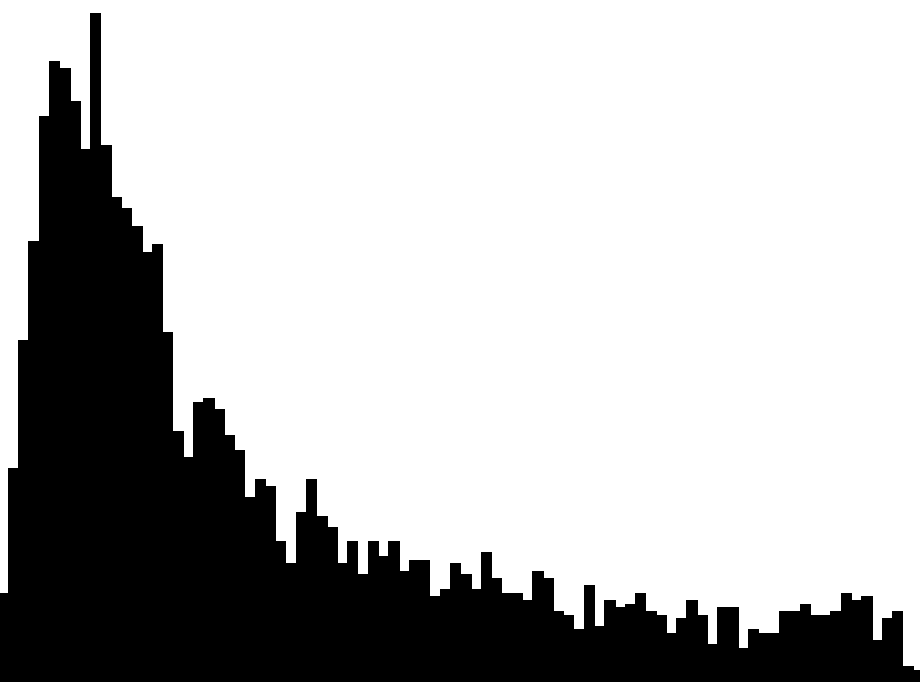} & \ih{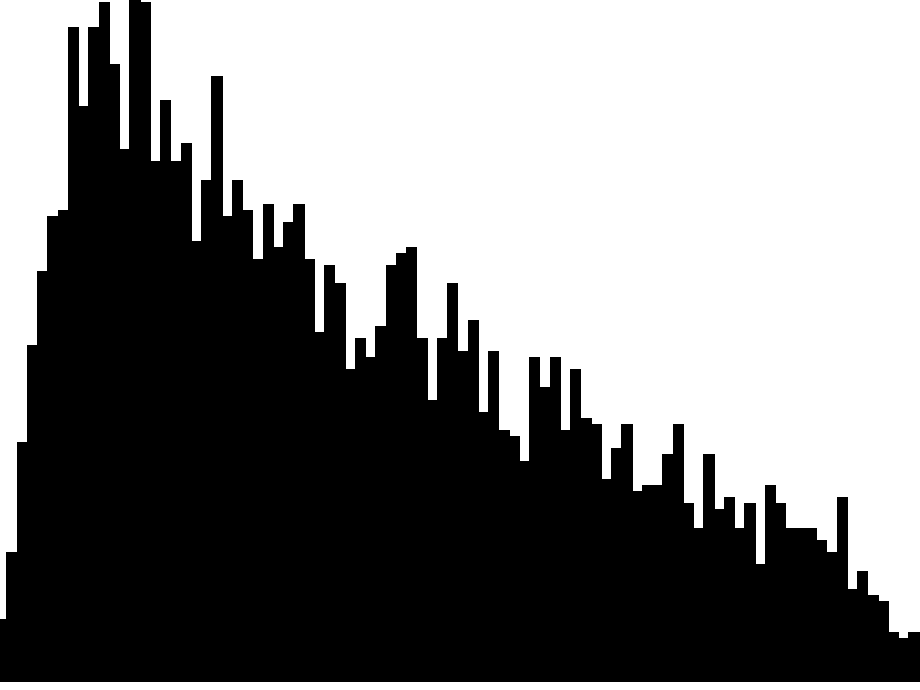} & \ih{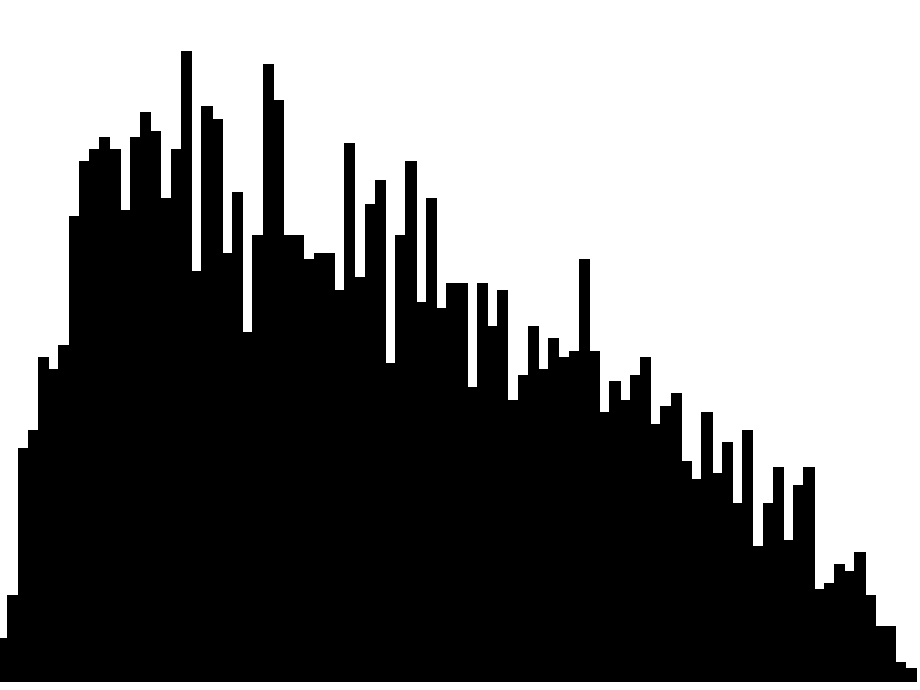} & \ih{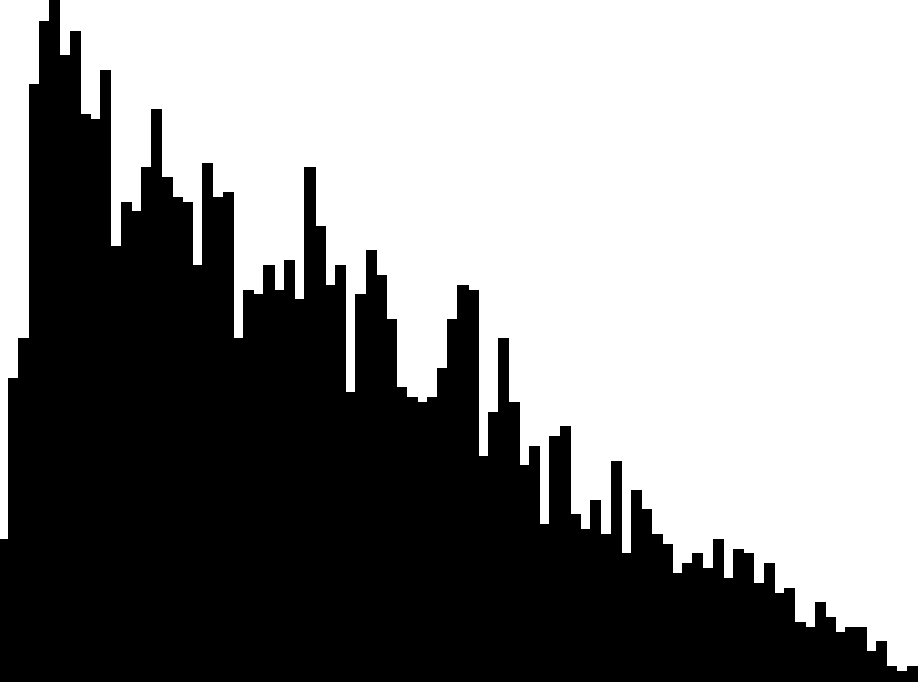} \\
\midrule\midrule
 \textit{unc} &\ih{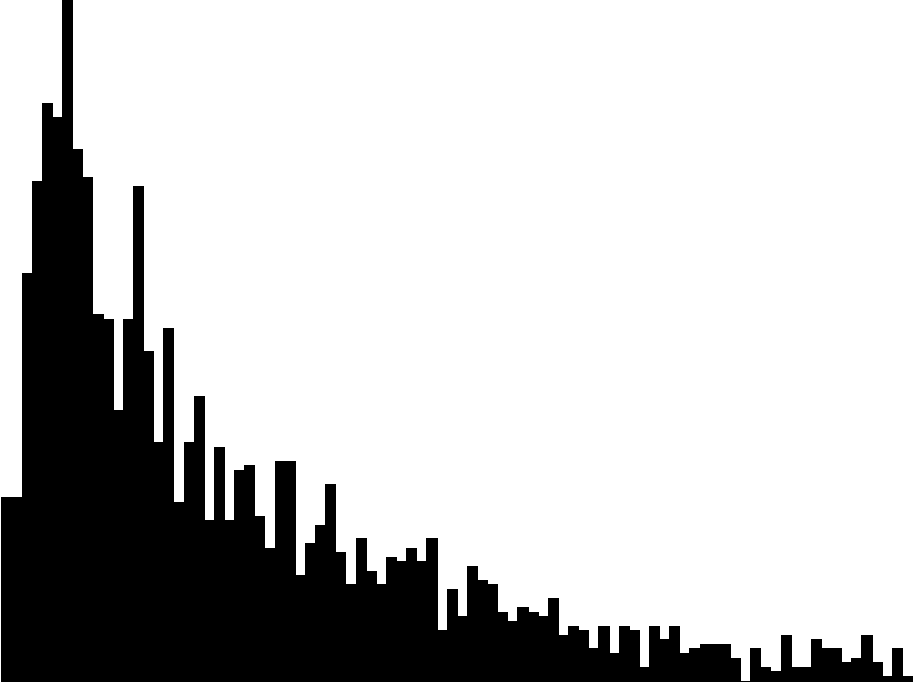}& \ih{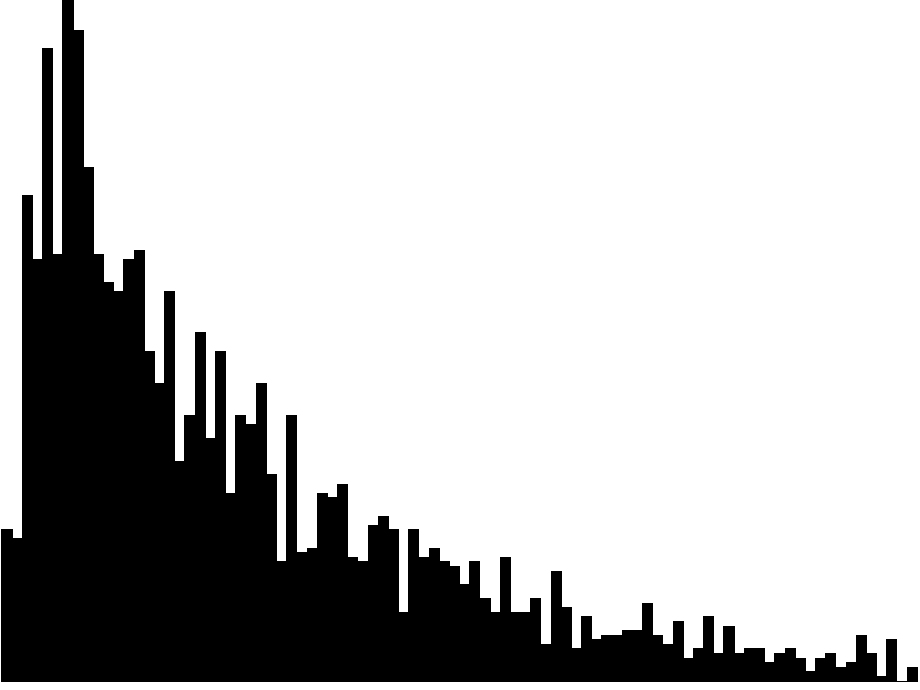} & \ih{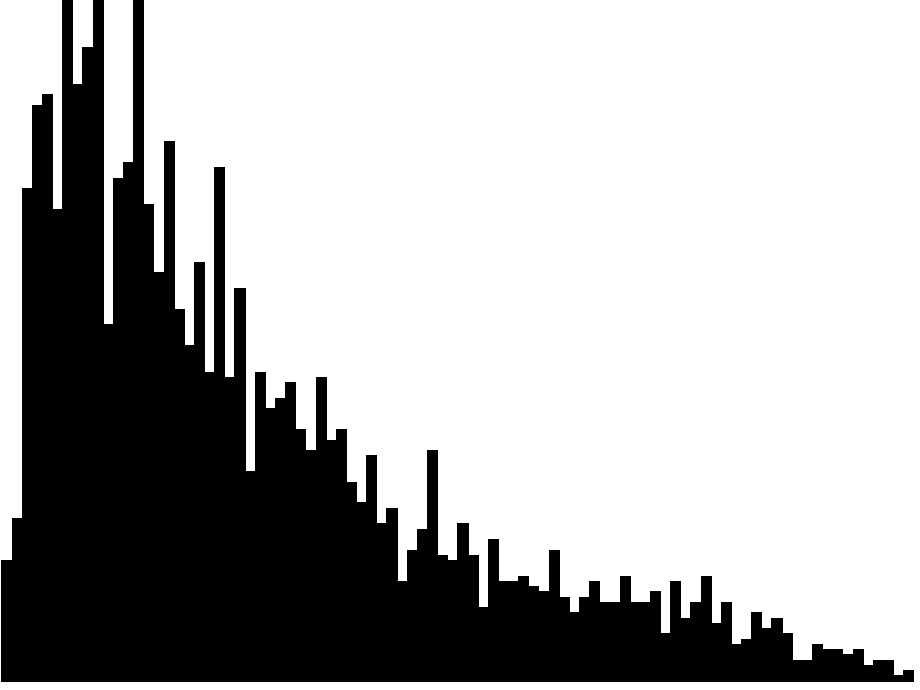} & \ih{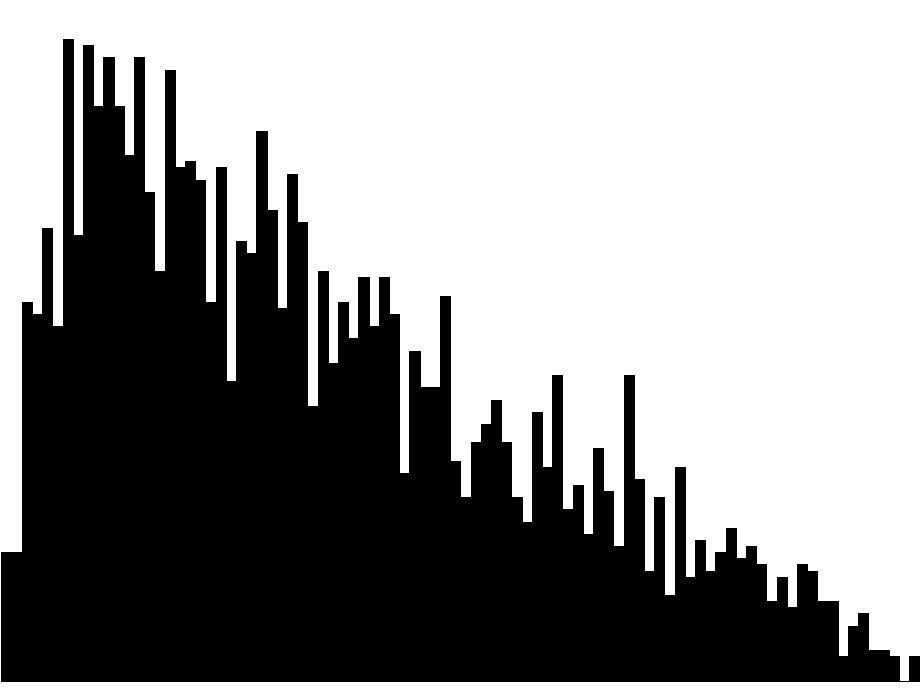} & \ih{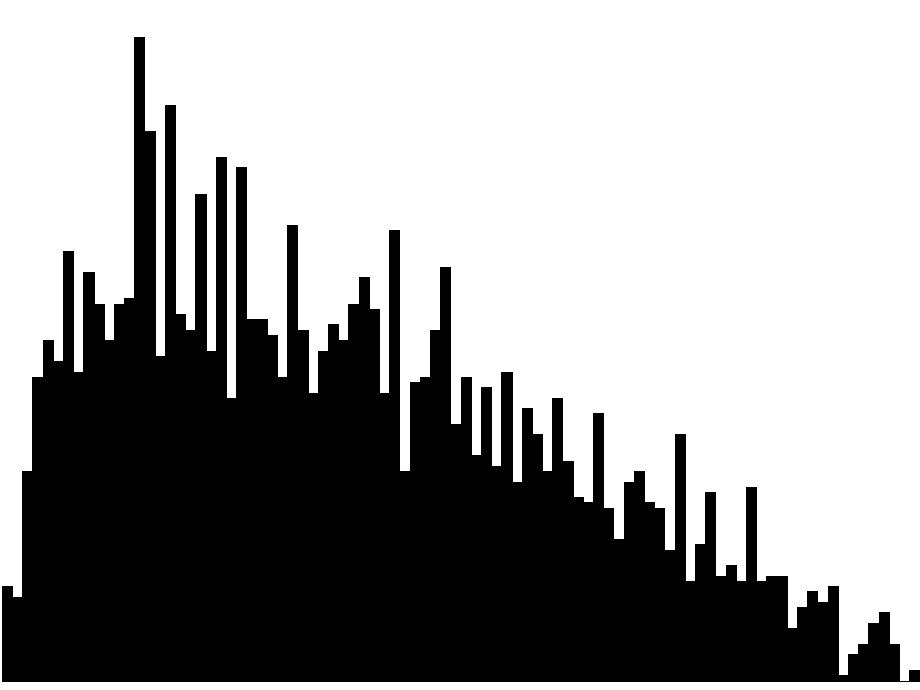} &
\ih{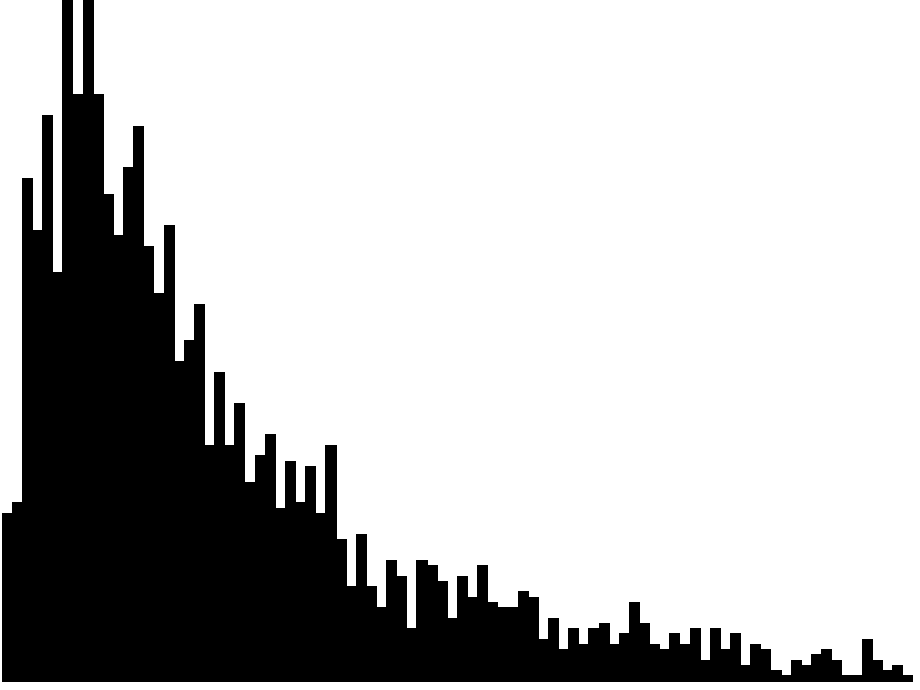} & \ih{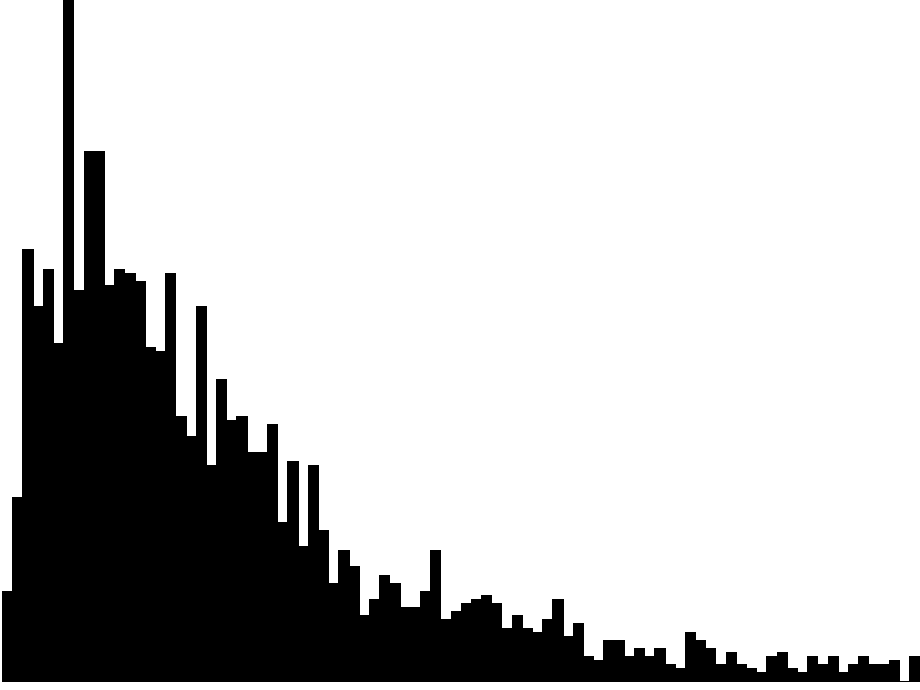} & \ih{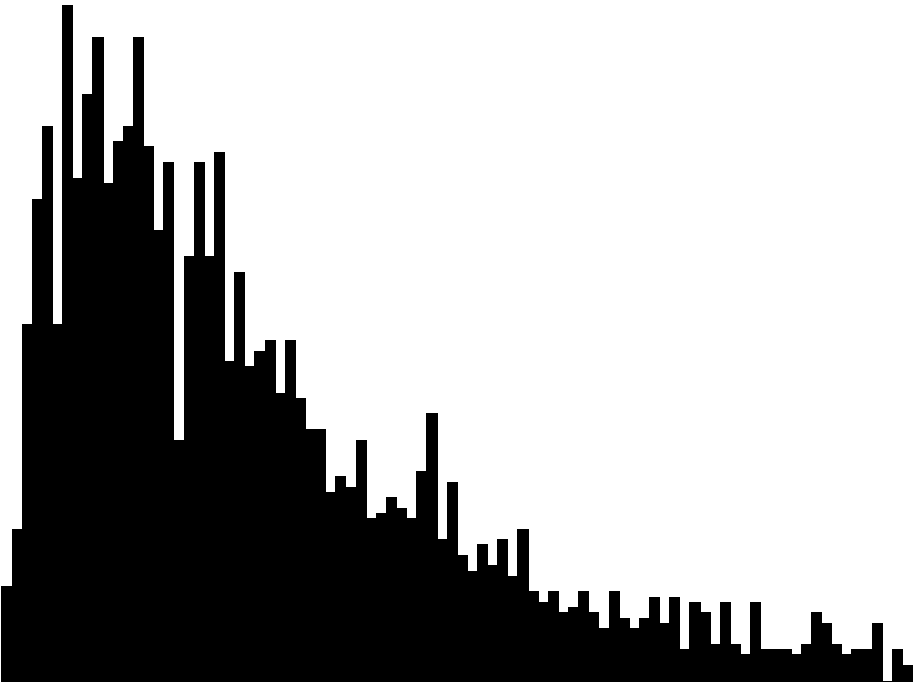} & \ih{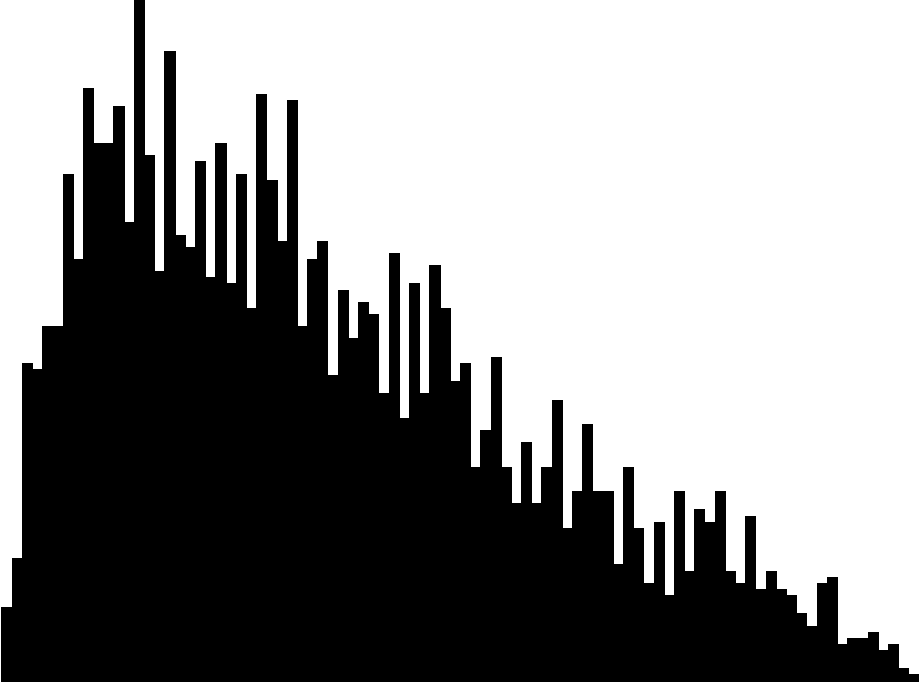} & \ih{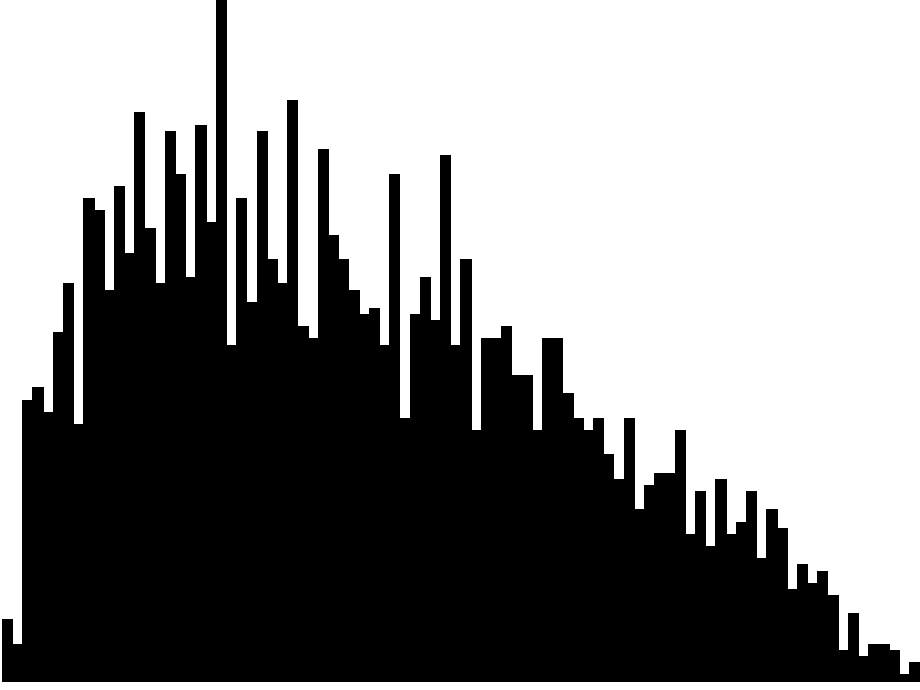} & \ih{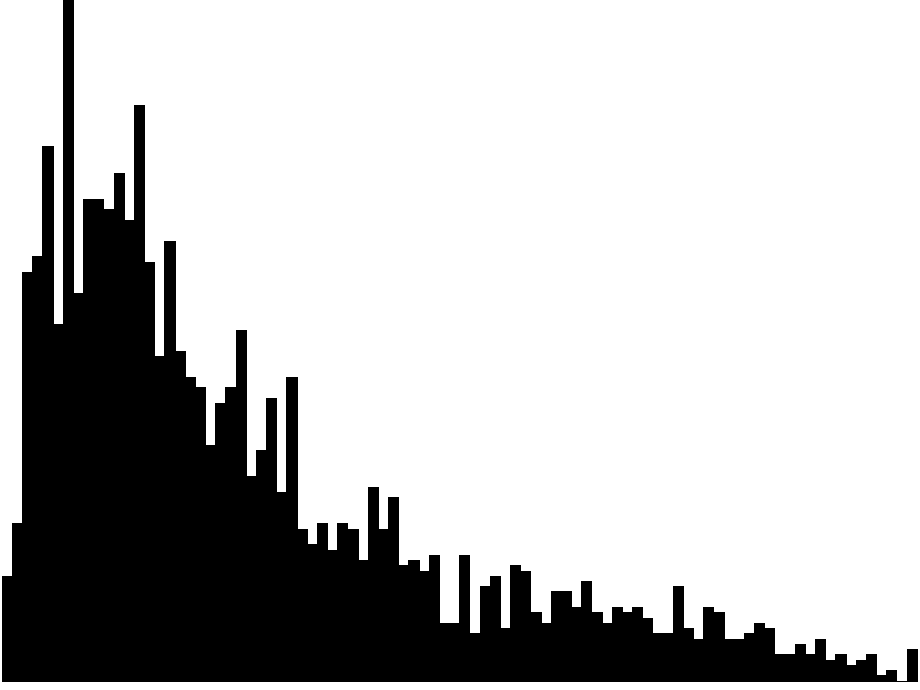} &
\ih{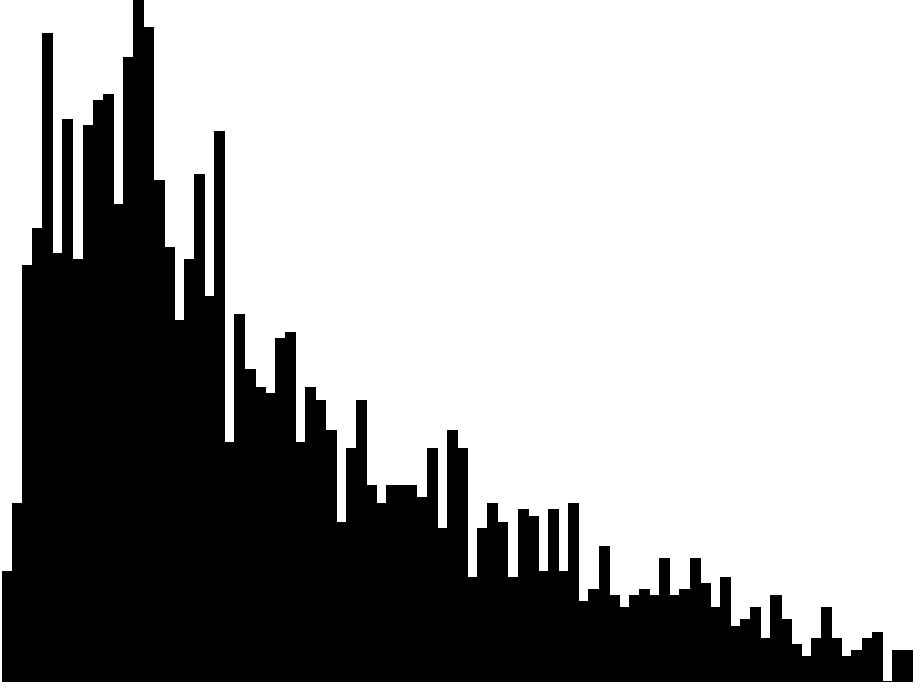} & \ih{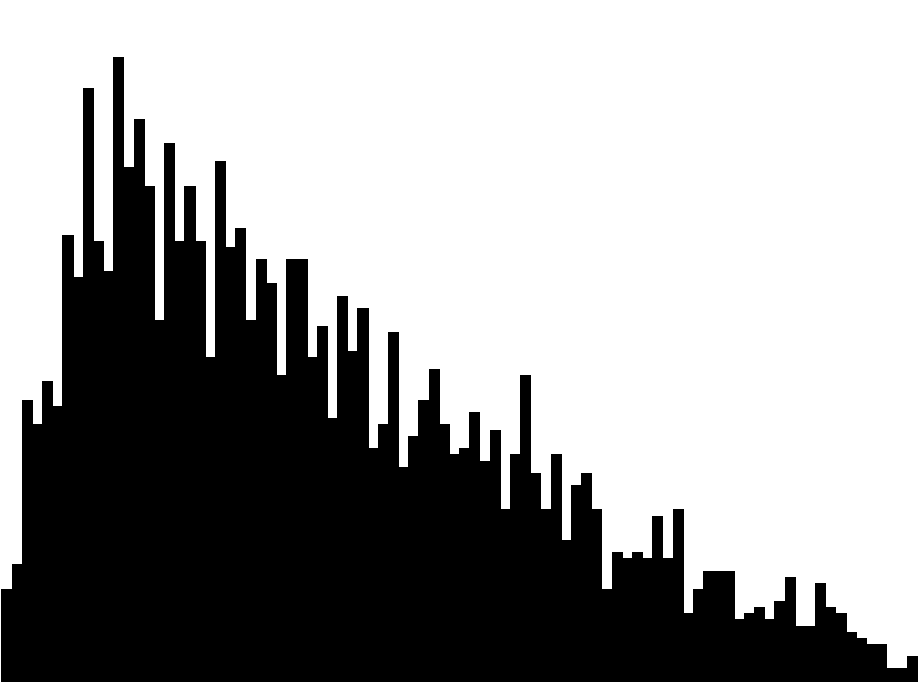} & \ih{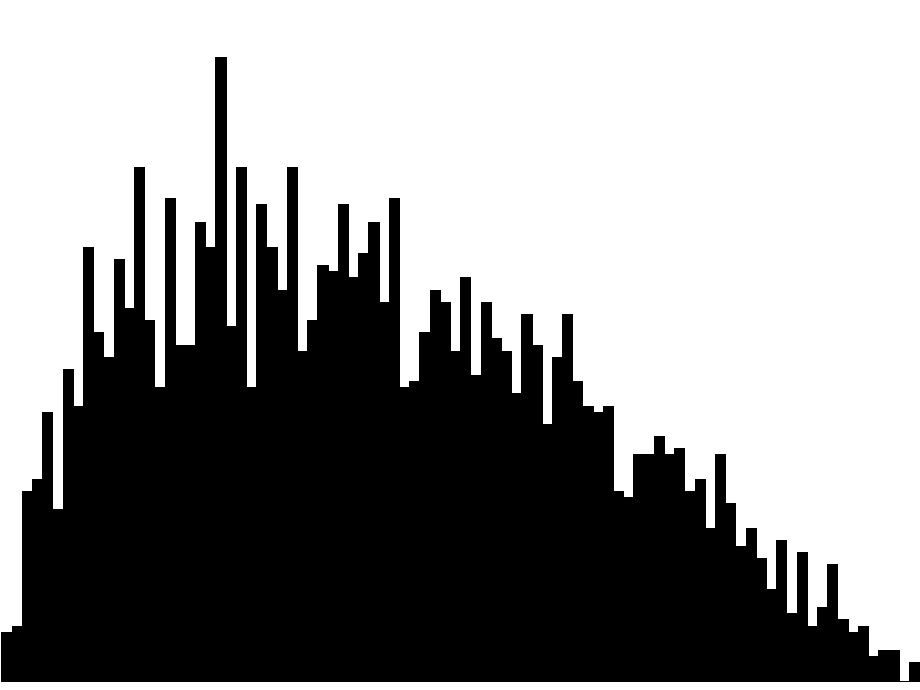} & \ih{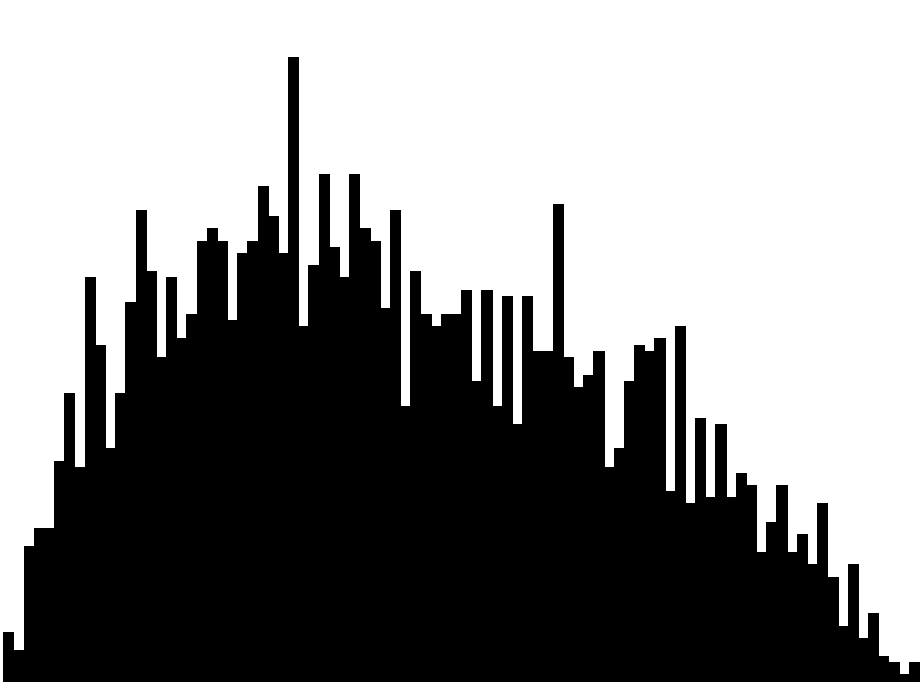} & \ih{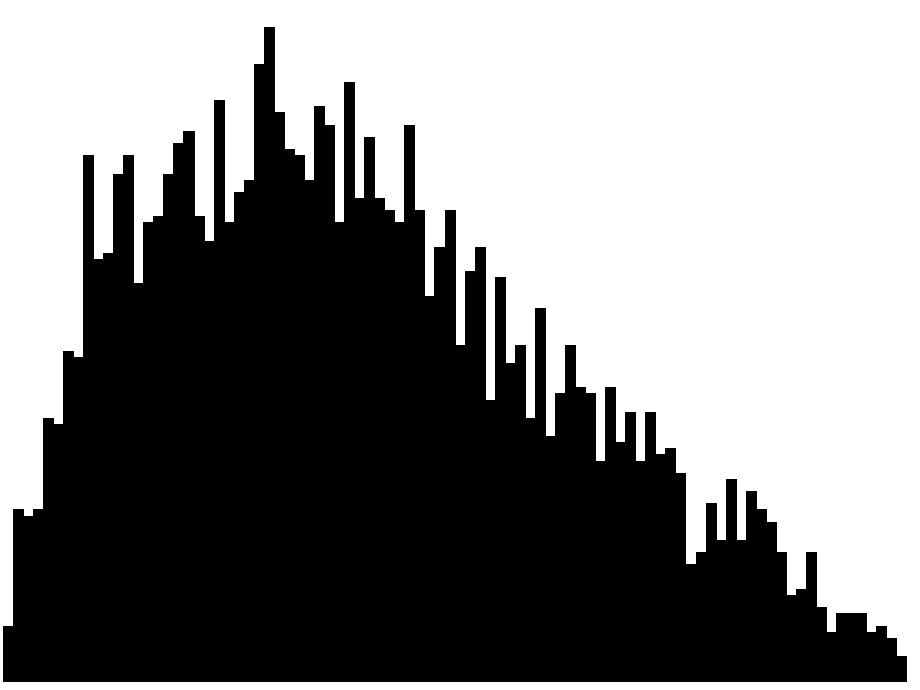} \\
\midrule
 \textit{d-lb} &\ih{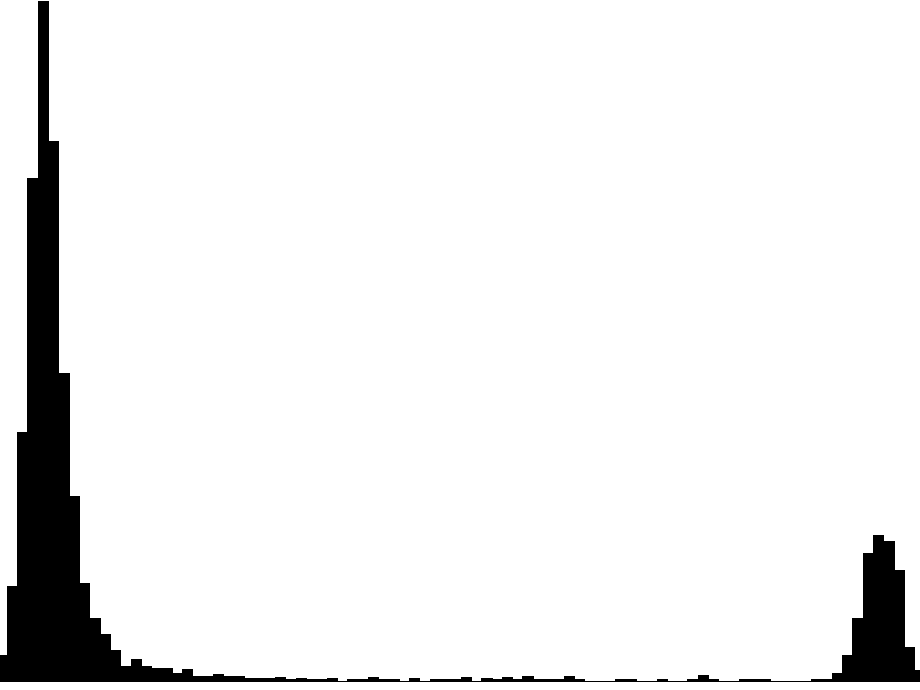}& \ih{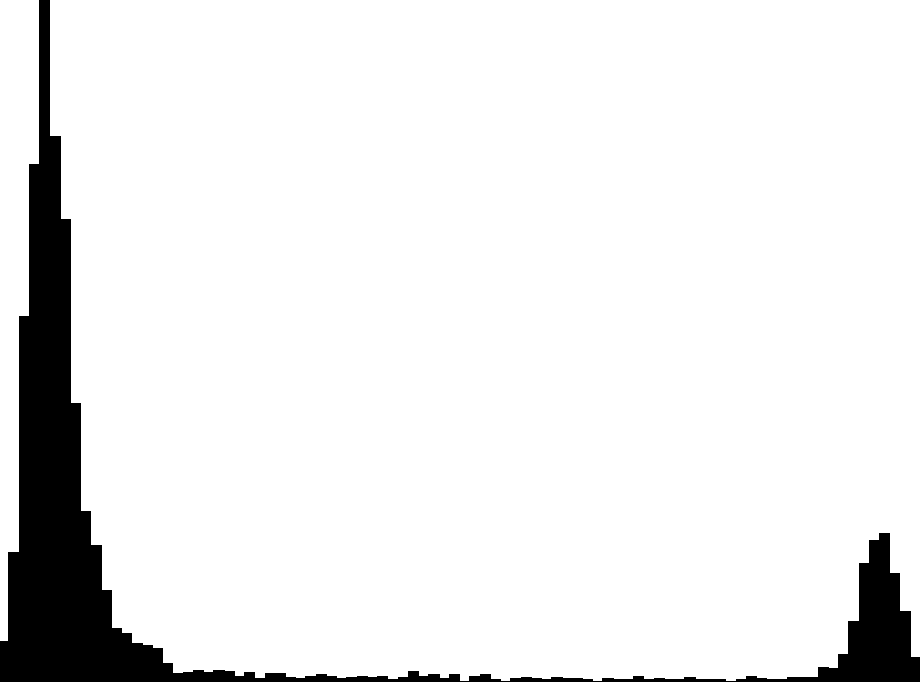} & \ih{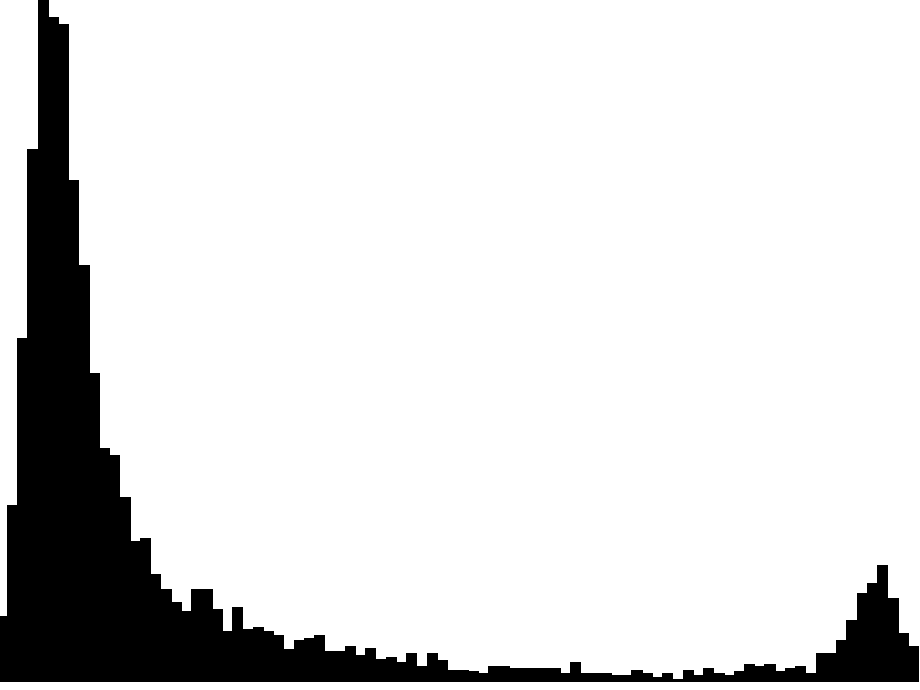} & \ih{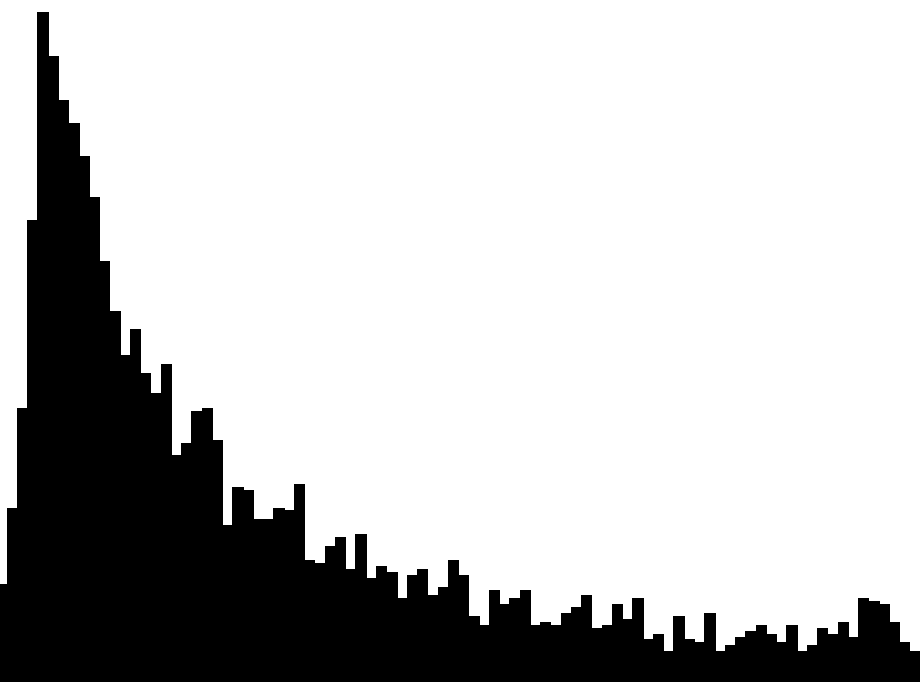} & \ih{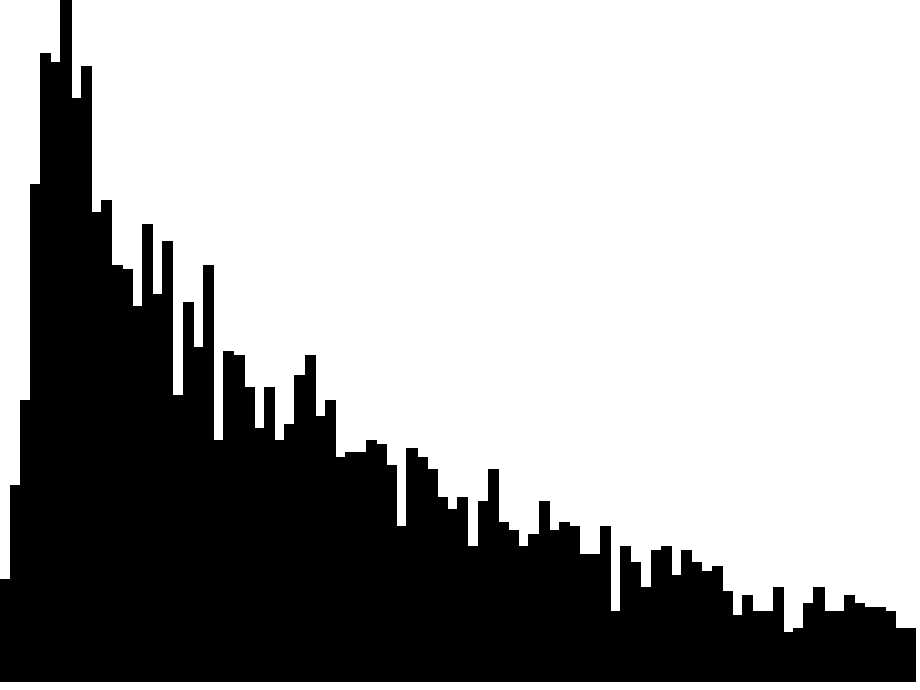} &
\ih{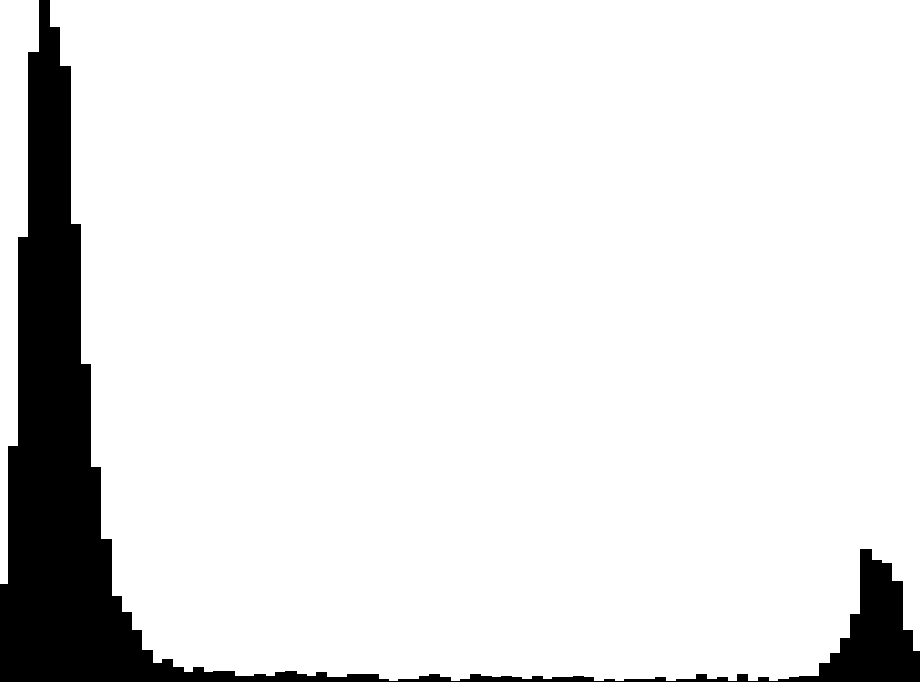} & \ih{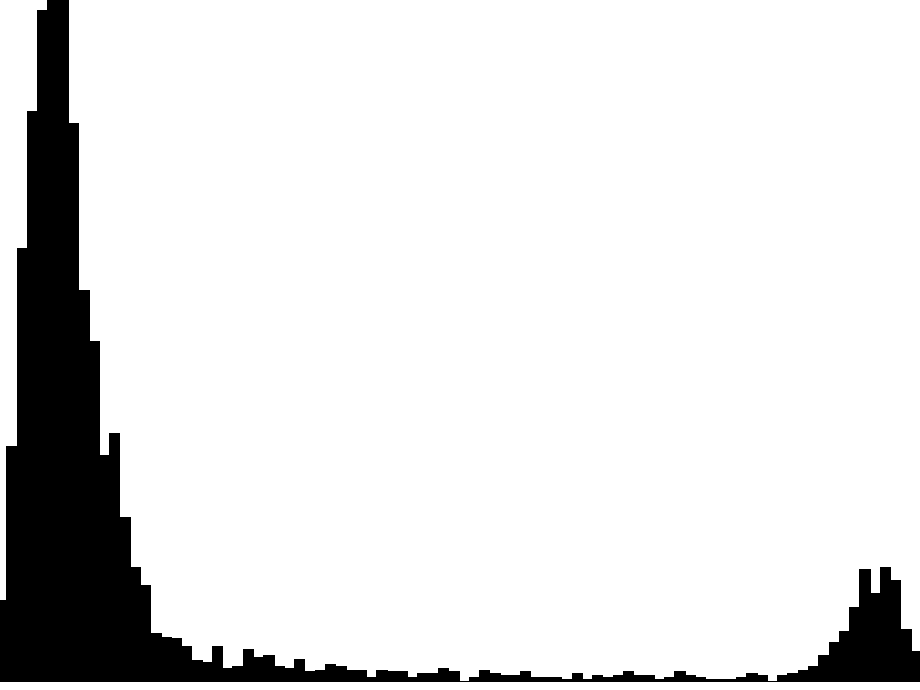} & \ih{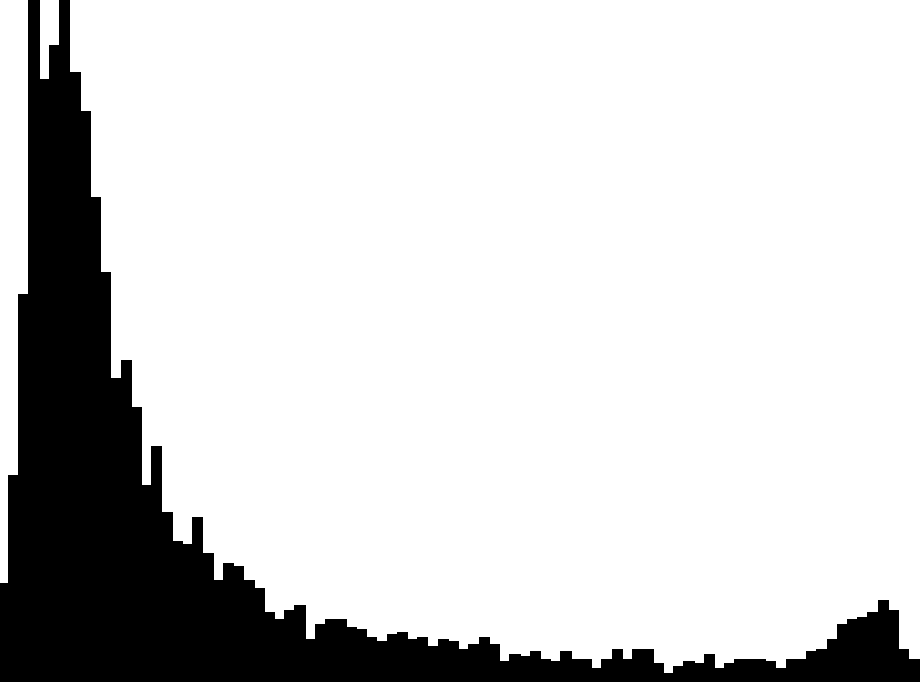} & \ih{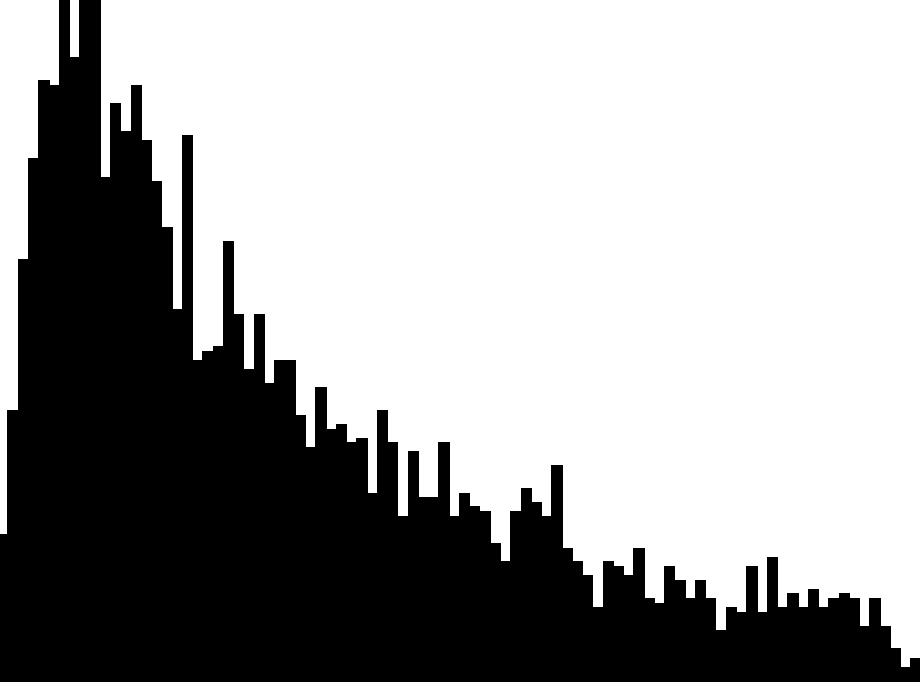} & \ih{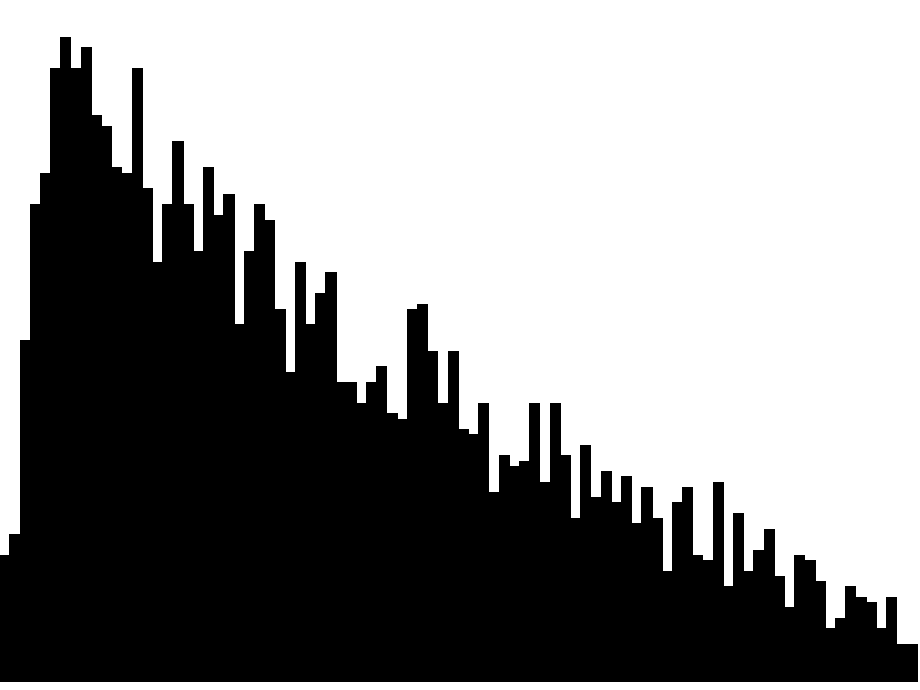} & \ih{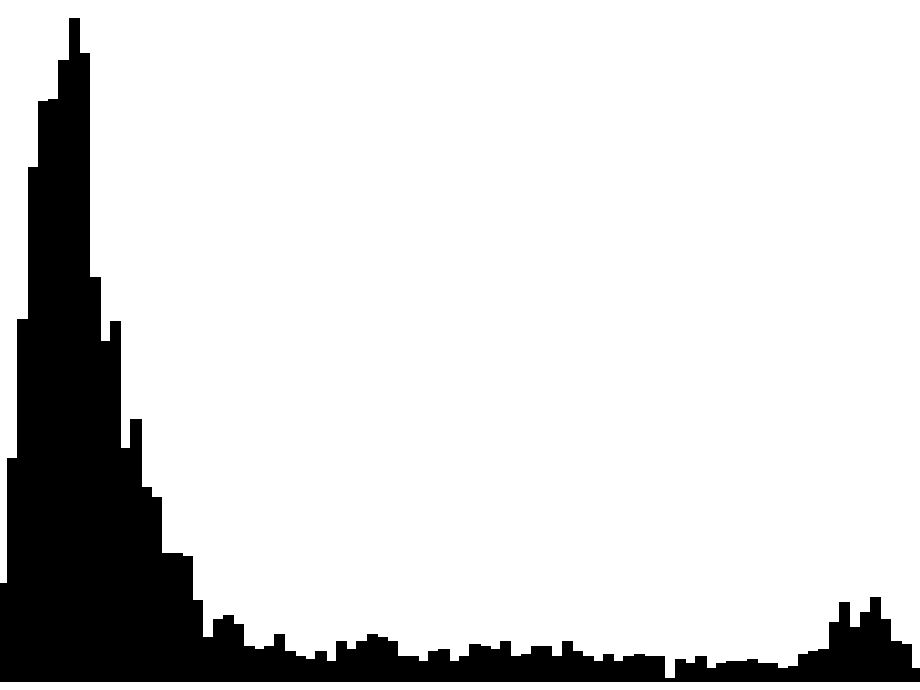} &
\ih{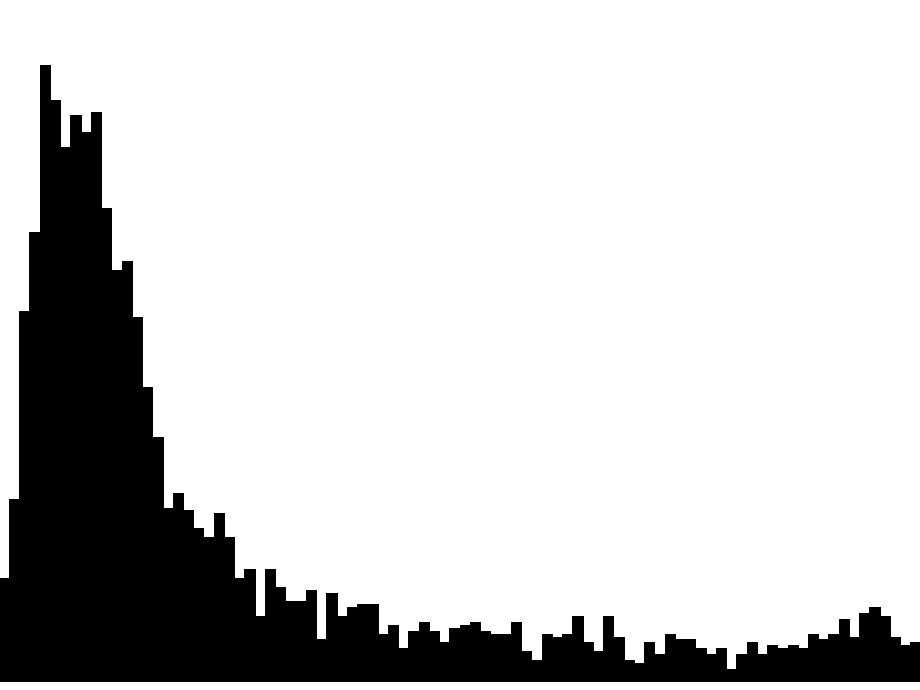} & \ih{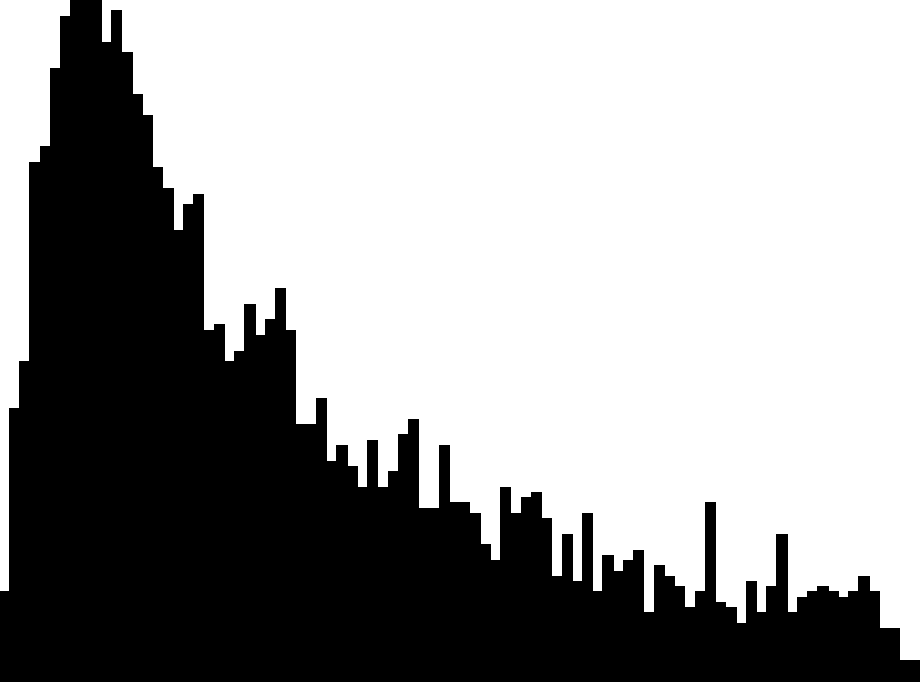} & \ih{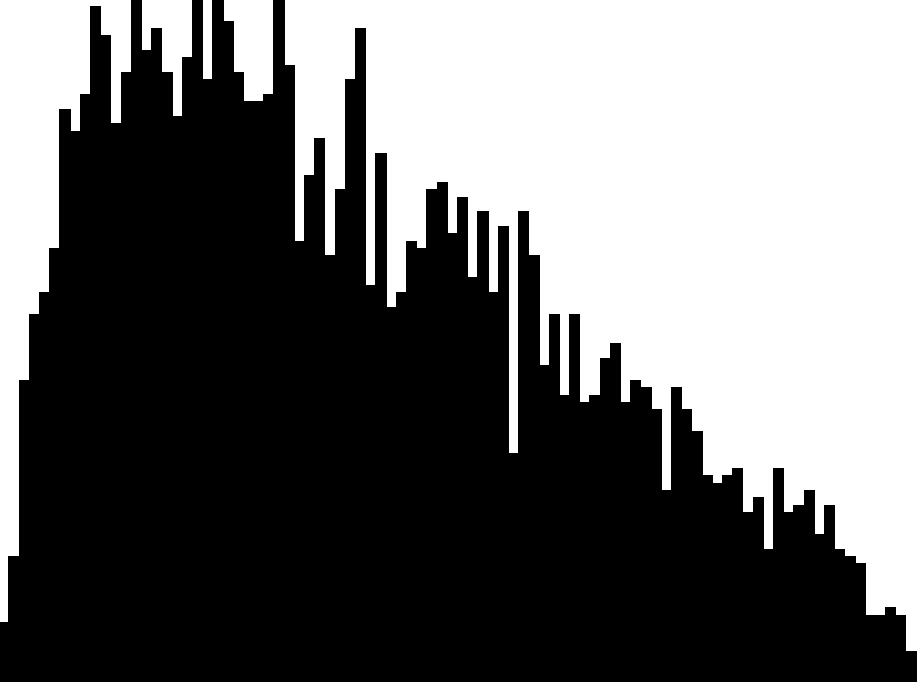} & \ih{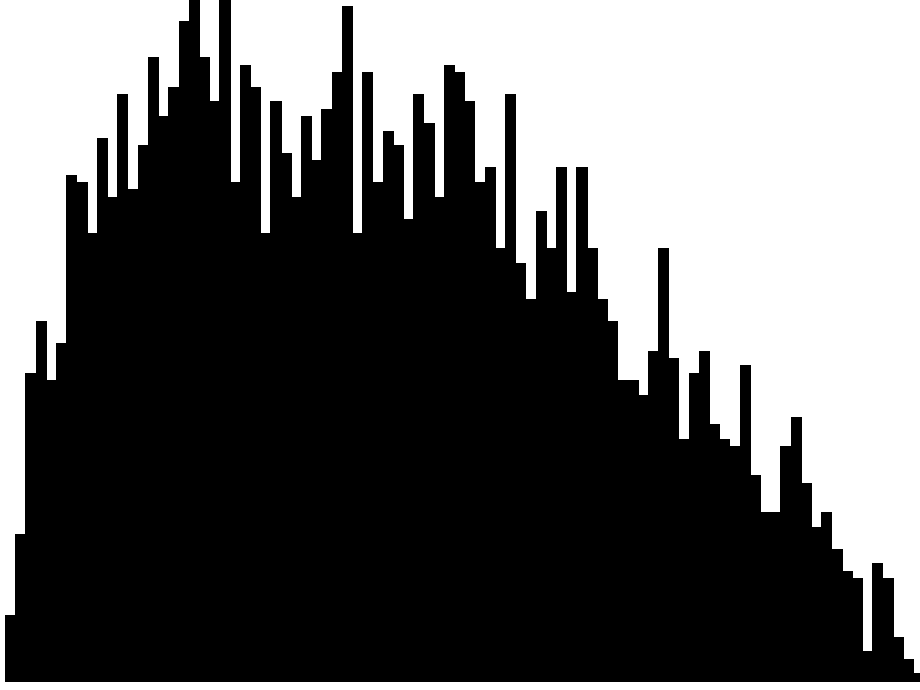} & \ih{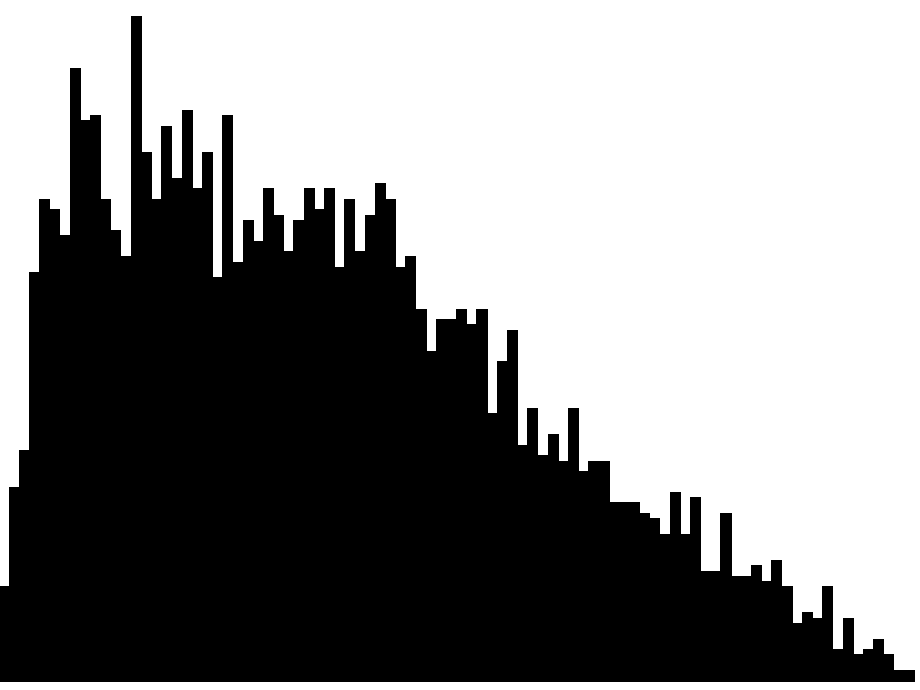} \\
\midrule
 \textit{s-lb} &\ih{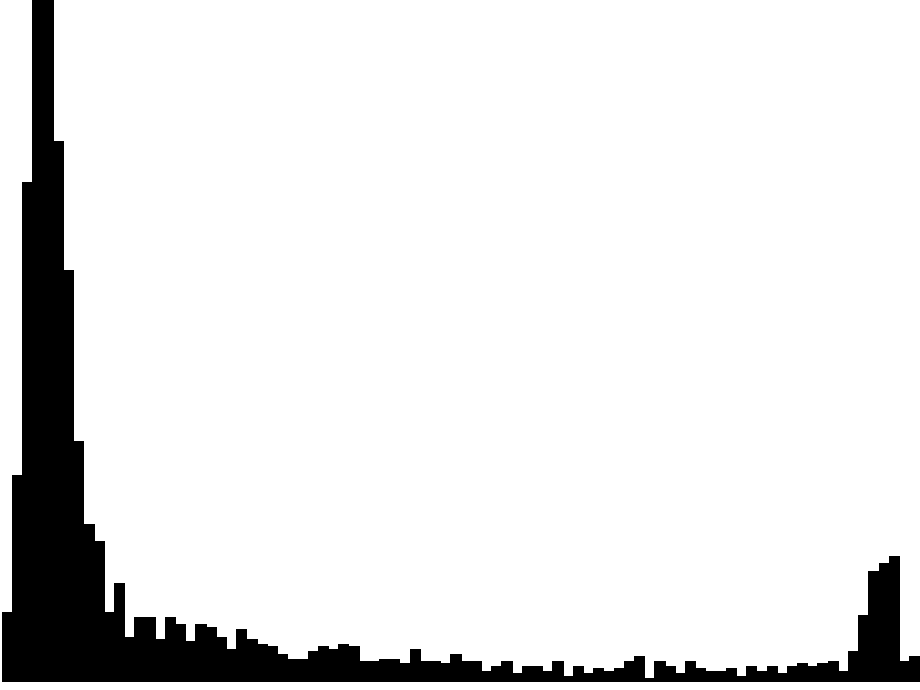}& \ih{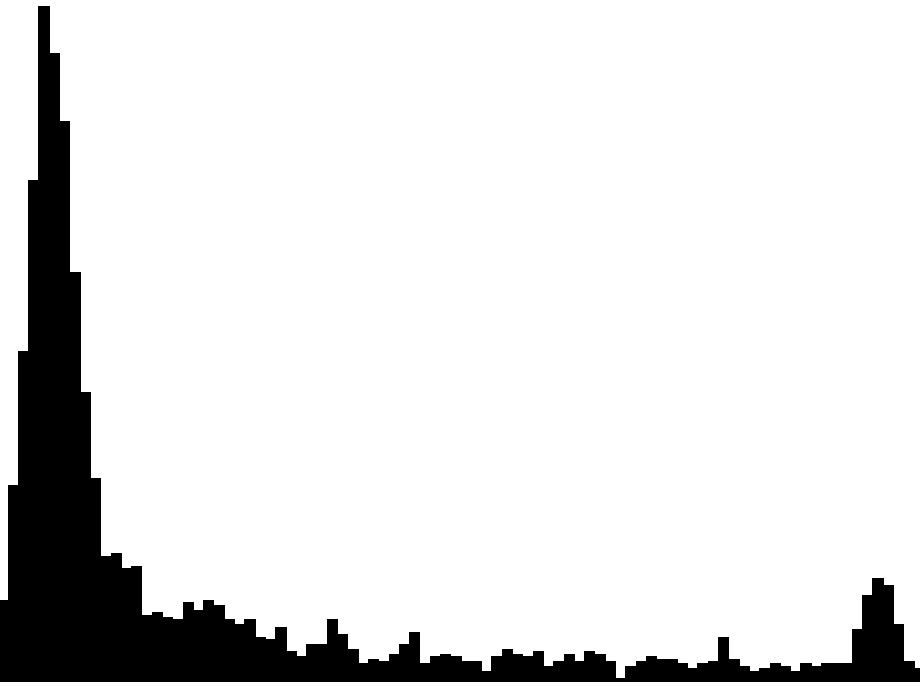} & \ih{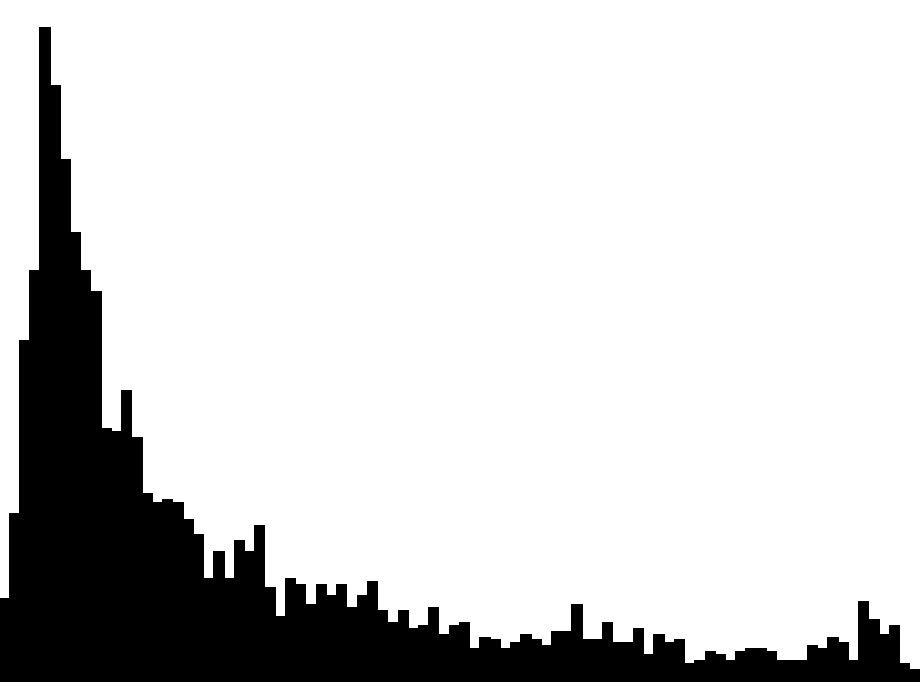} & \ih{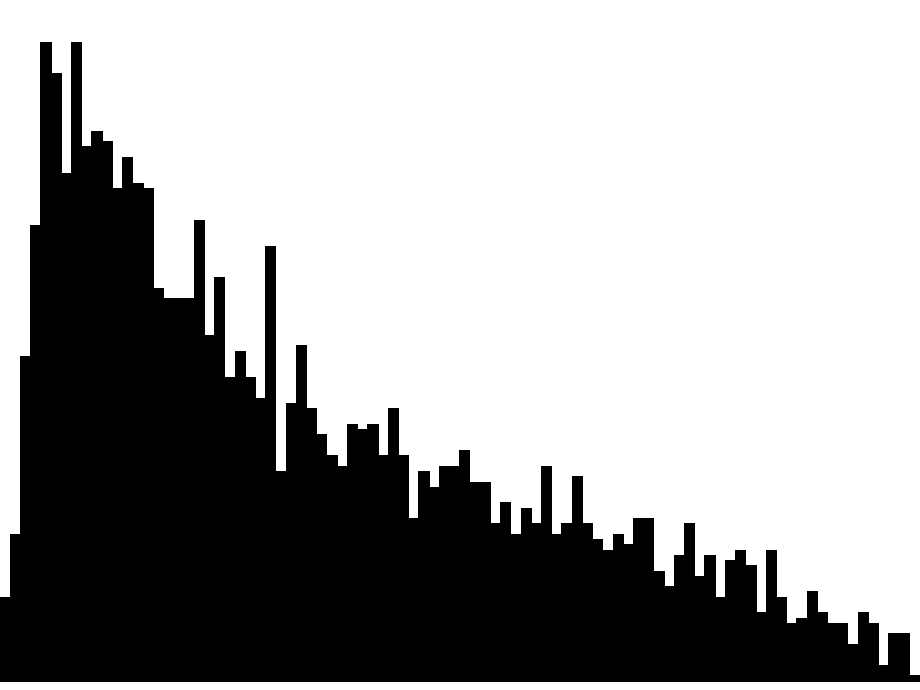} & \ih{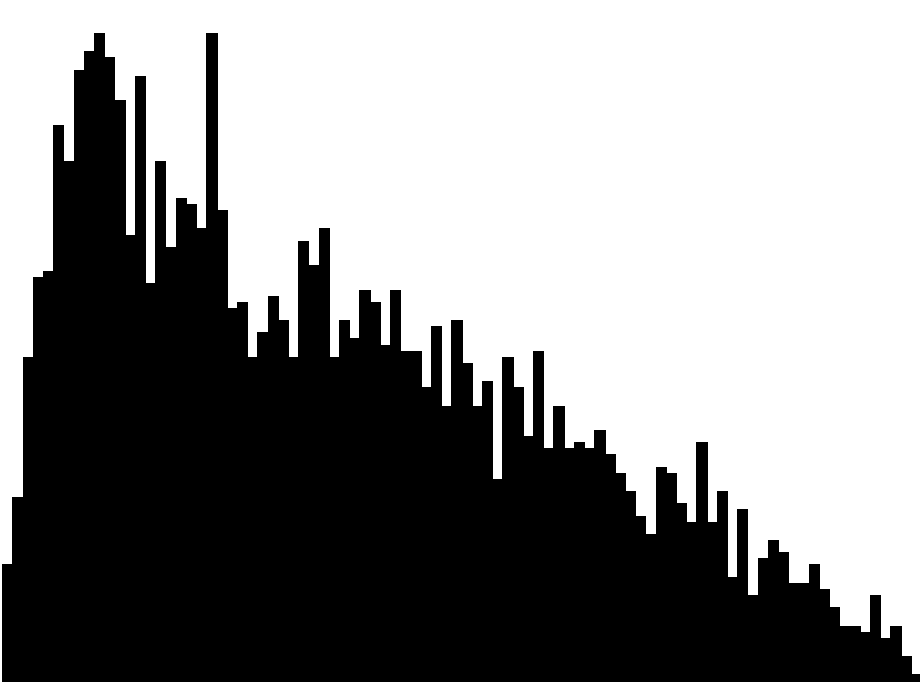} &
\ih{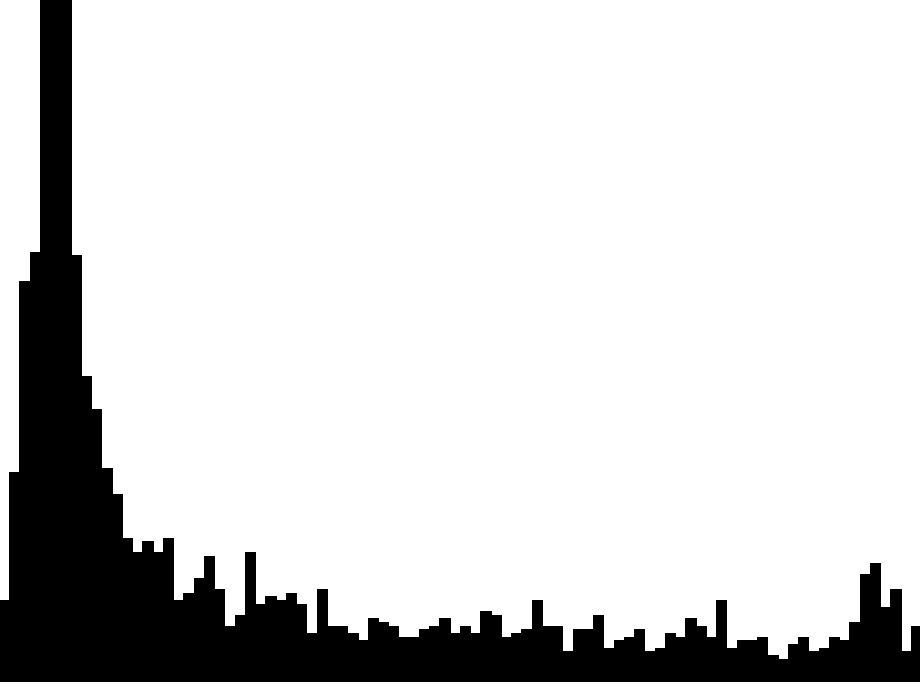} & \ih{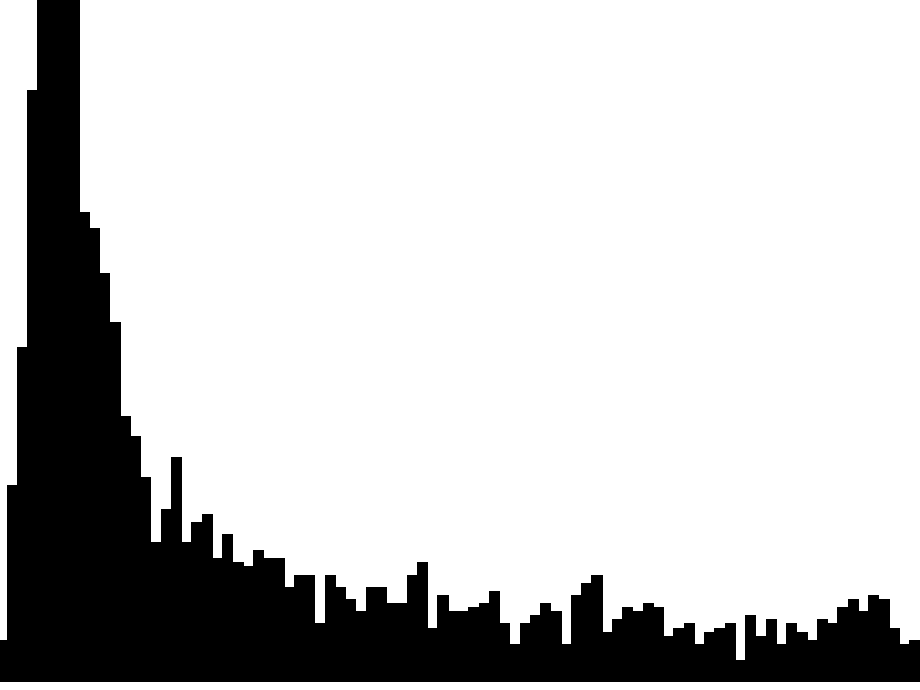} & \ih{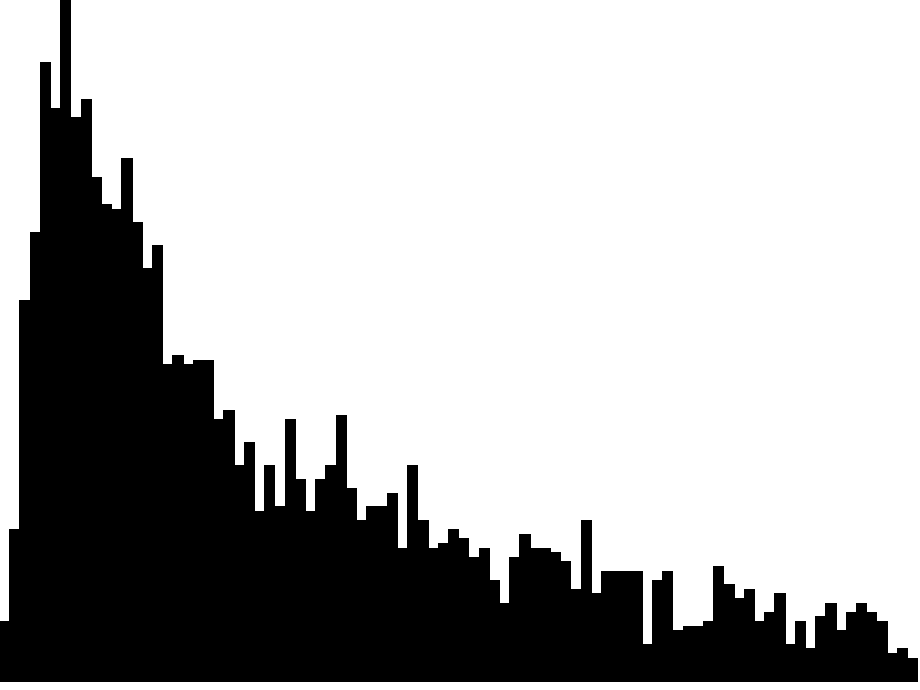} & \ih{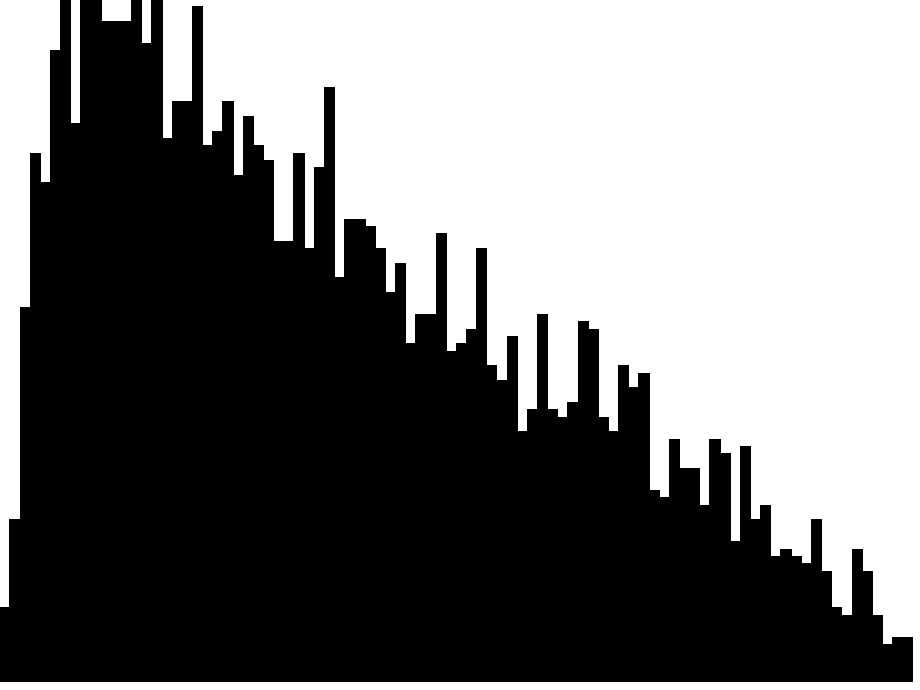} & \ih{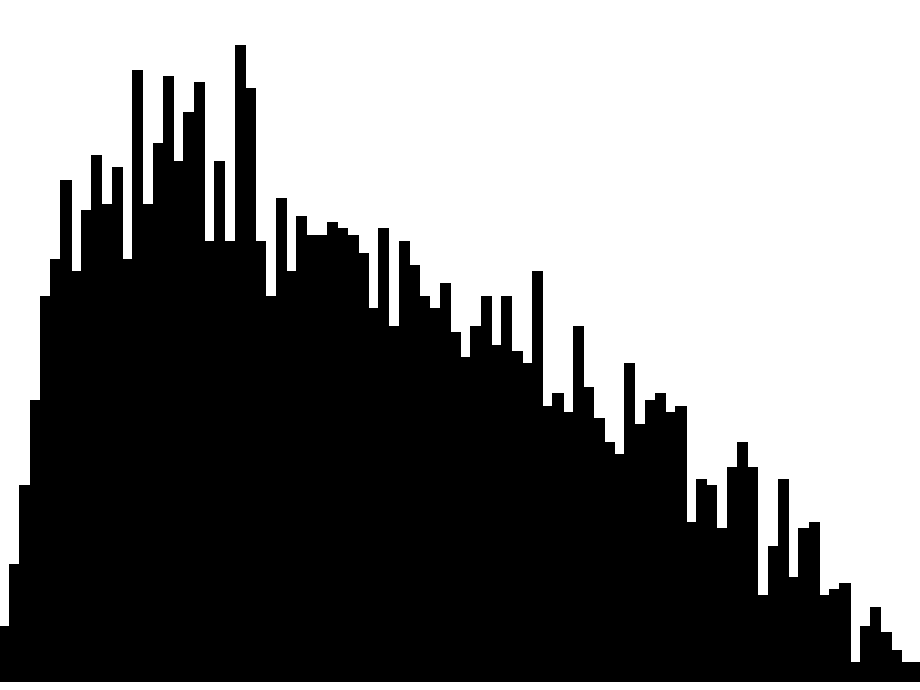} & \ih{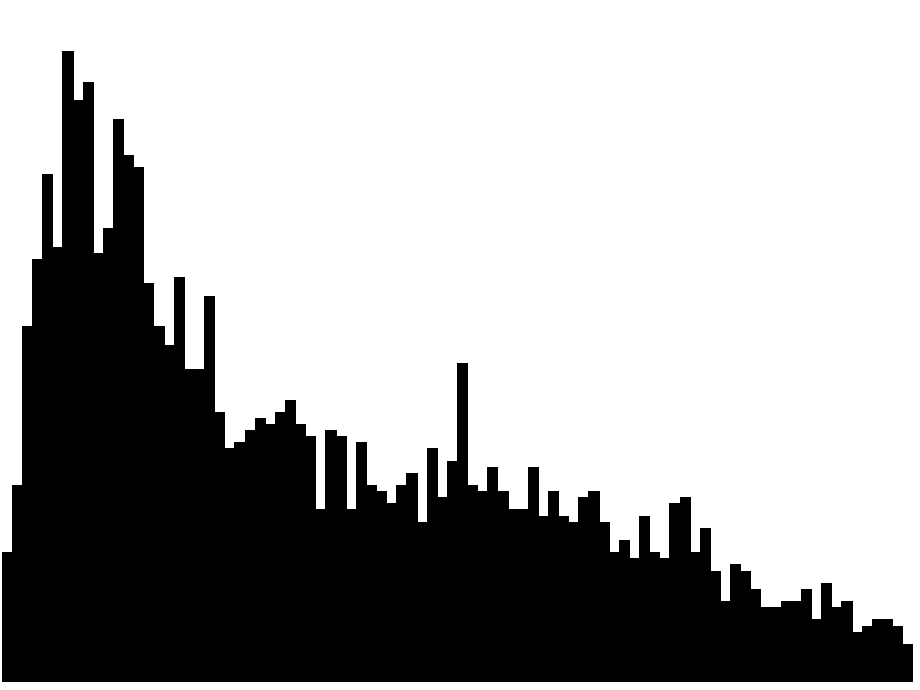} &
\ih{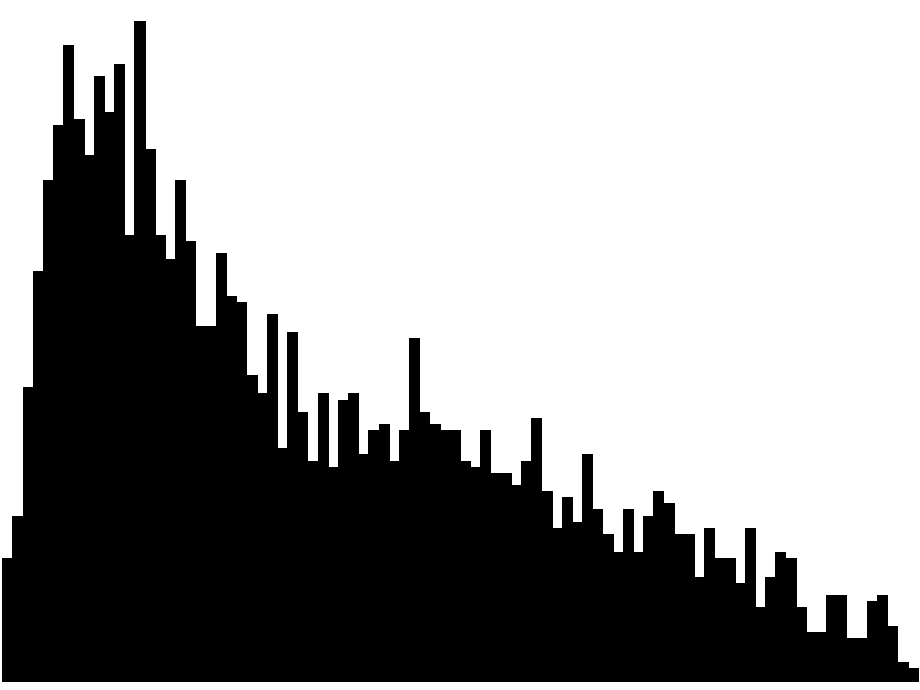} & \ih{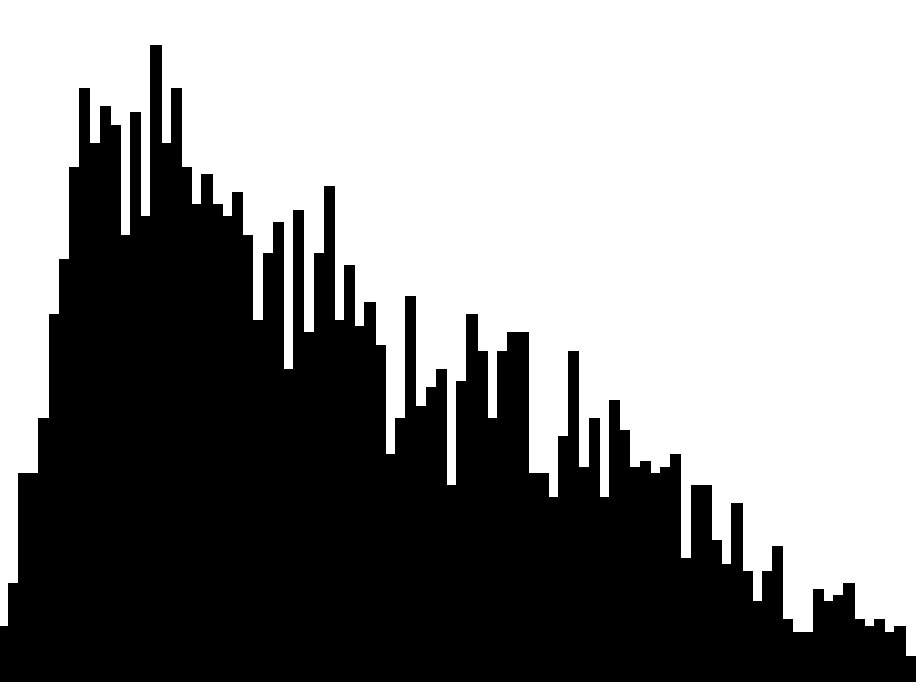} & \ih{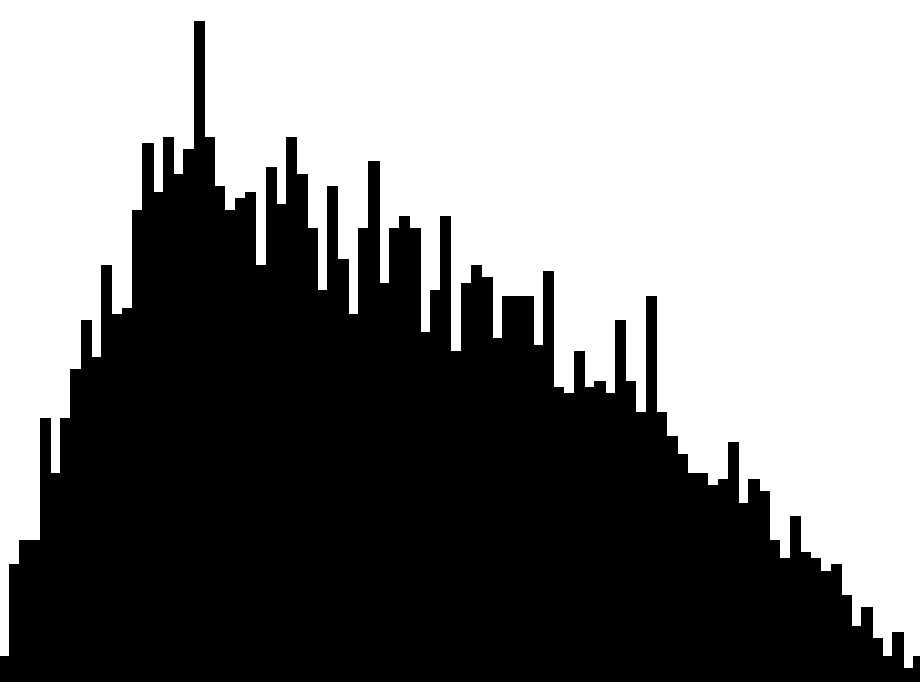} & \ih{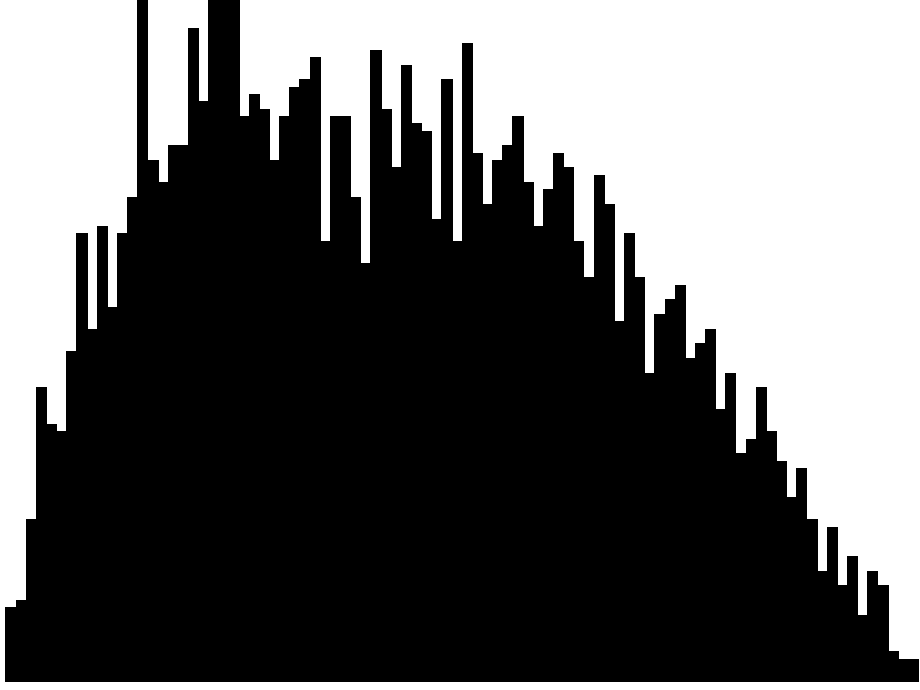} & \ih{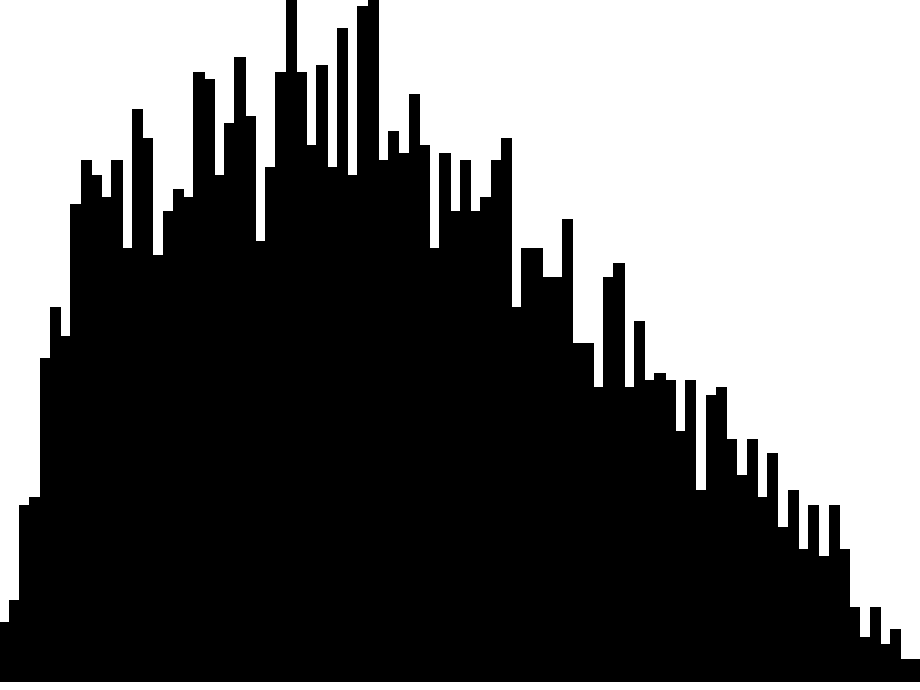} \\
\midrule
\textit{dm} &\ih{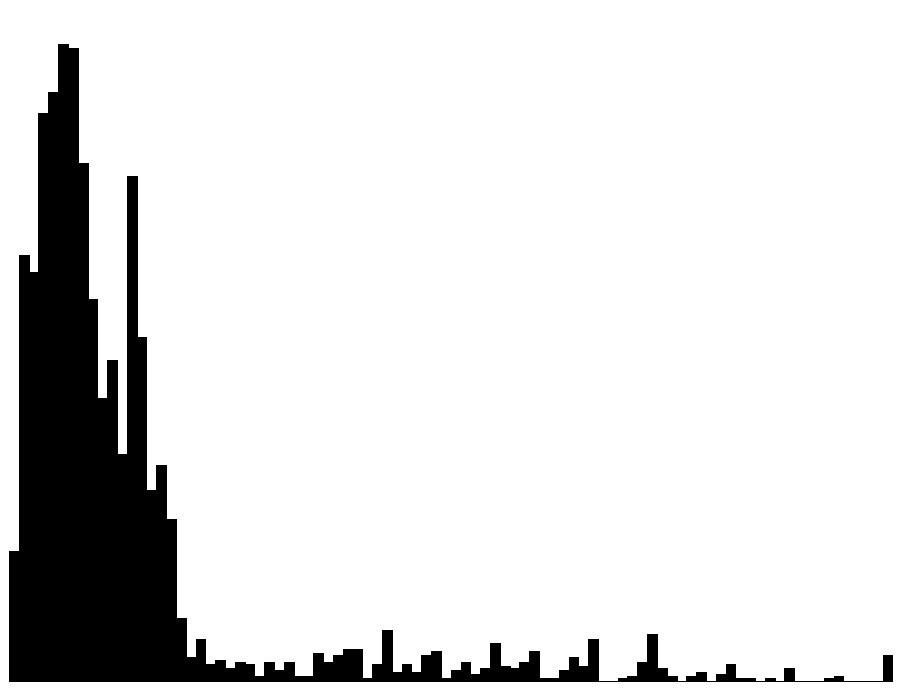}& \ih{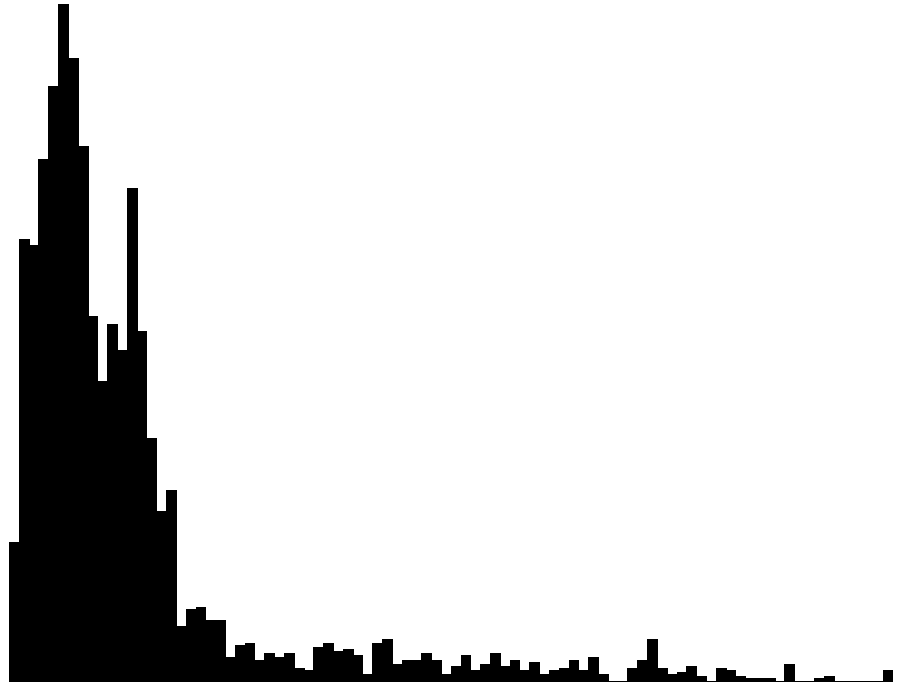} & \ih{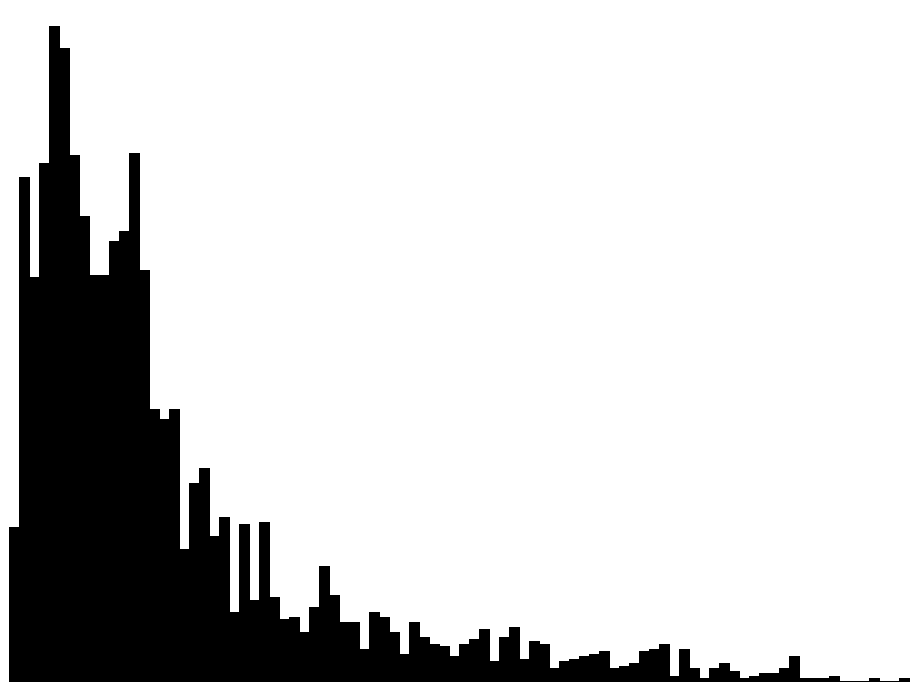} & \ih{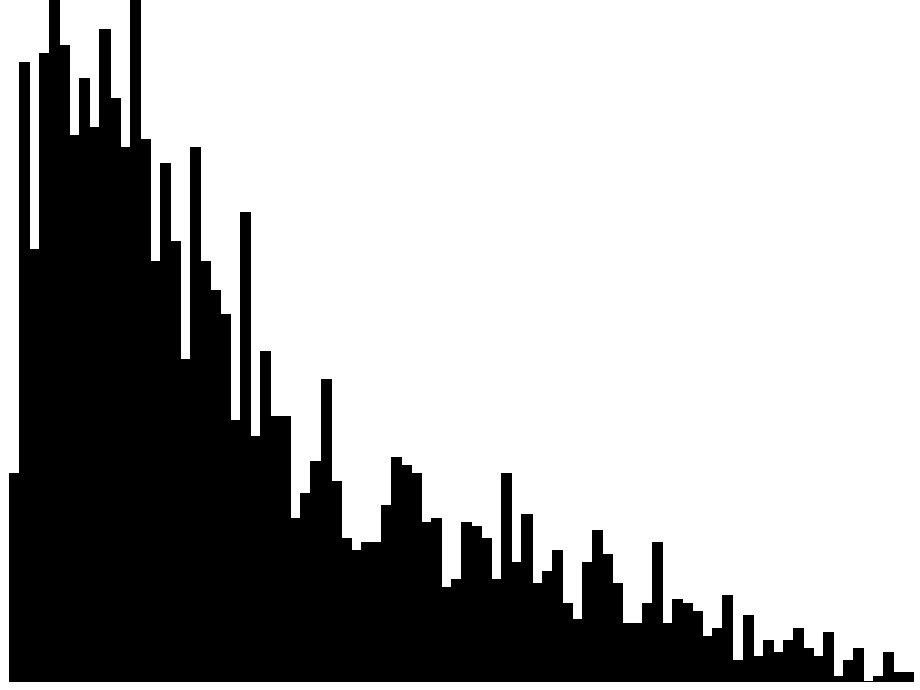} & \ih{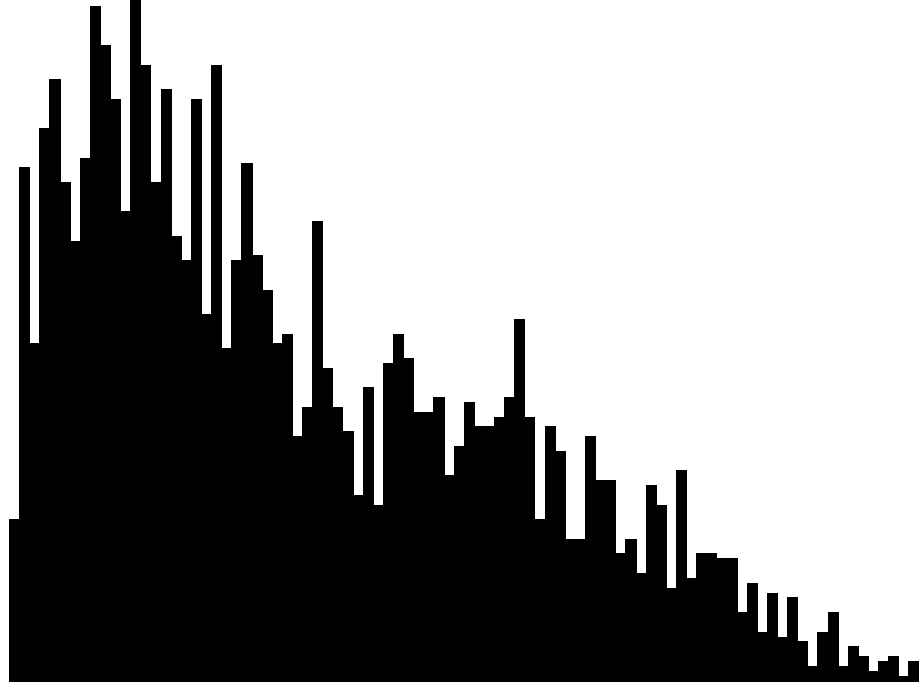} &
\ih{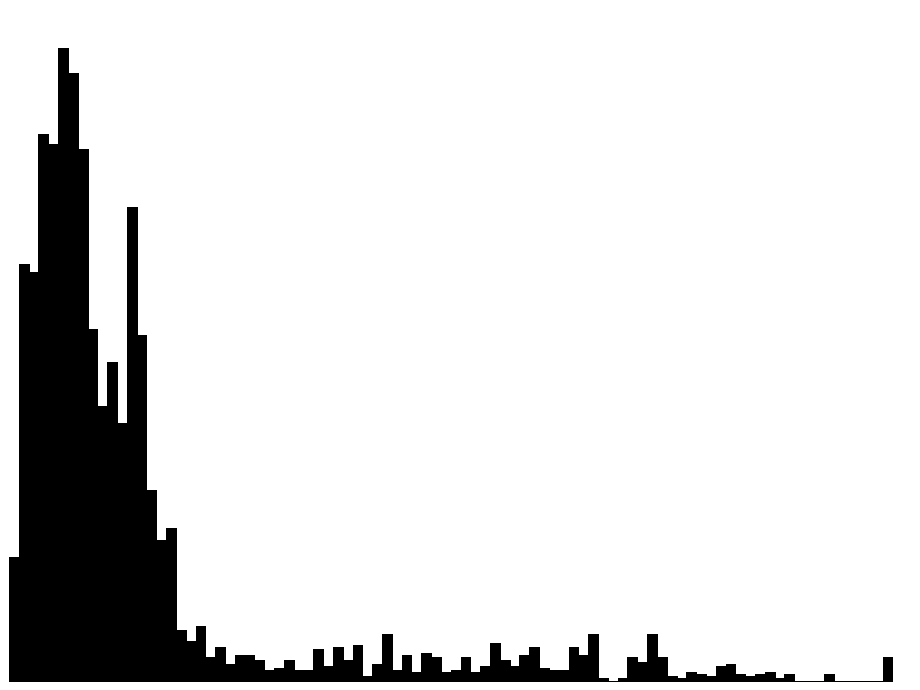} & \ih{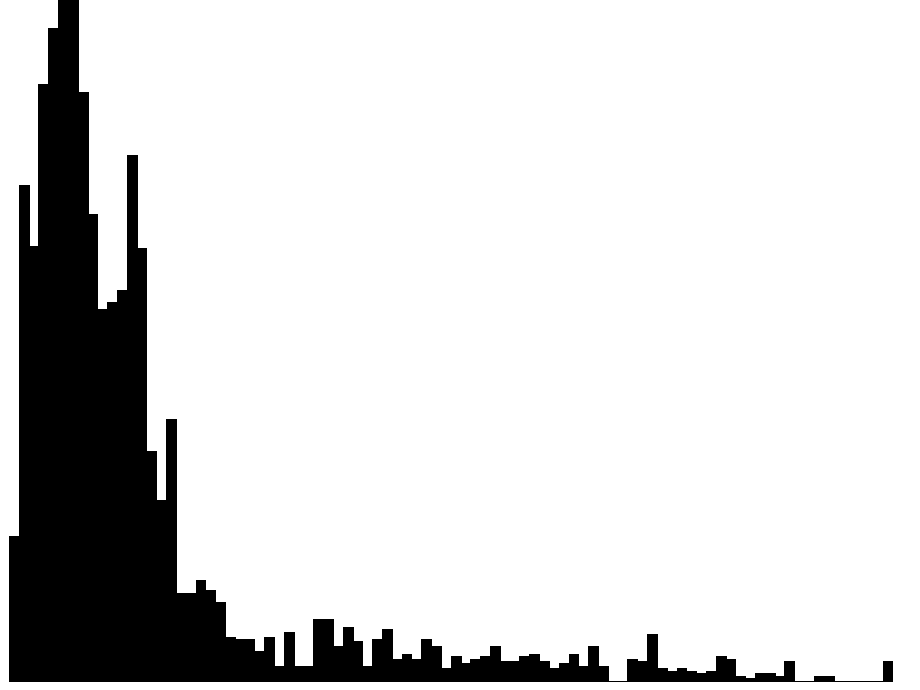} & \ih{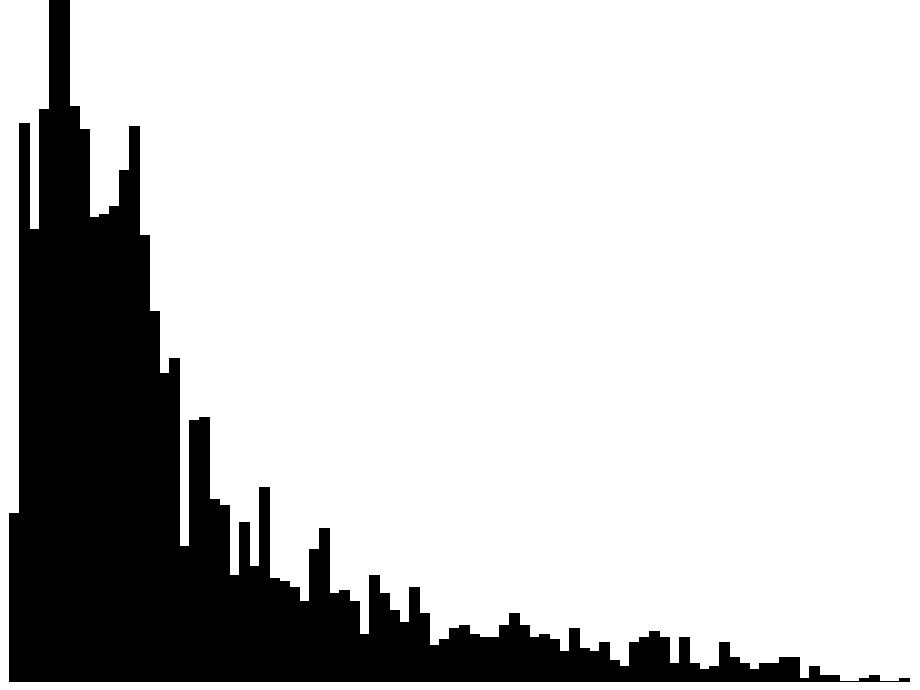} & \ih{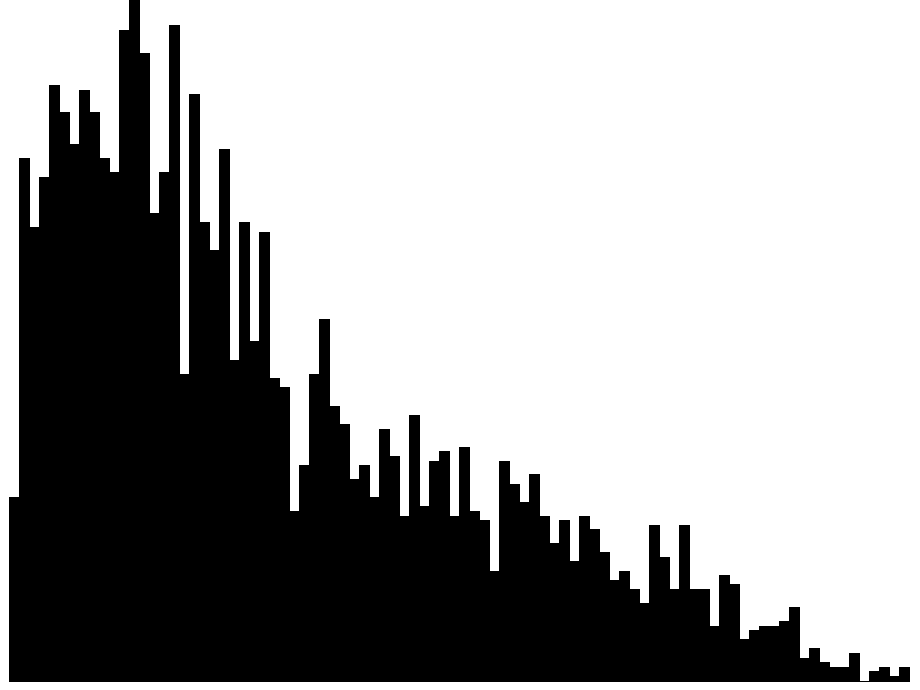} & \ih{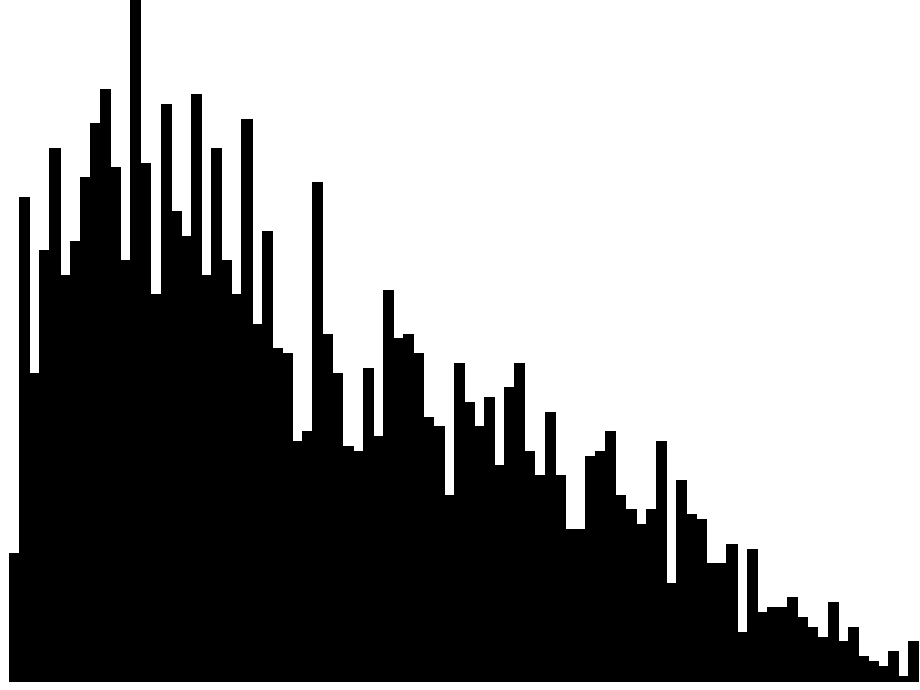} & \ih{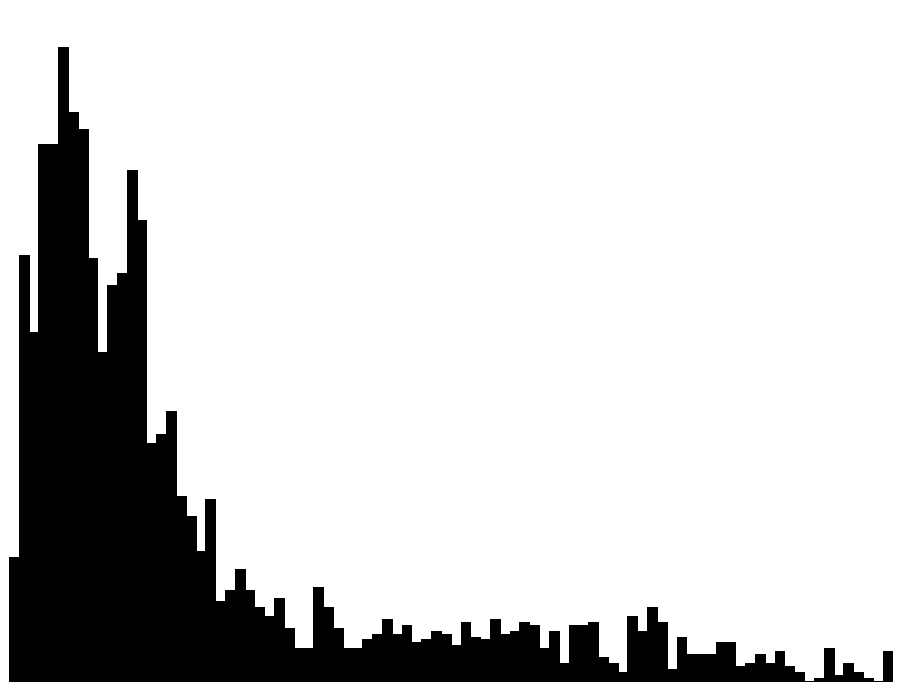} &
\ih{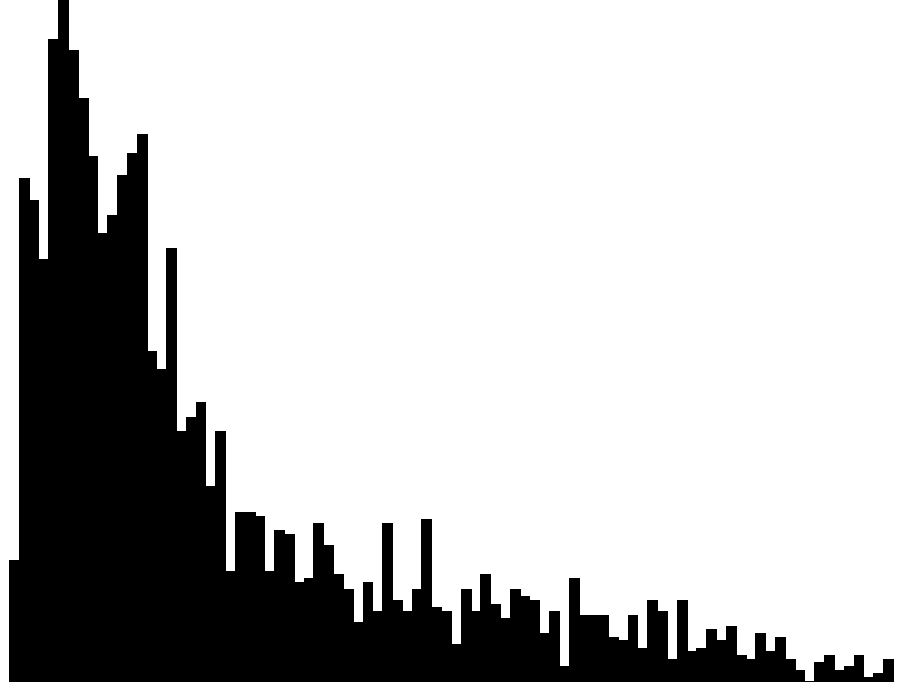} & \ih{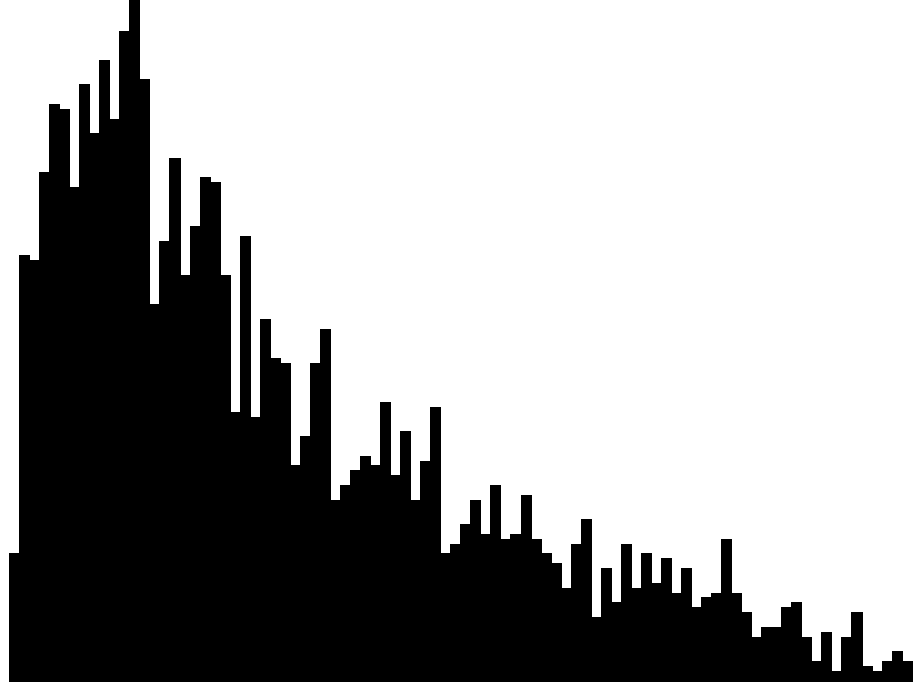} & \ih{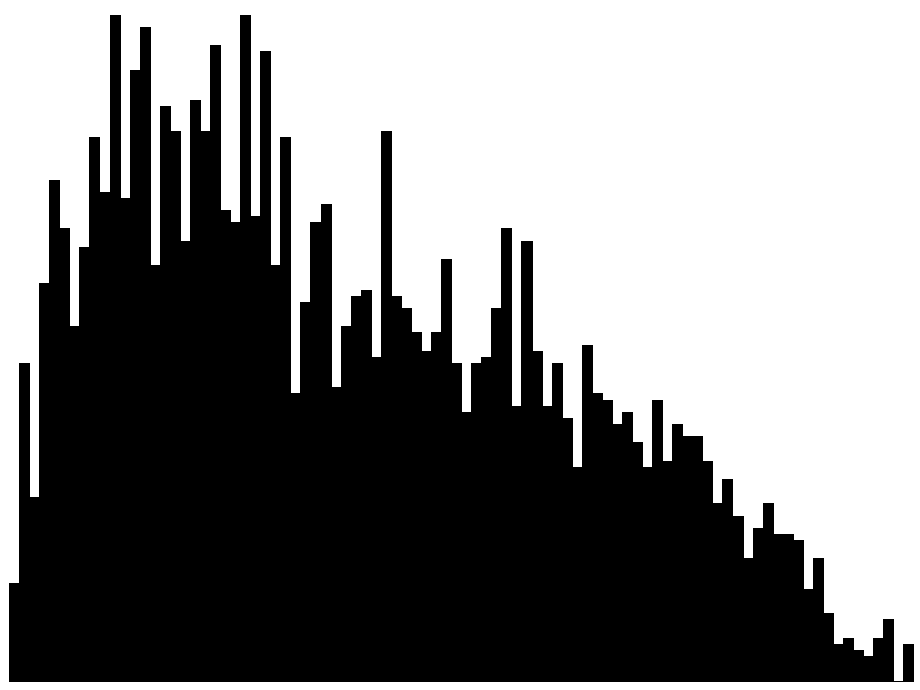} & \ih{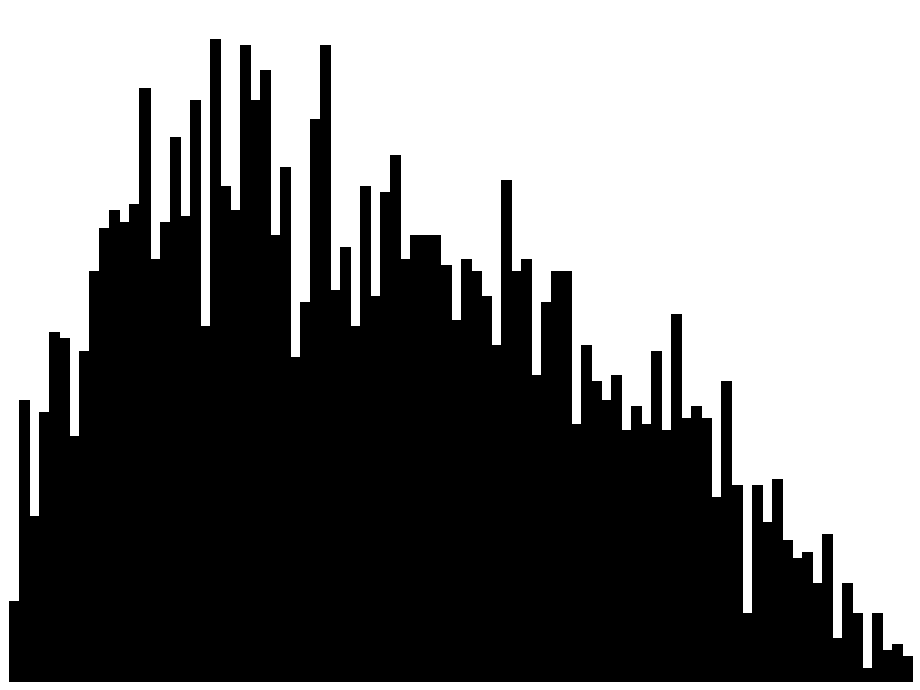} & \ih{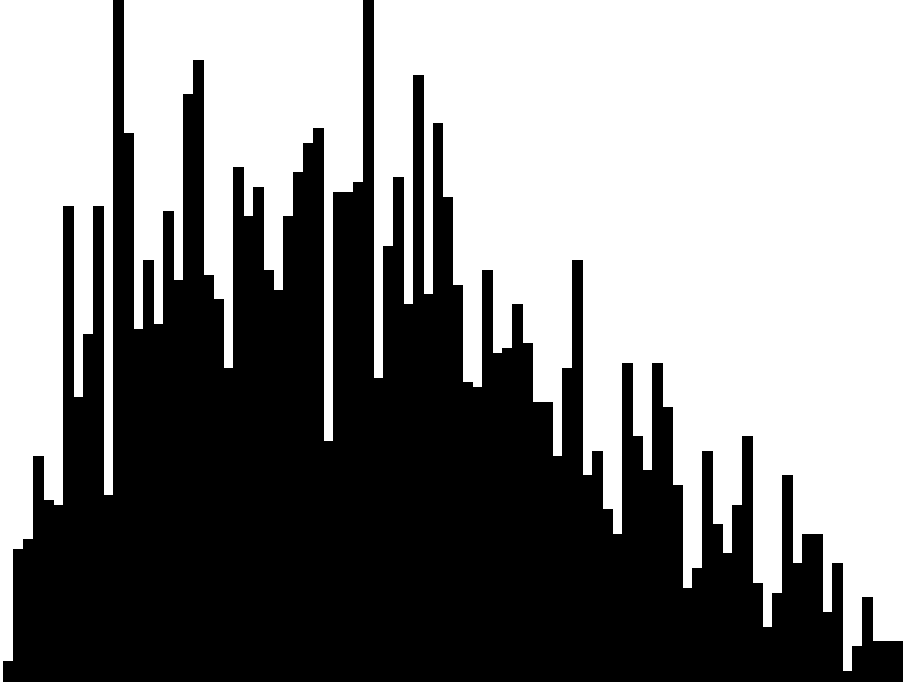} \\
\midrule\midrule
\textit{n-mult} &\ih{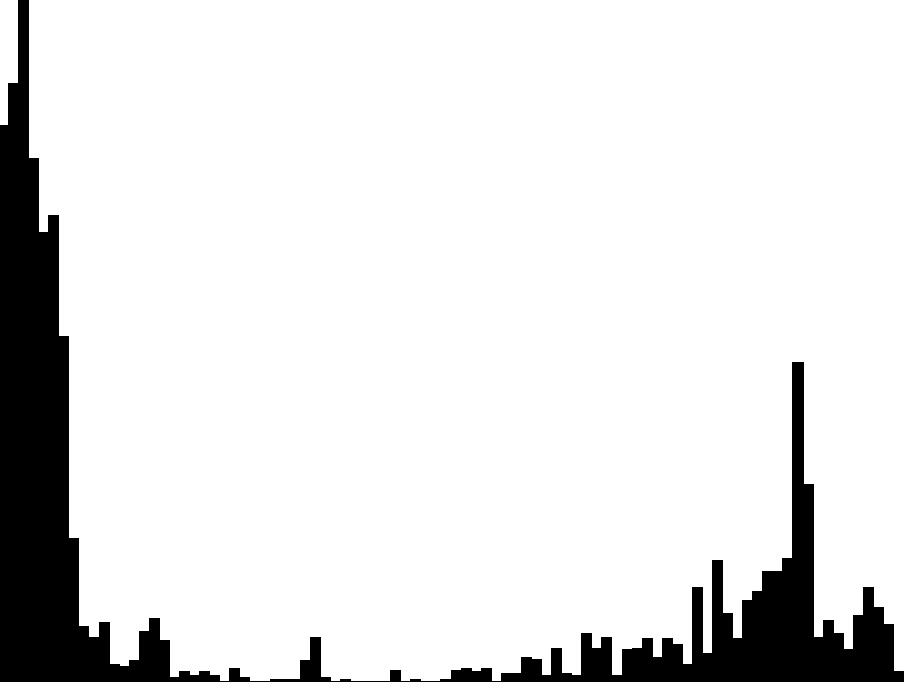}& \ih{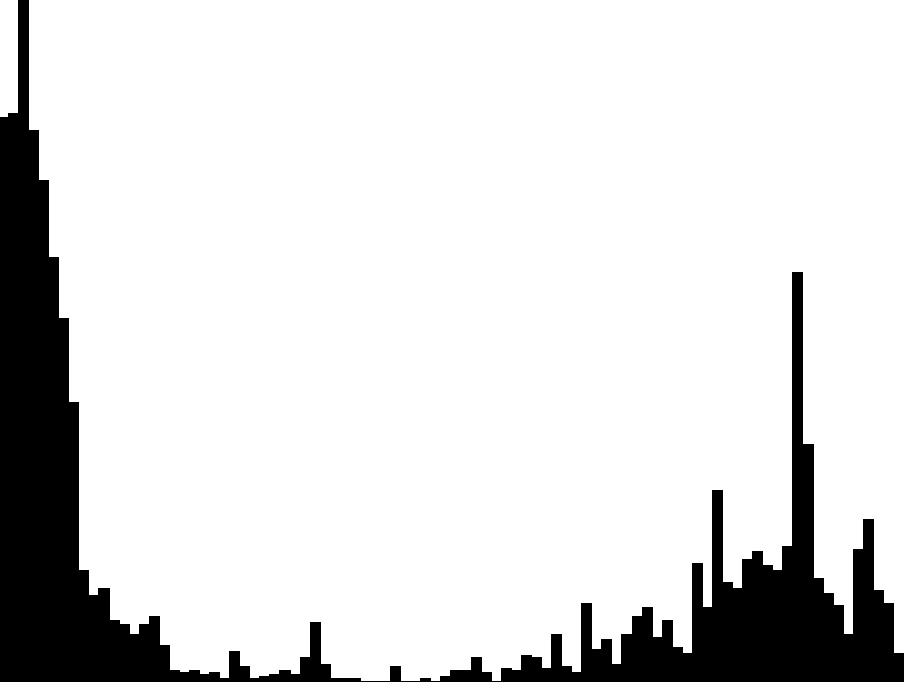} & \ih{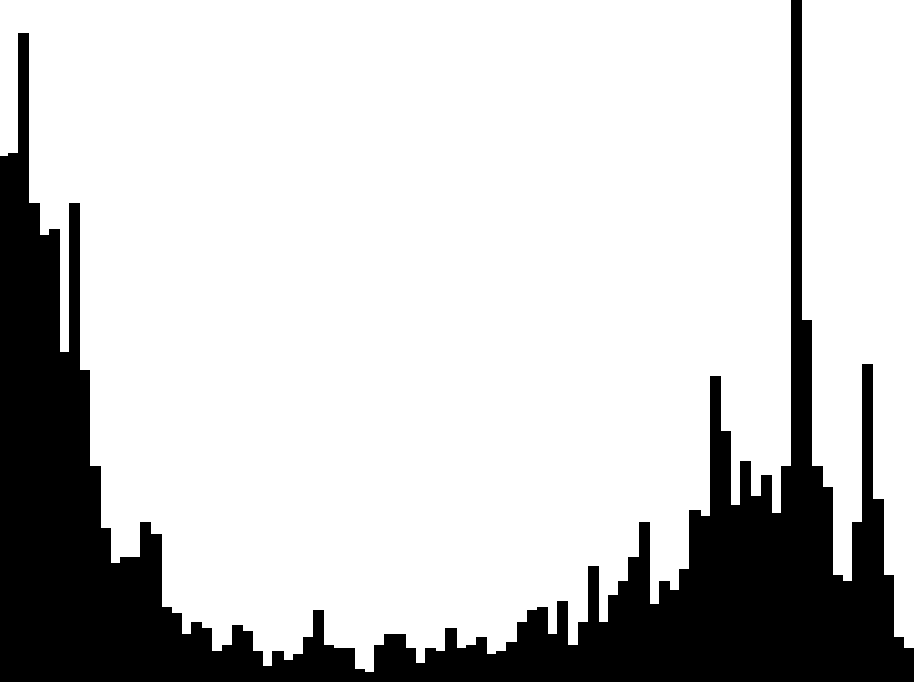} & \ih{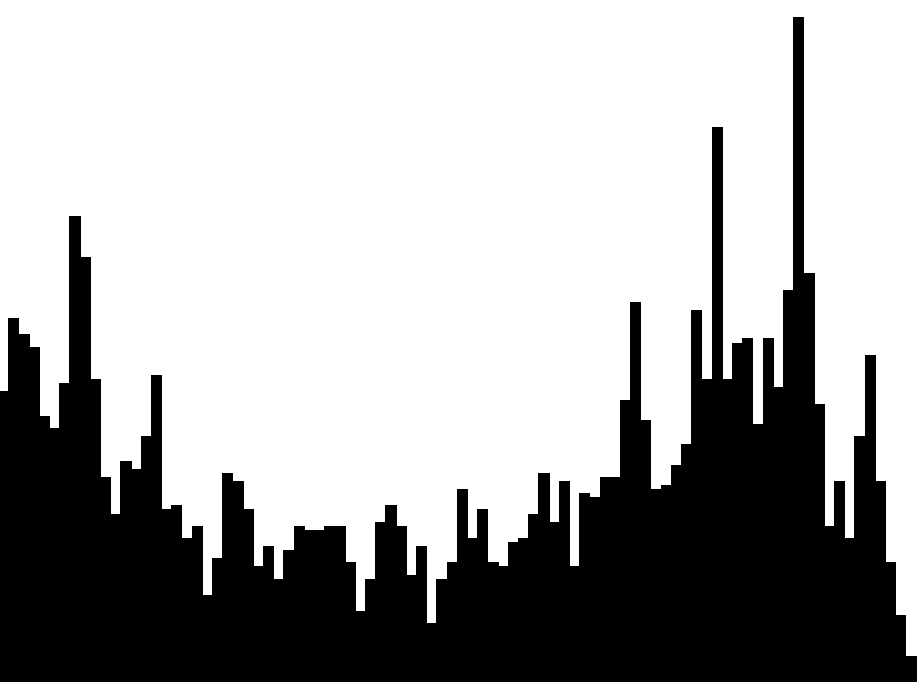} &
\ih{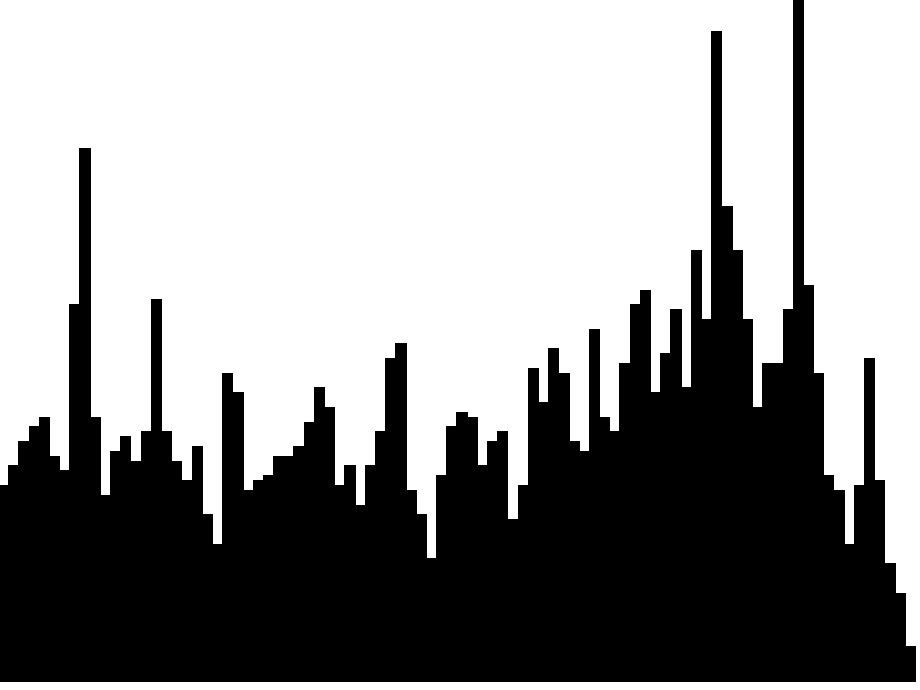} & \ih{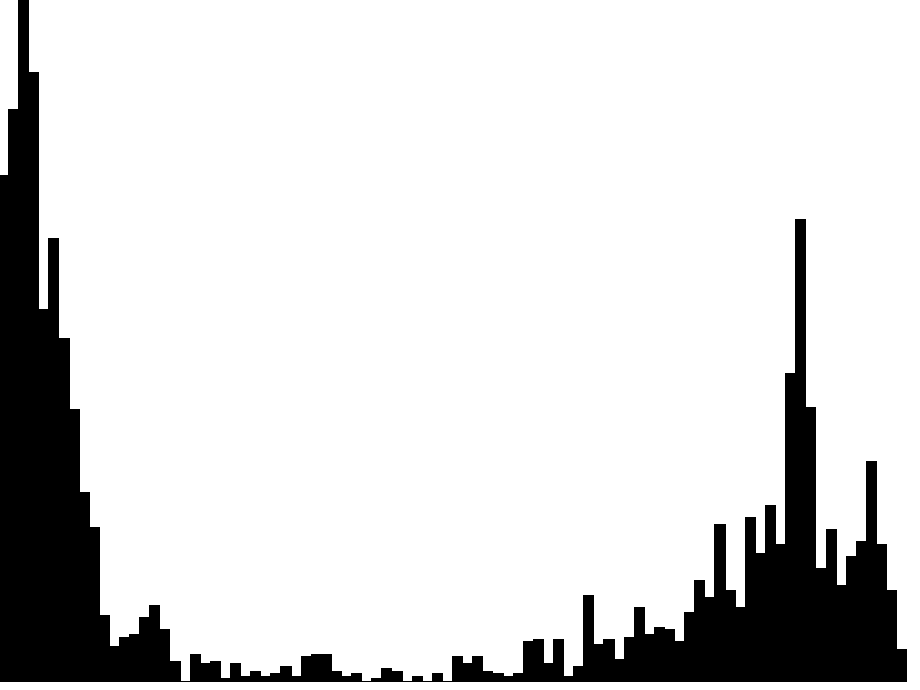} & \ih{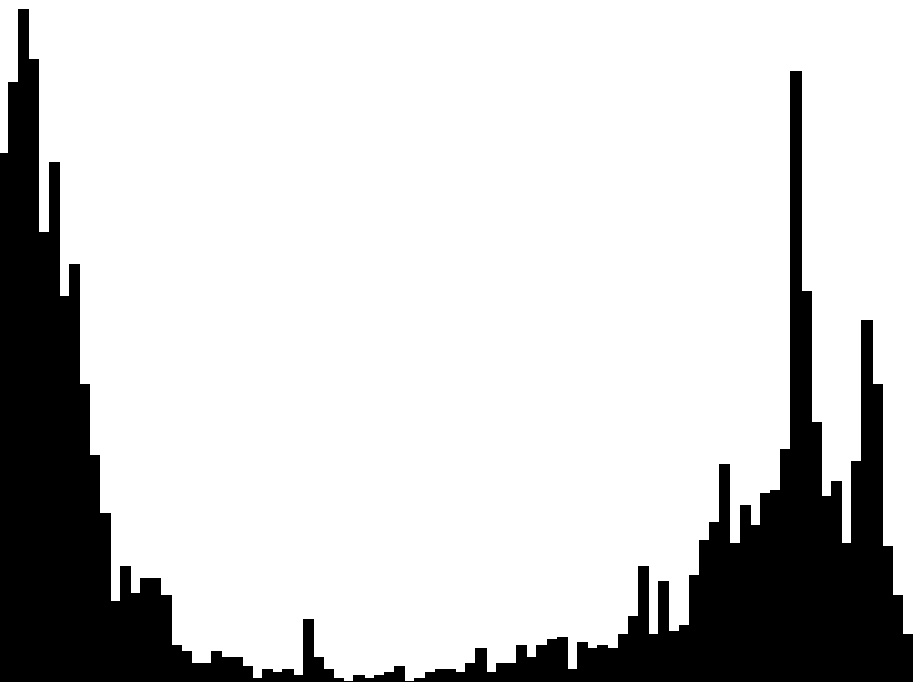} & \ih{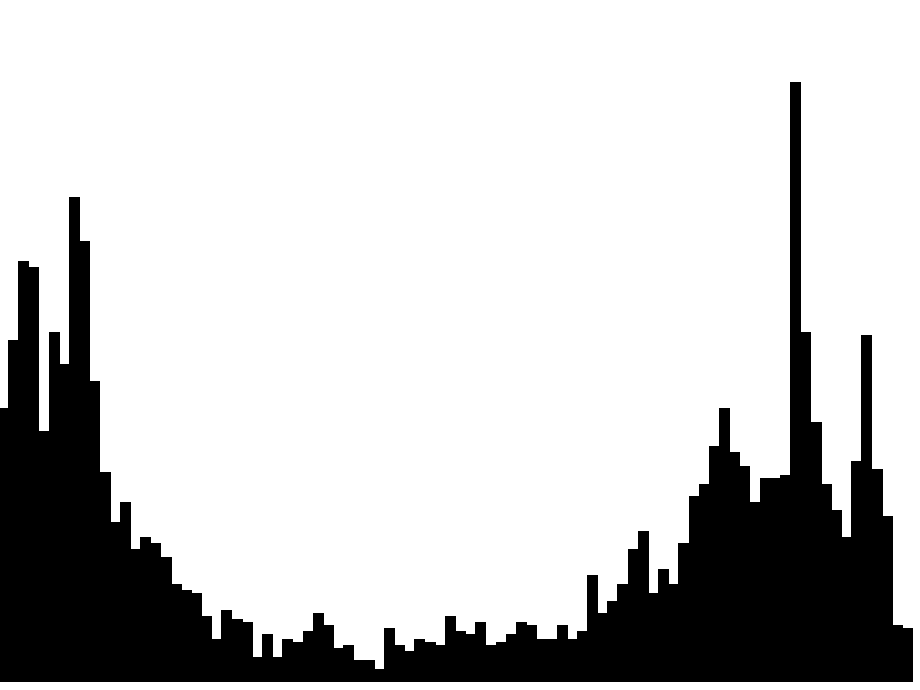} & \ih{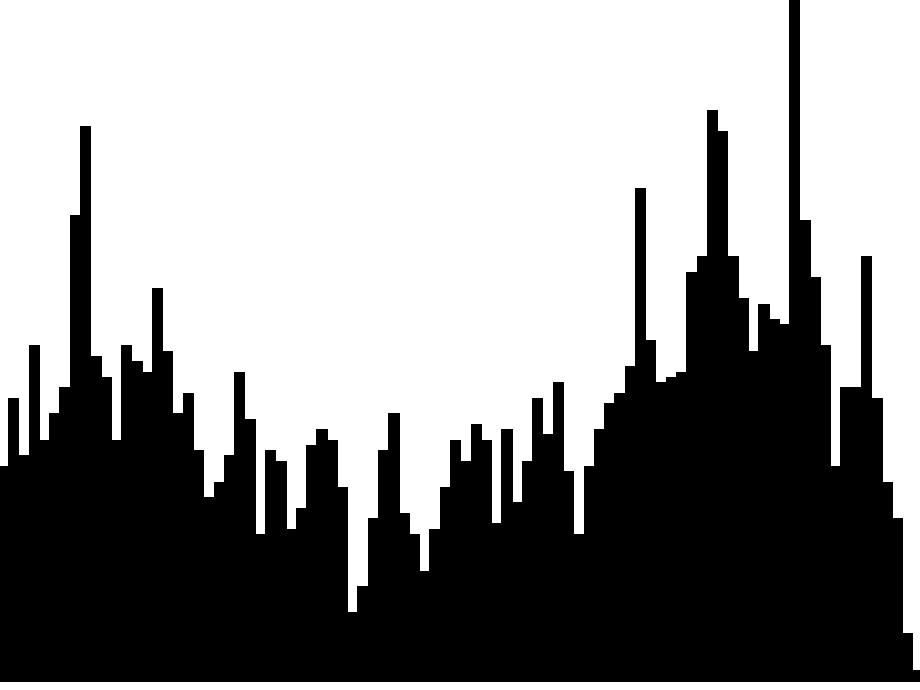} &
\ih{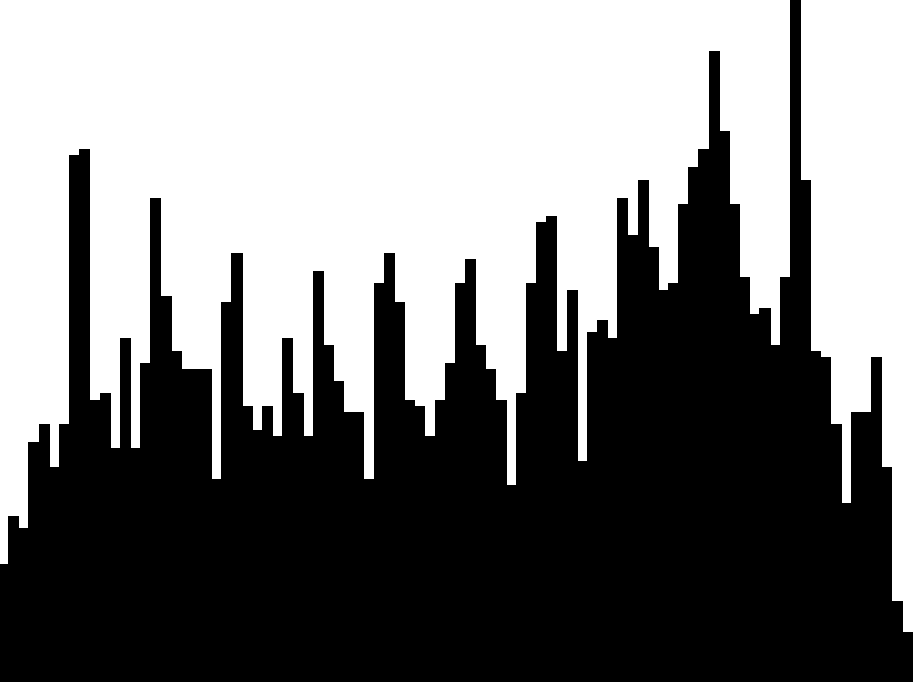} & \ih{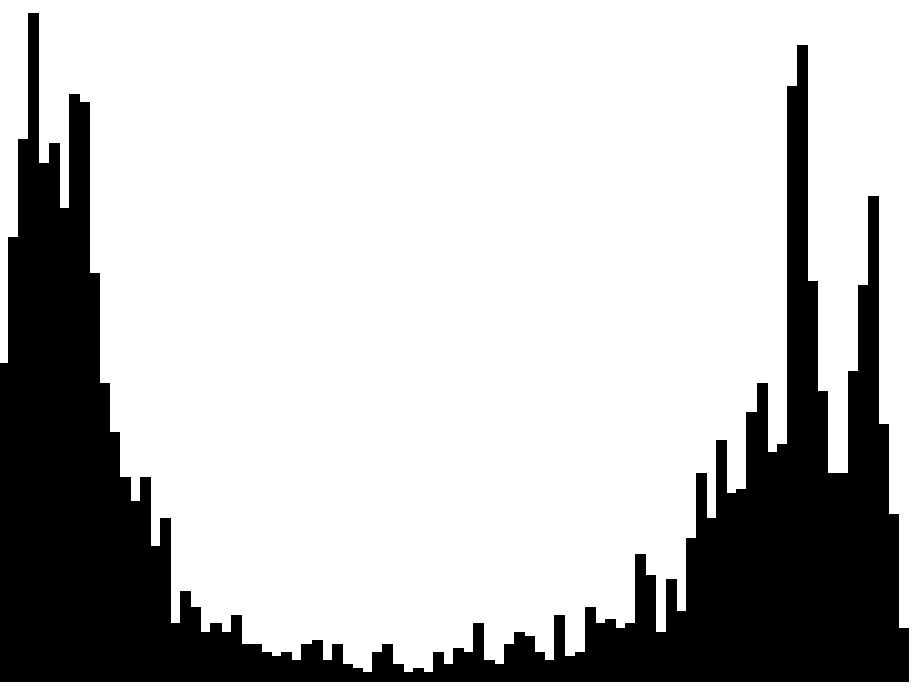} & \ih{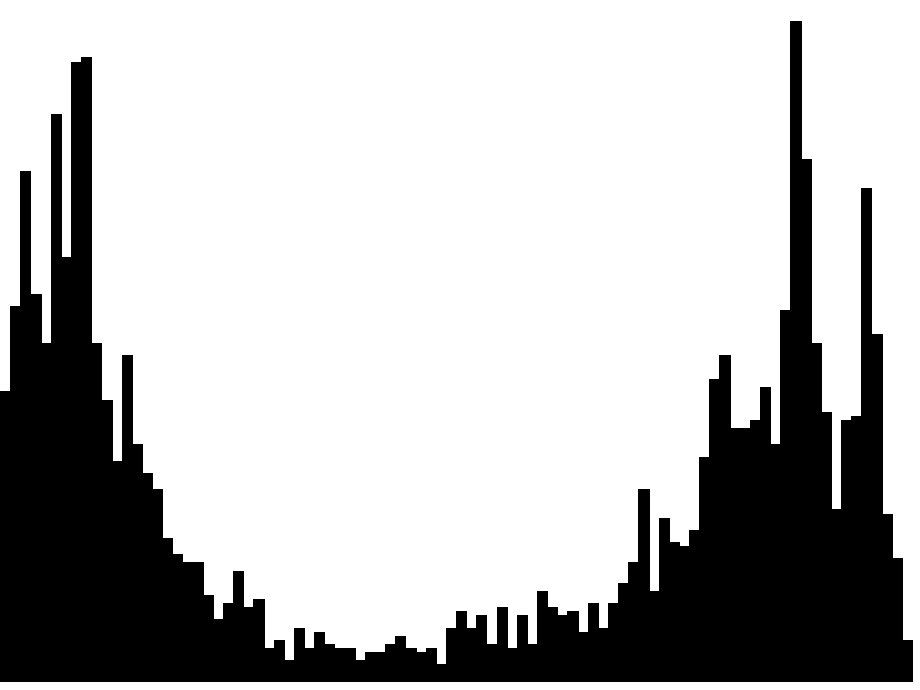} & \ih{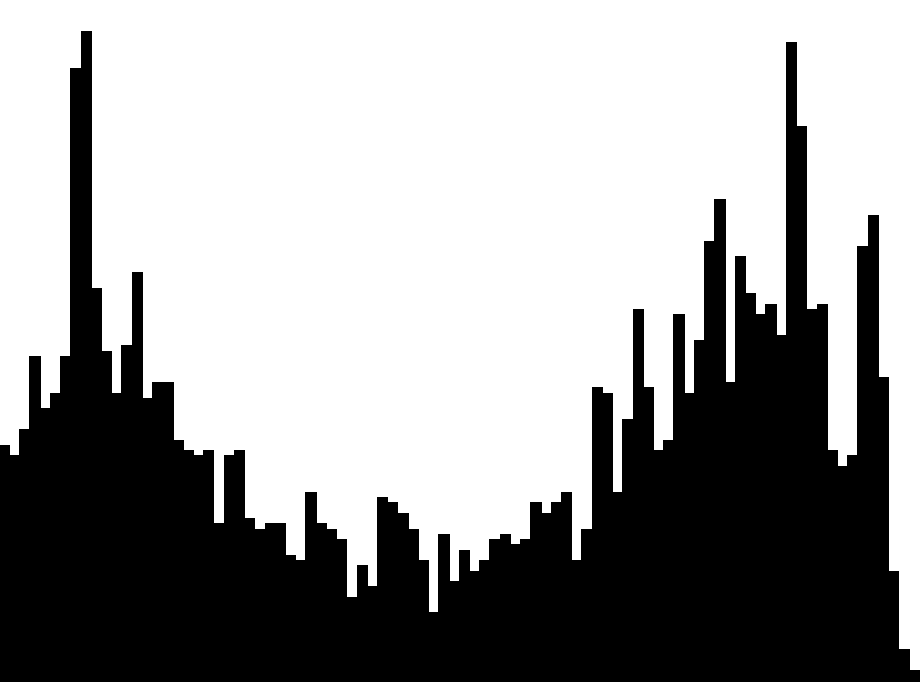} & \ih{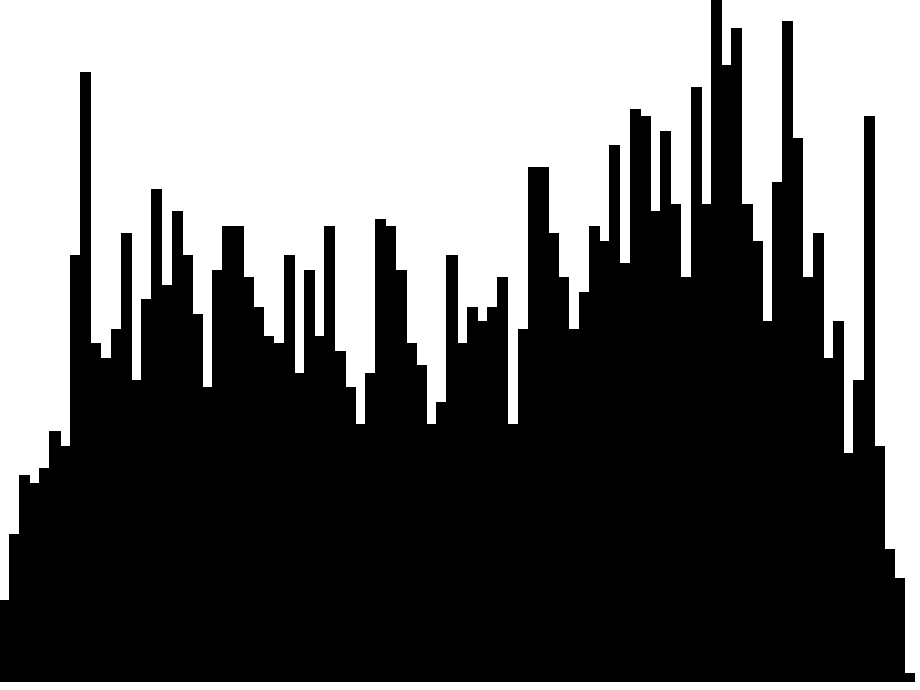} &
\ih{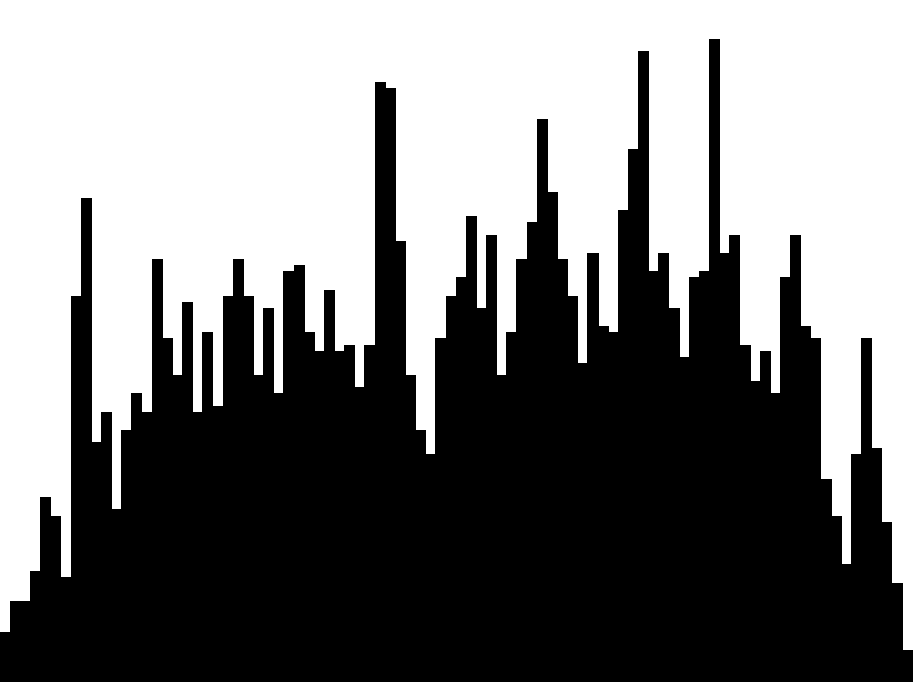} & \ih{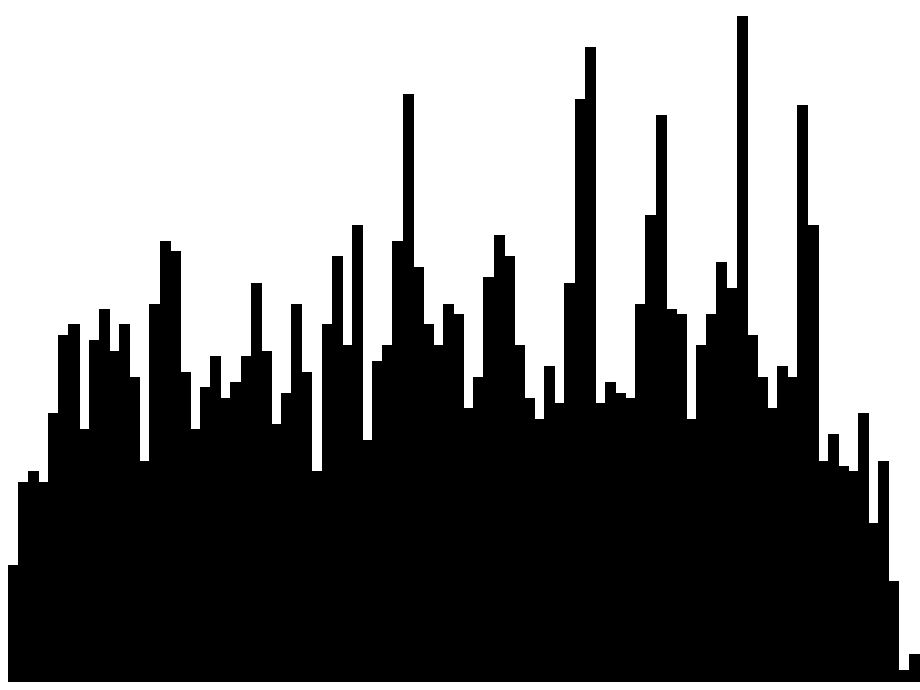} \\
\midrule
\textit{t-mult} &\ih{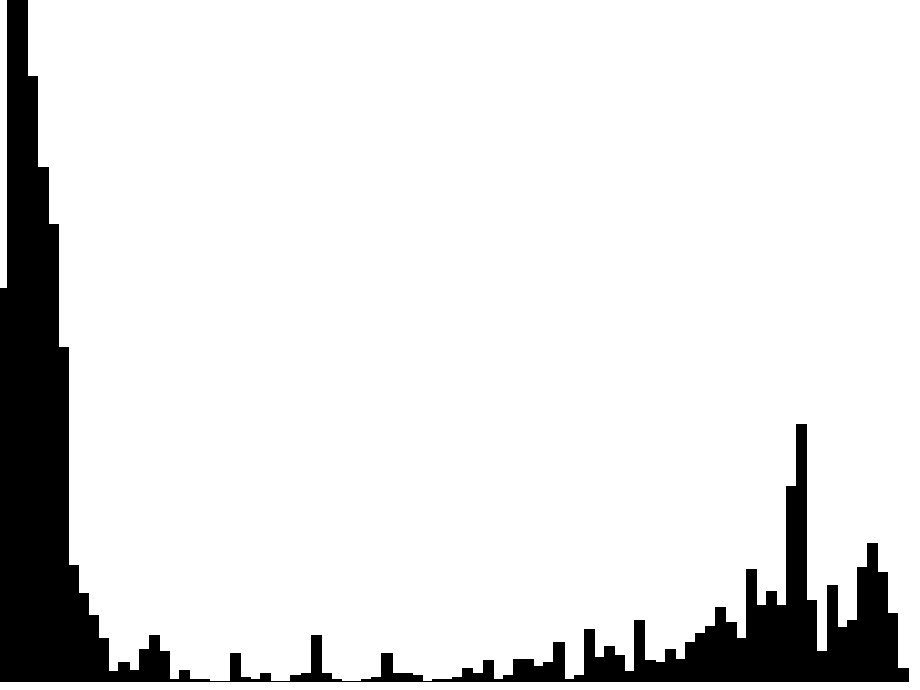}& \ih{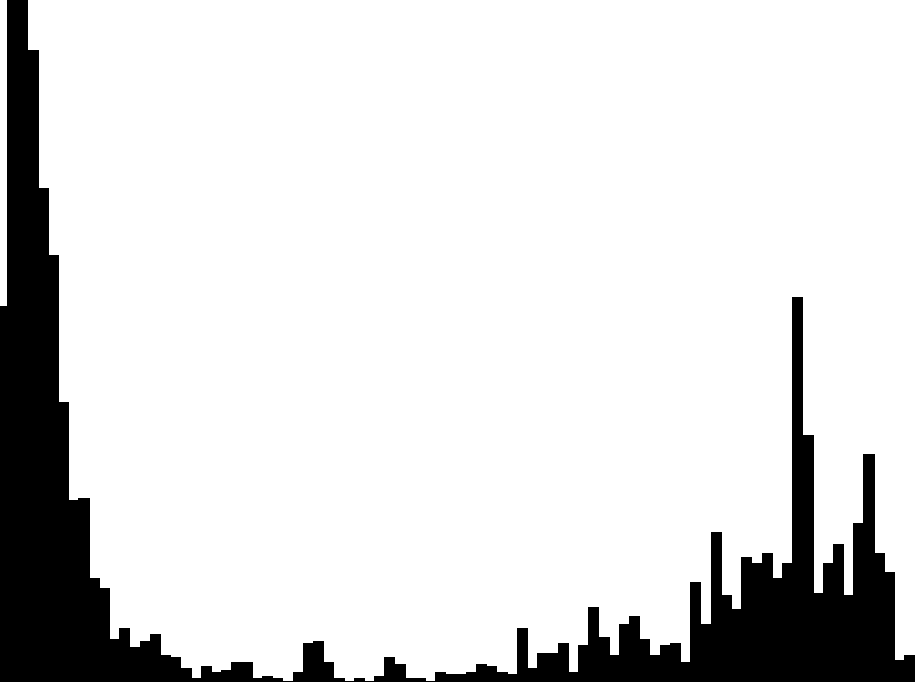} & \ih{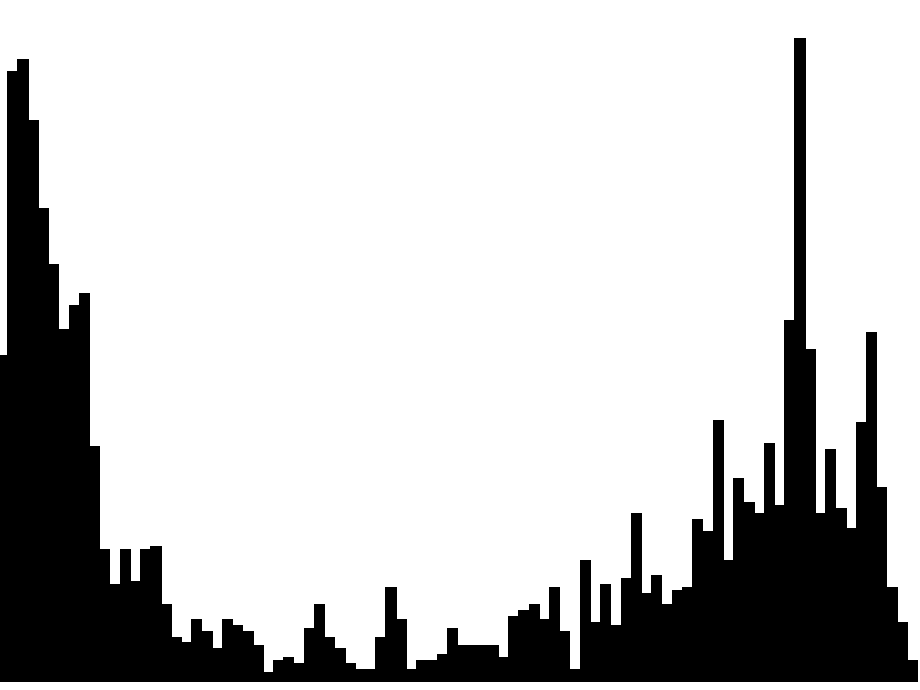} & \ih{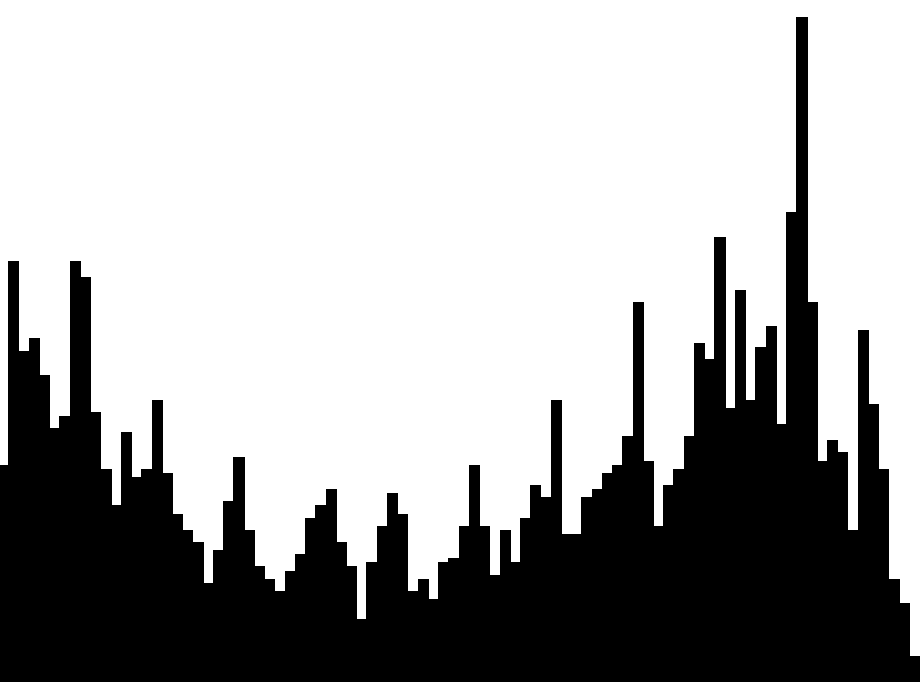} &
\ih{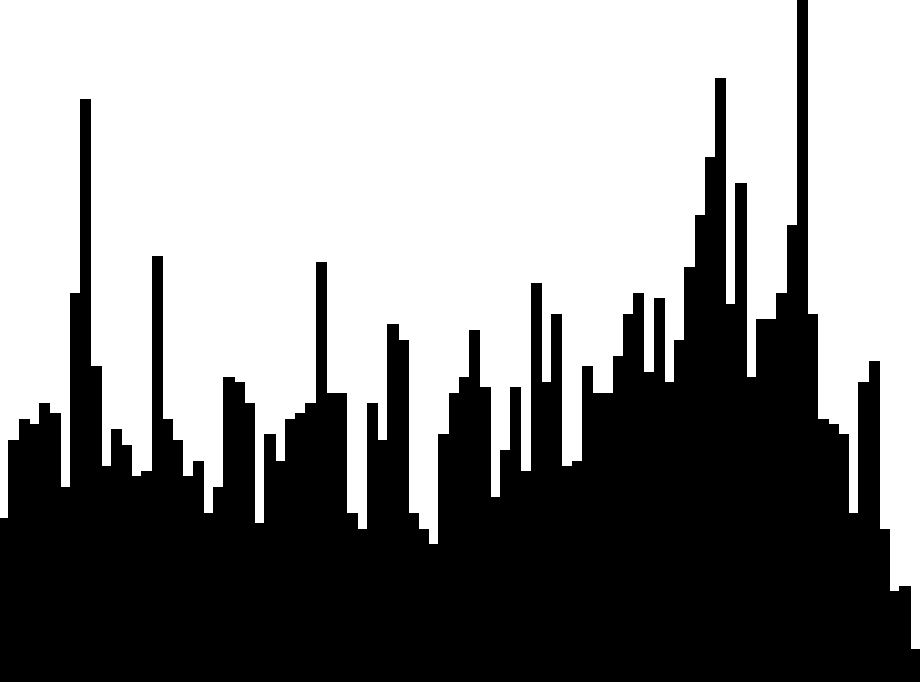} & \ih{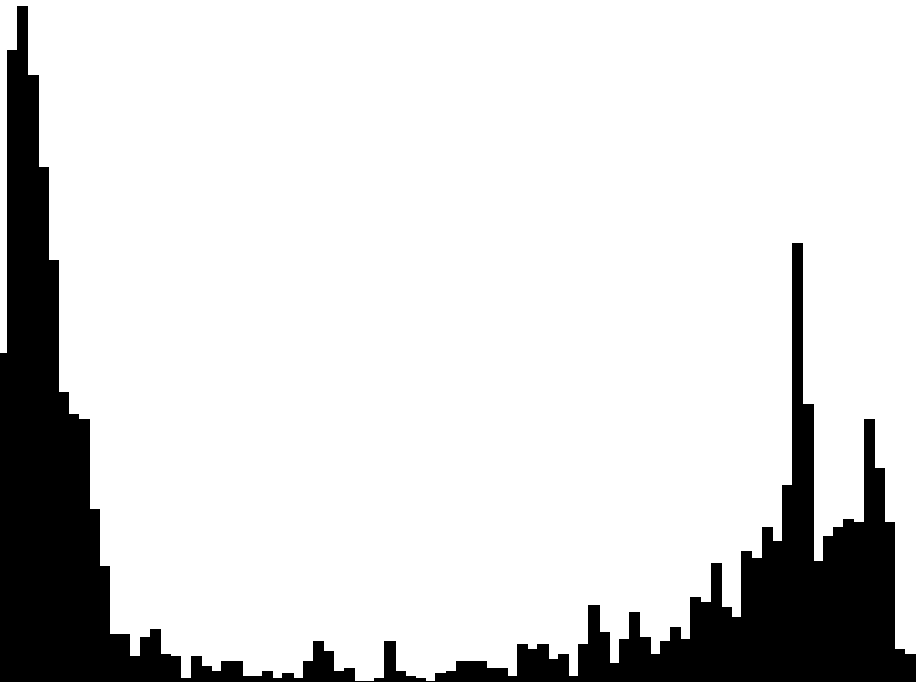} & \ih{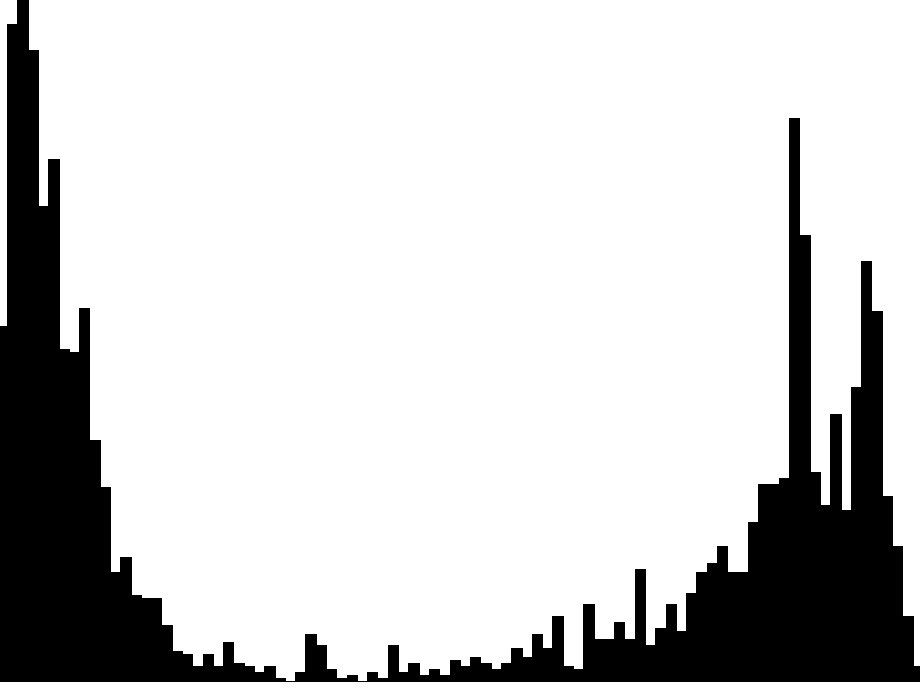} & \ih{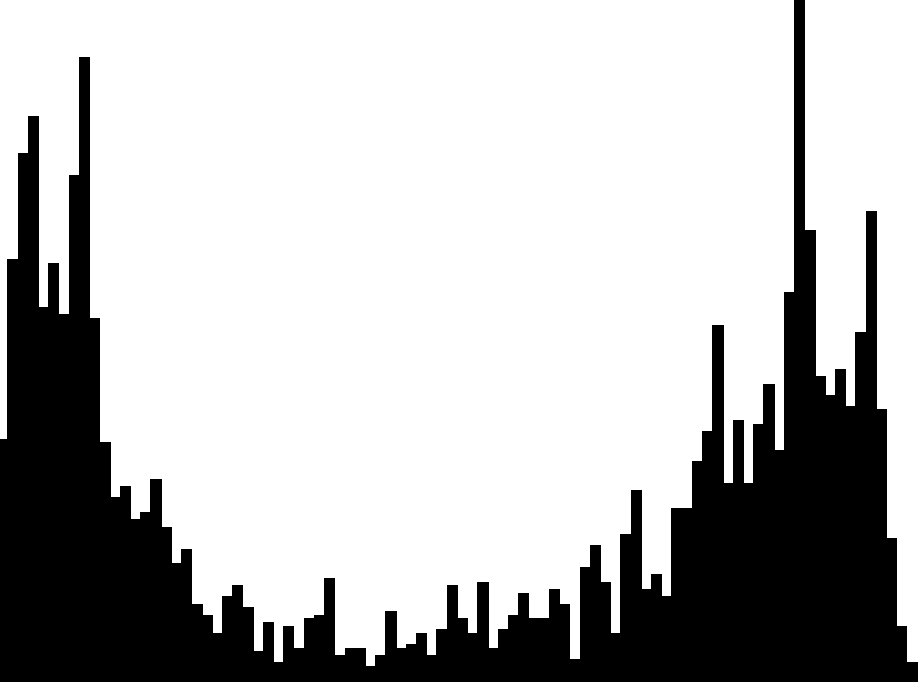} & \ih{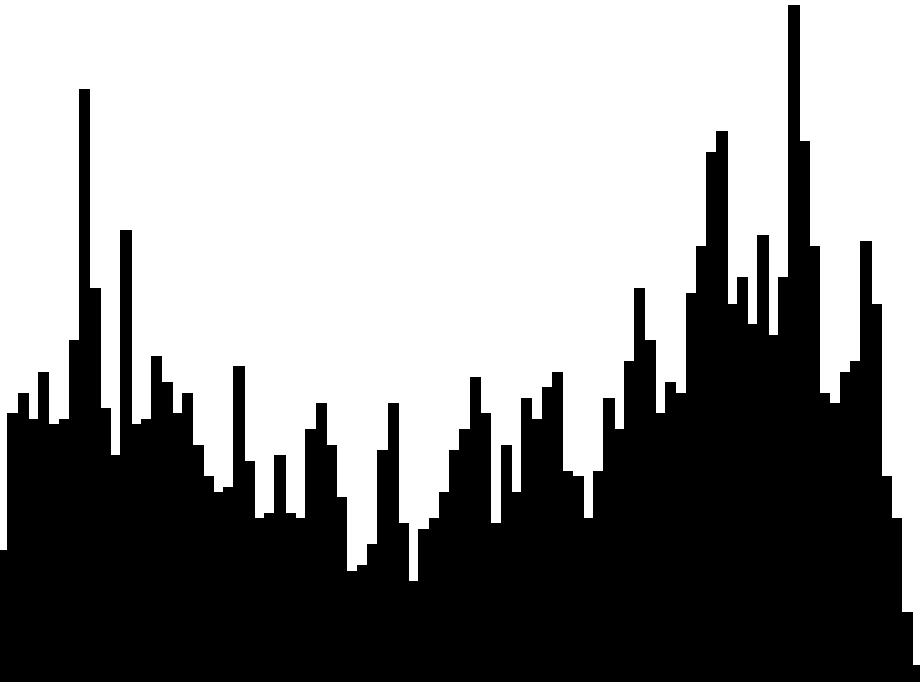} &
\ih{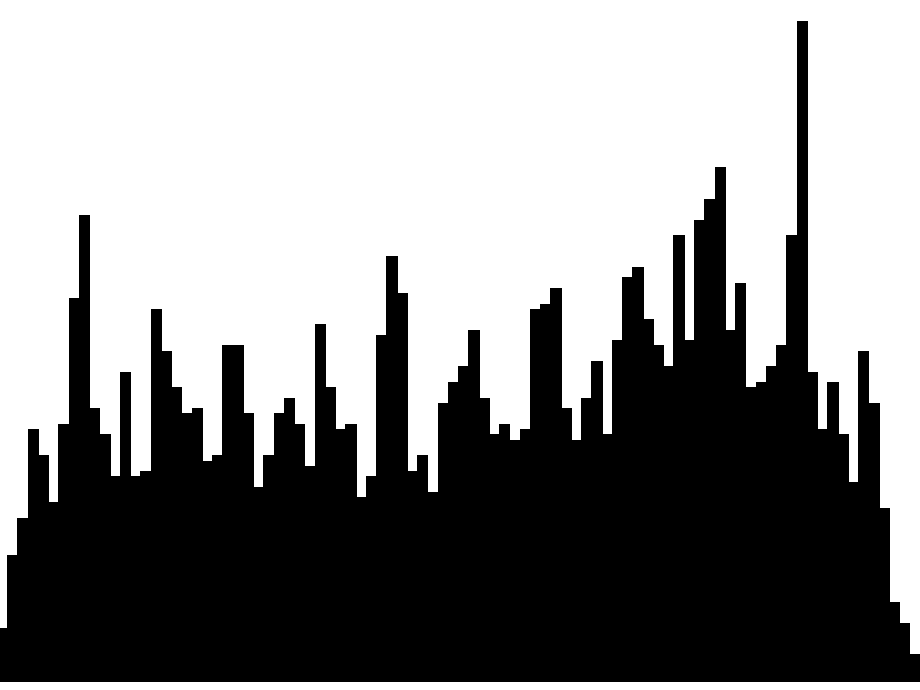} & \ih{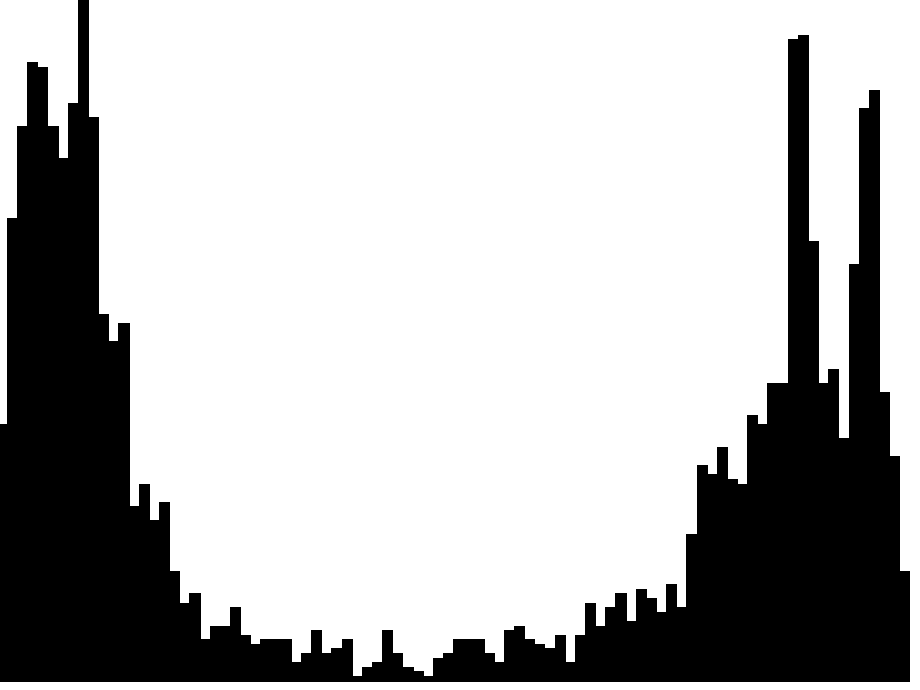} & \ih{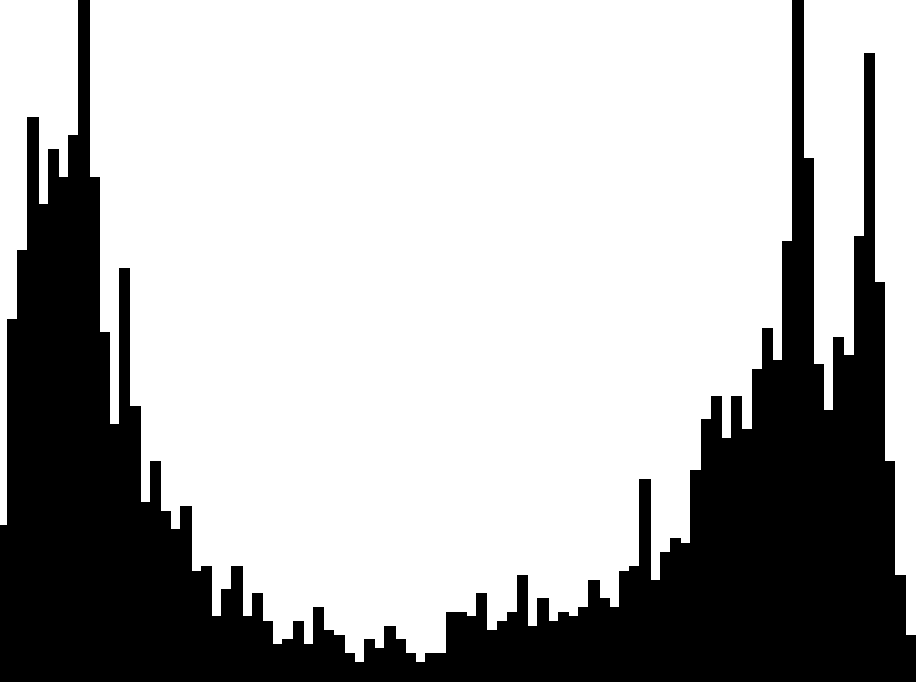} & \ih{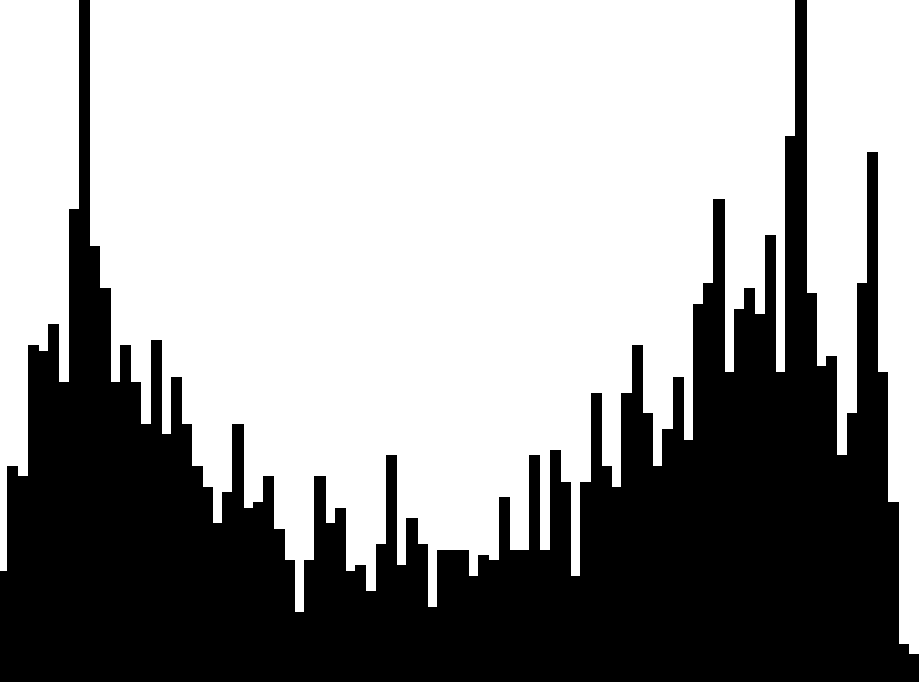} & \ih{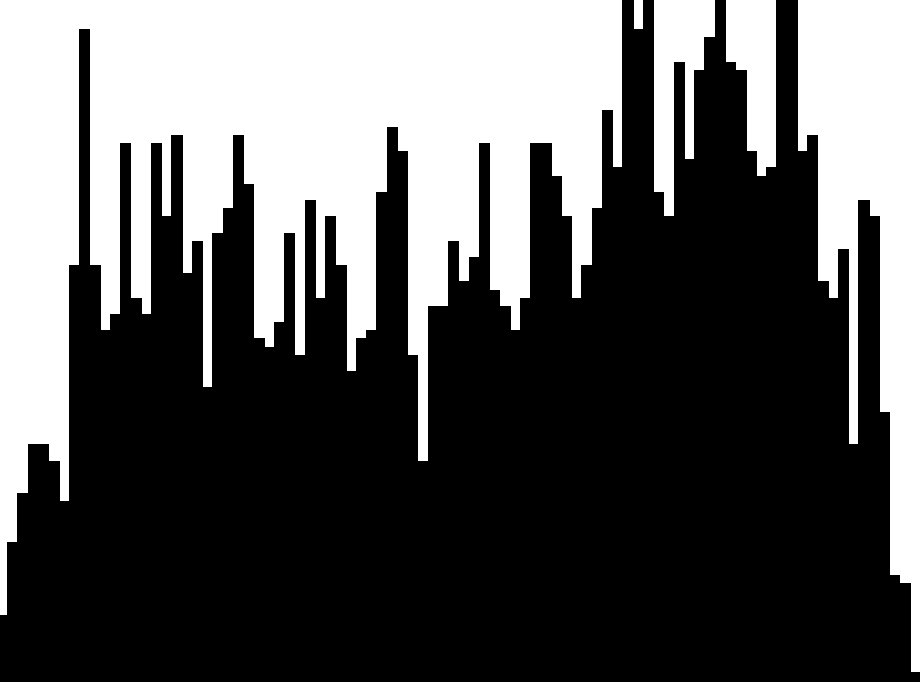} &
\ih{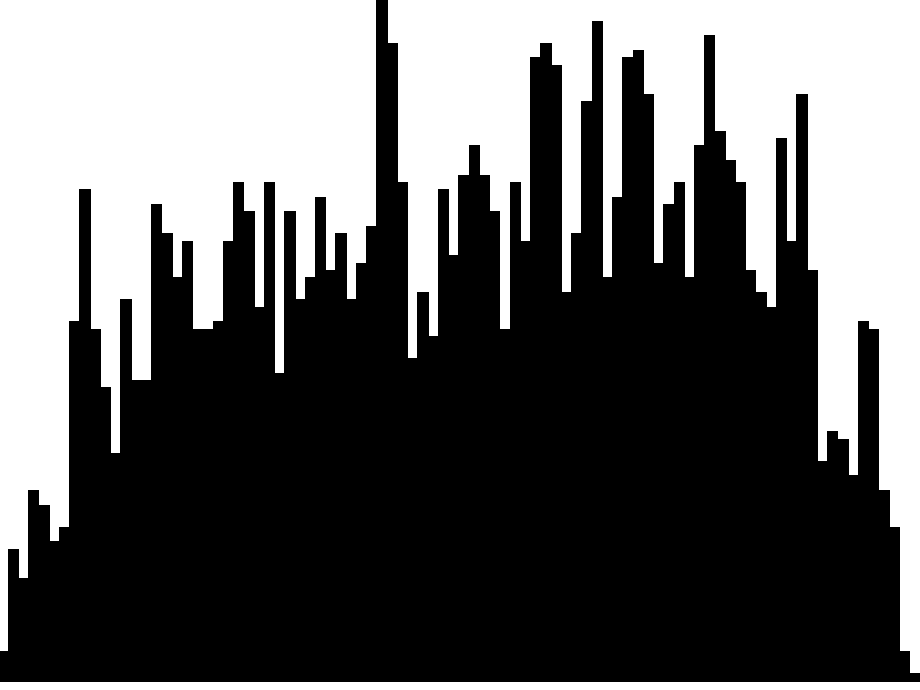} & \ih{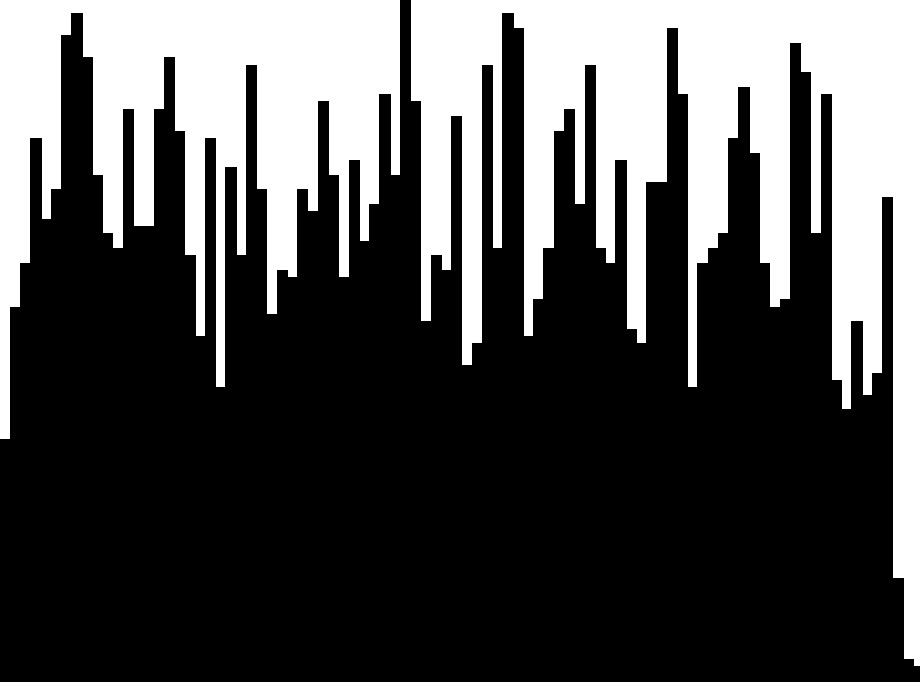} \\
\midrule
\textit{f-mult} &\ih{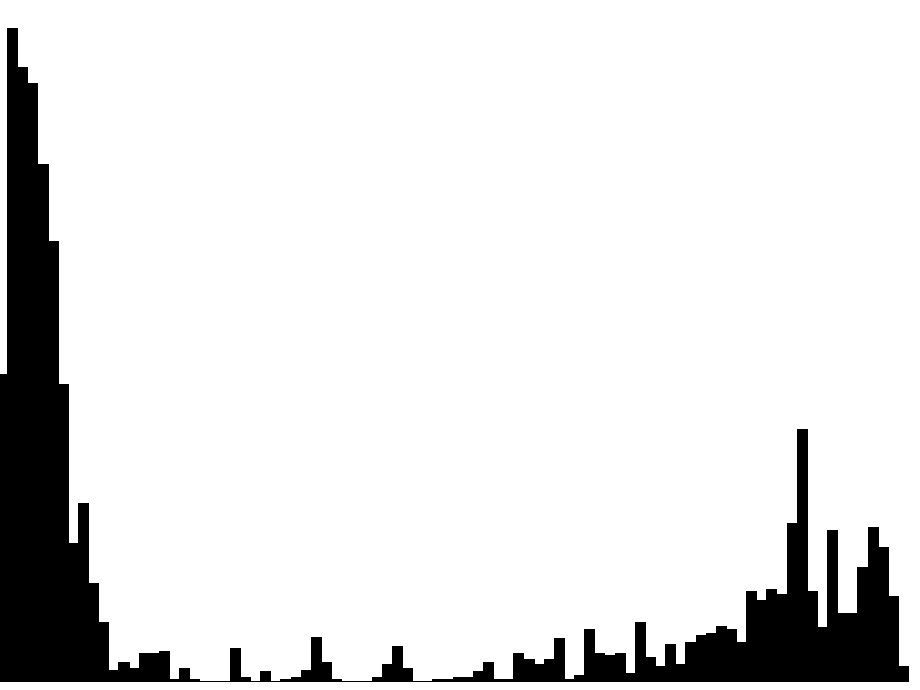}& \ih{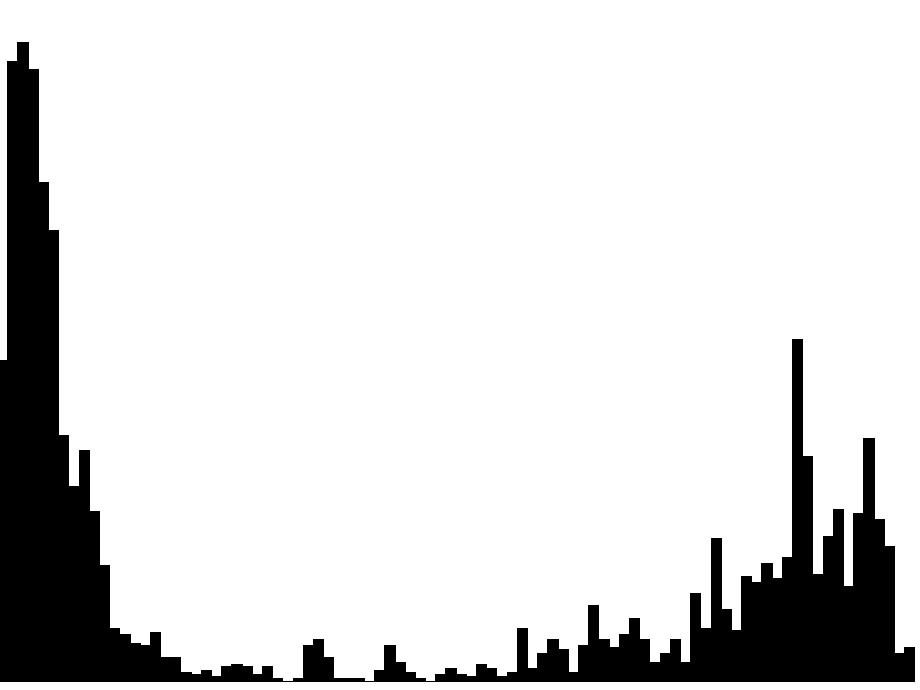} & \ih{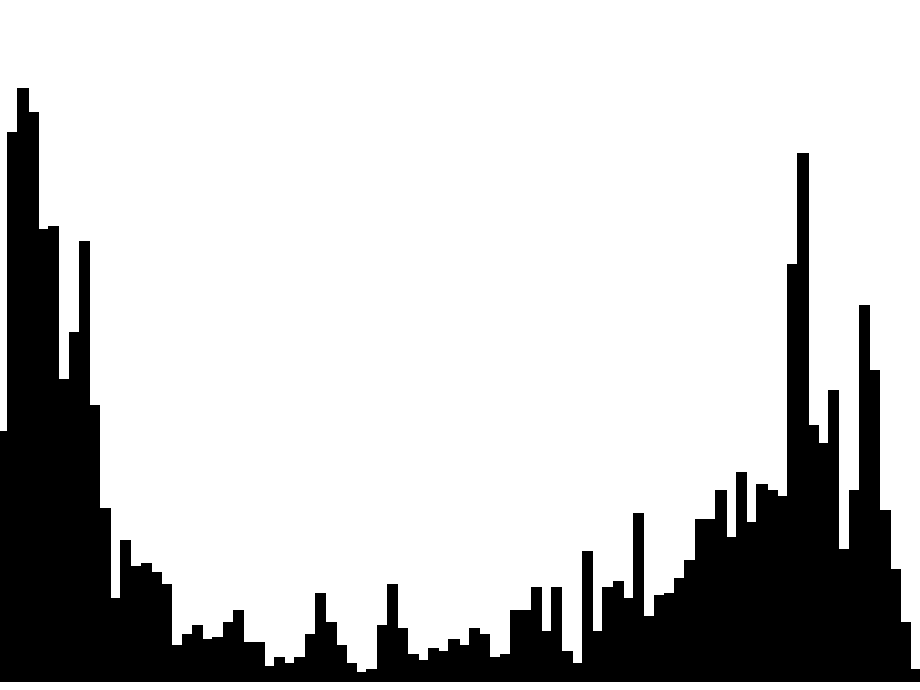} & \ih{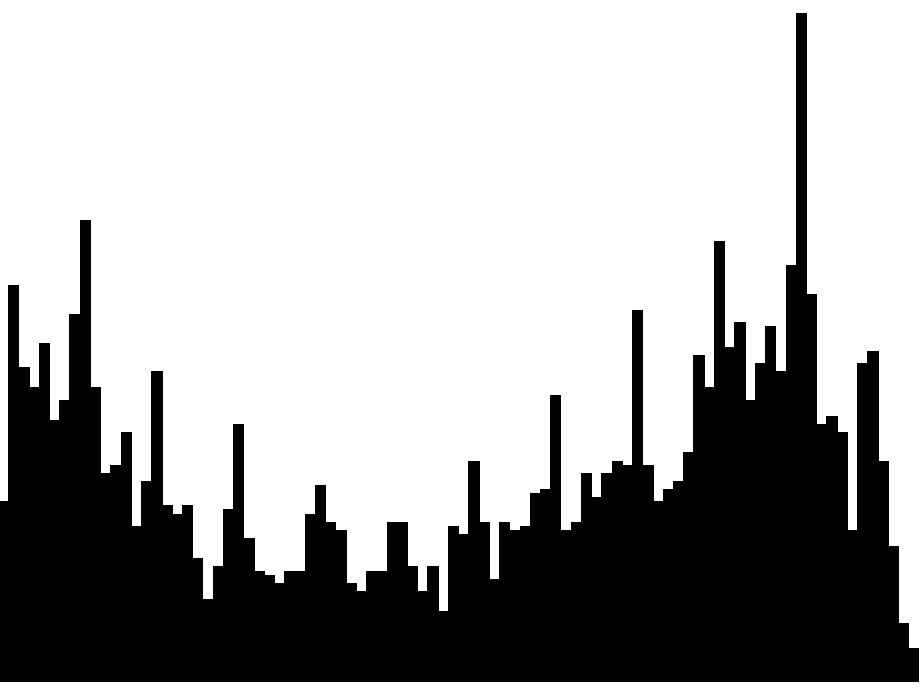} &
\ih{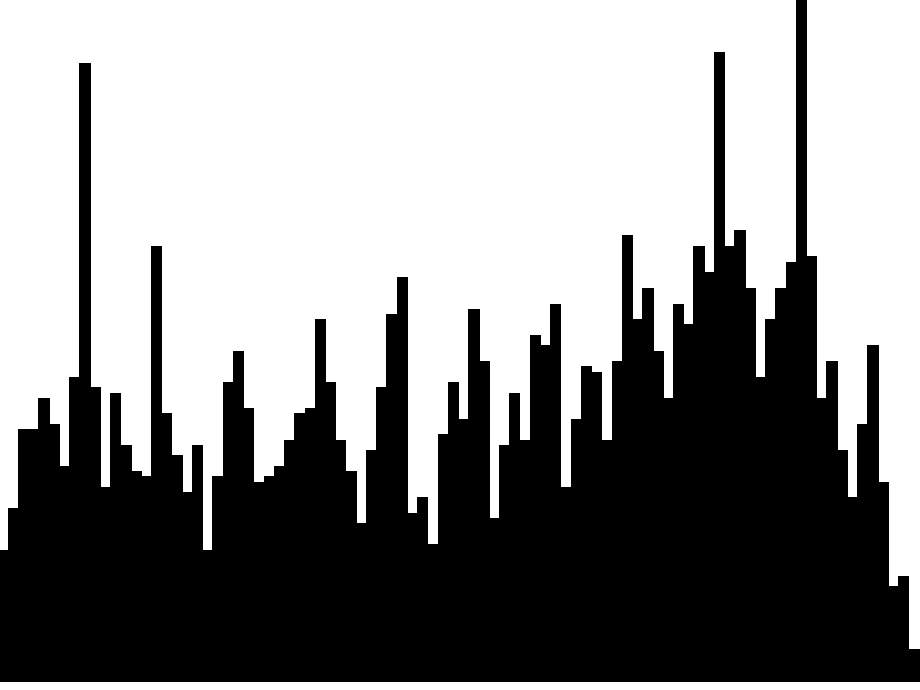} & \ih{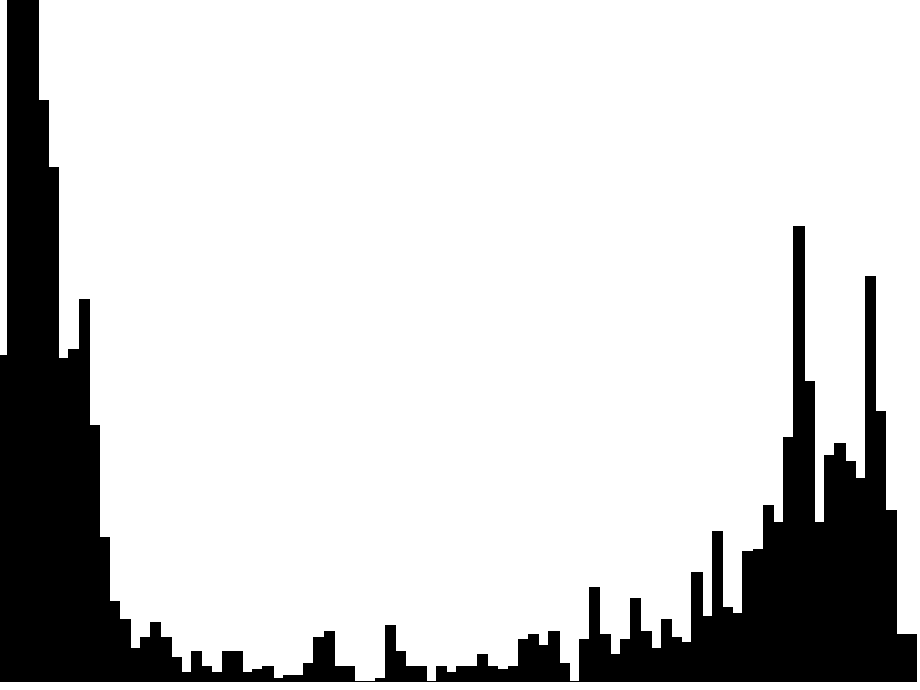} & \ih{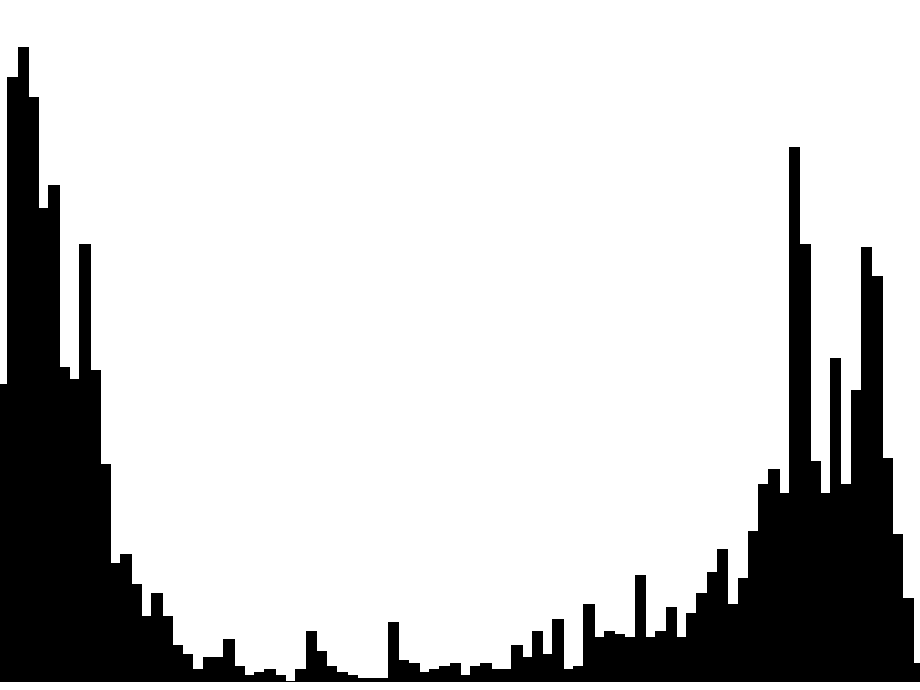} & \ih{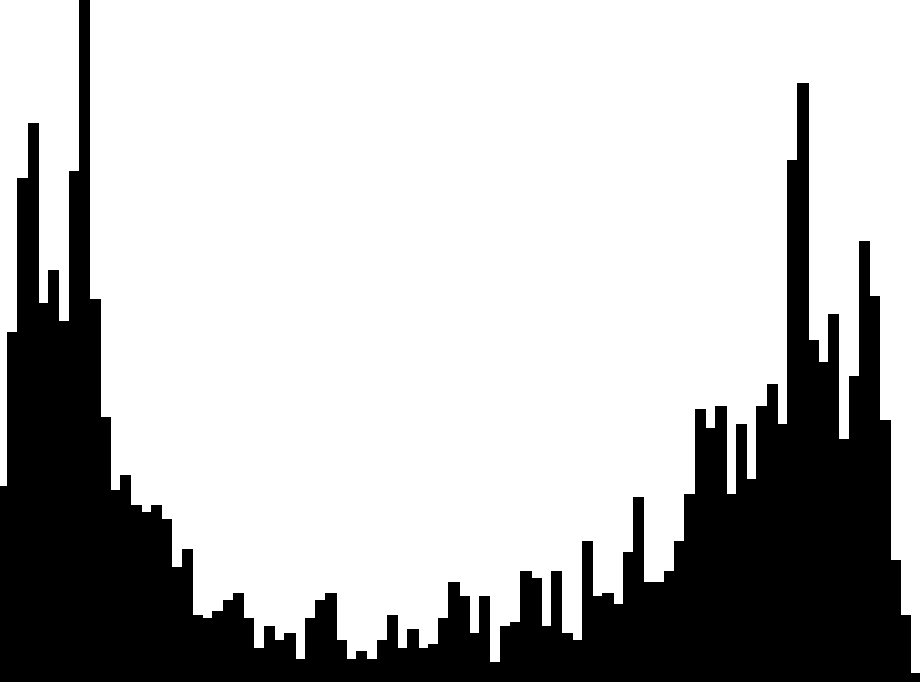} & \ih{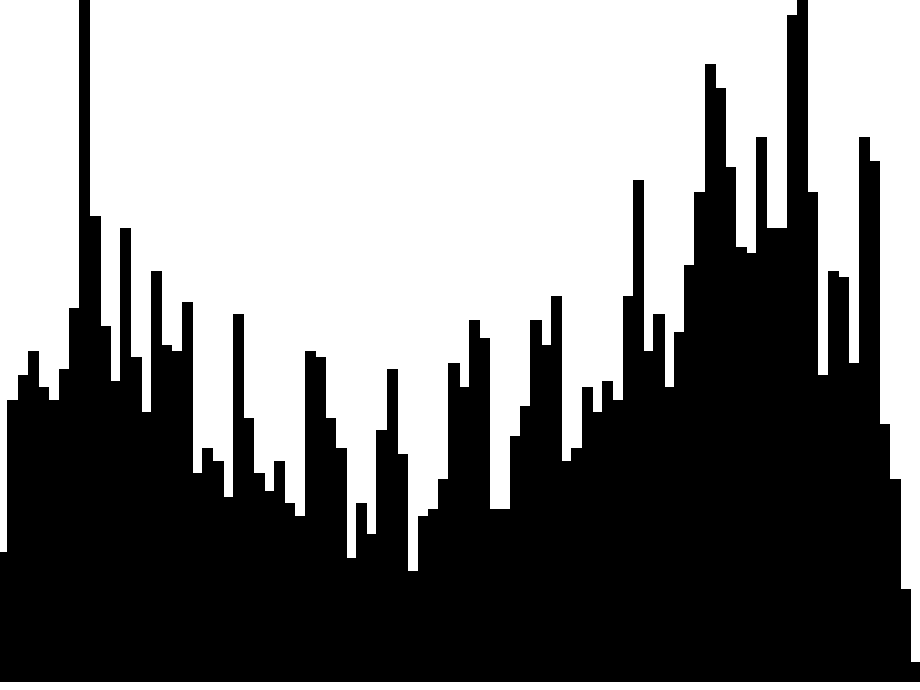} &
\ih{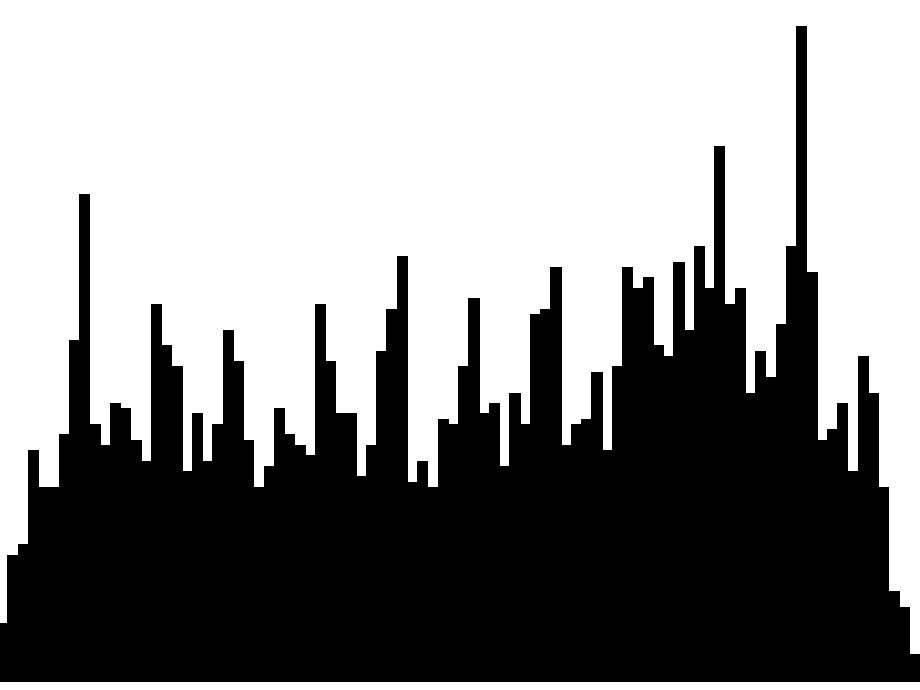} & \ih{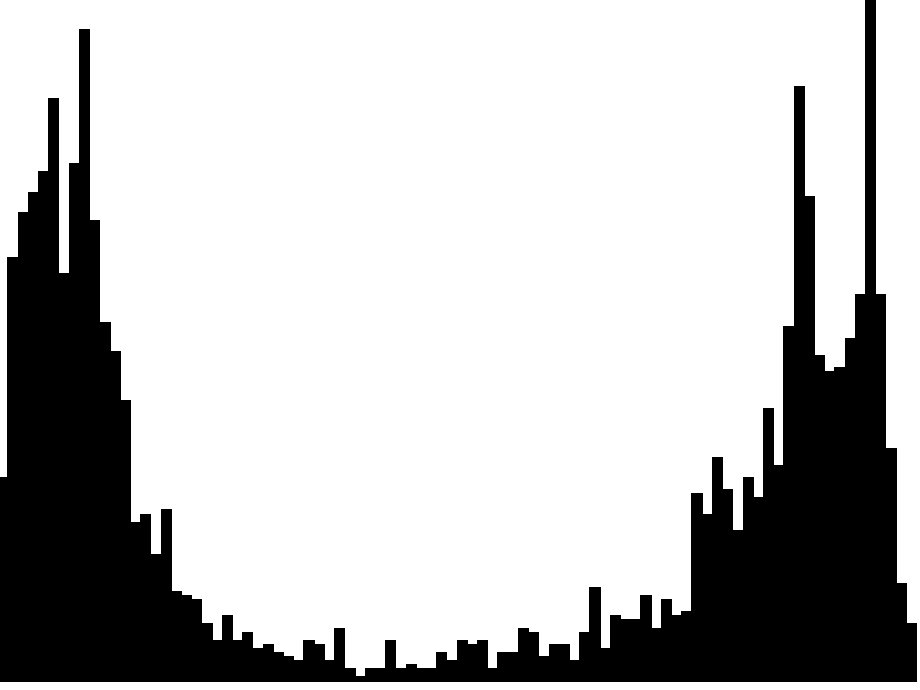} &\ih{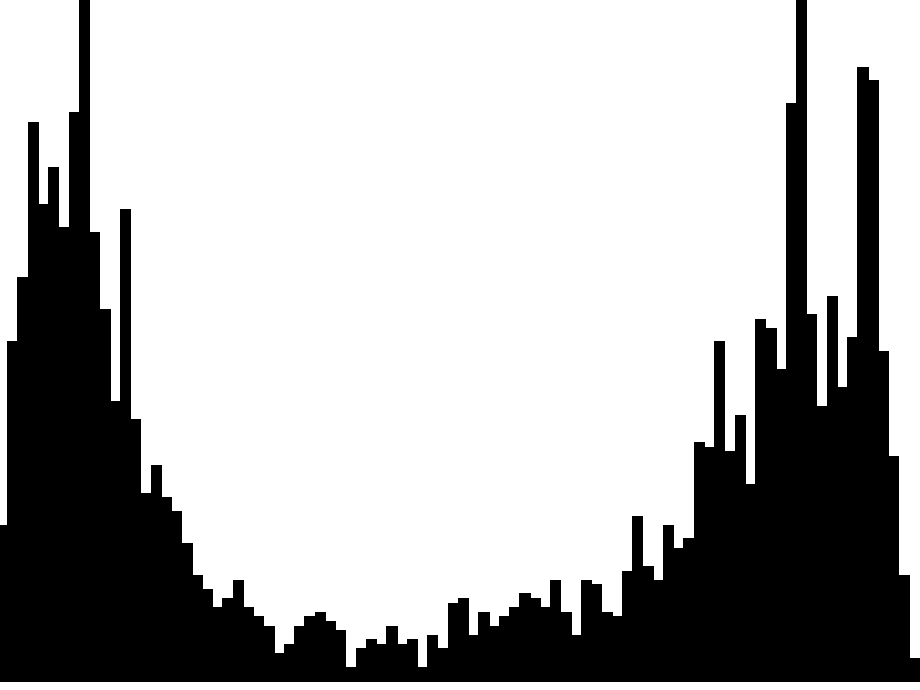} & \ih{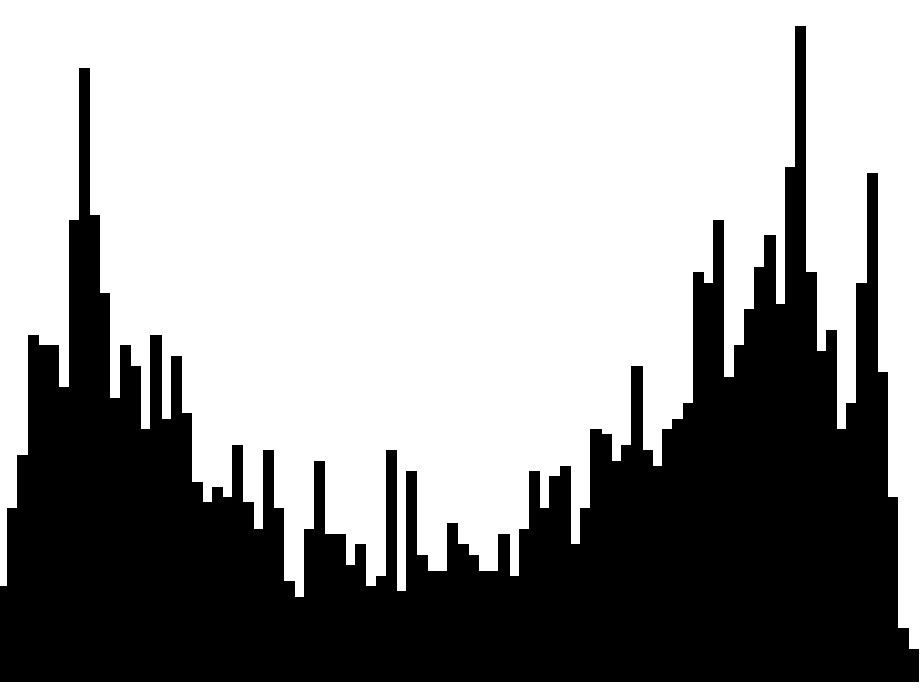} & \ih{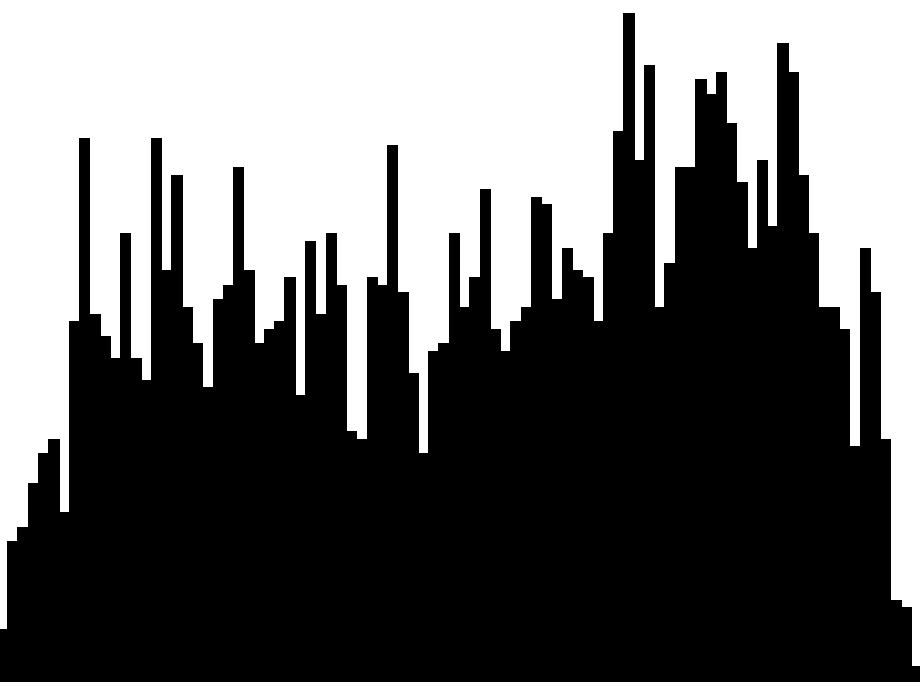} &
\ih{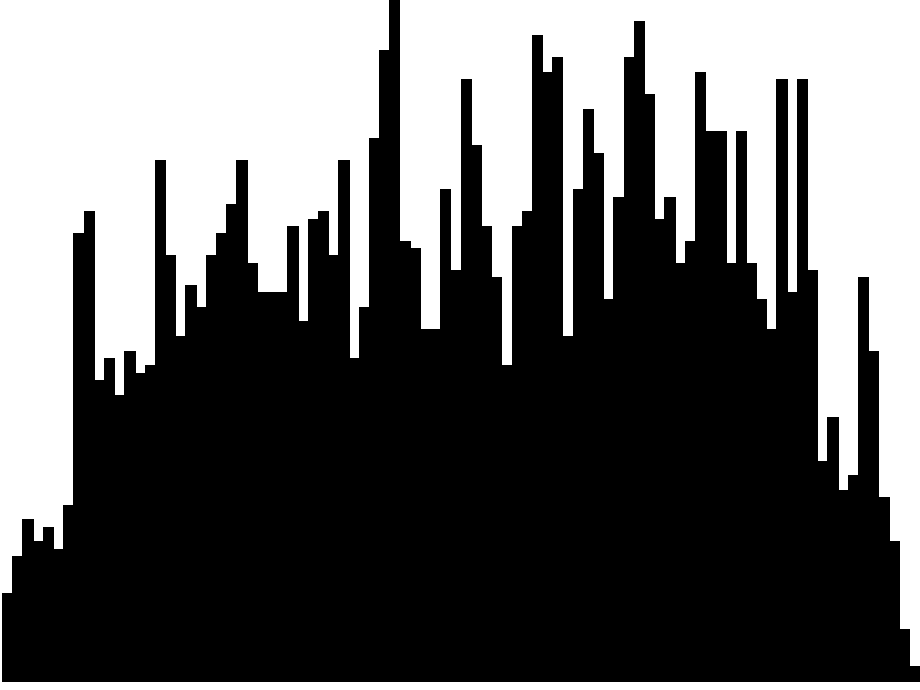} & \ih{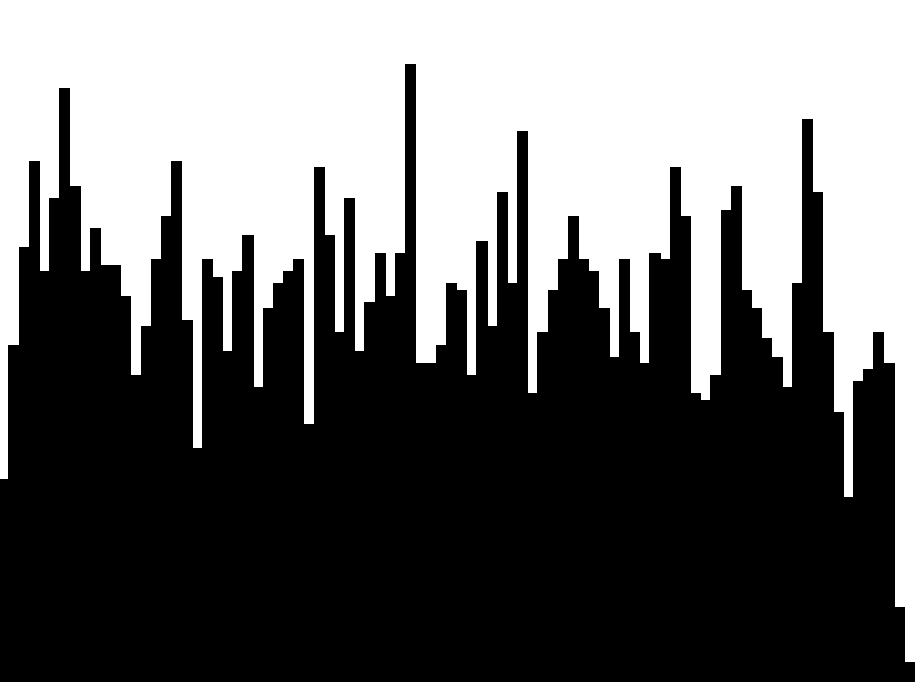} \\
\midrule
\textit{pi} &\ih{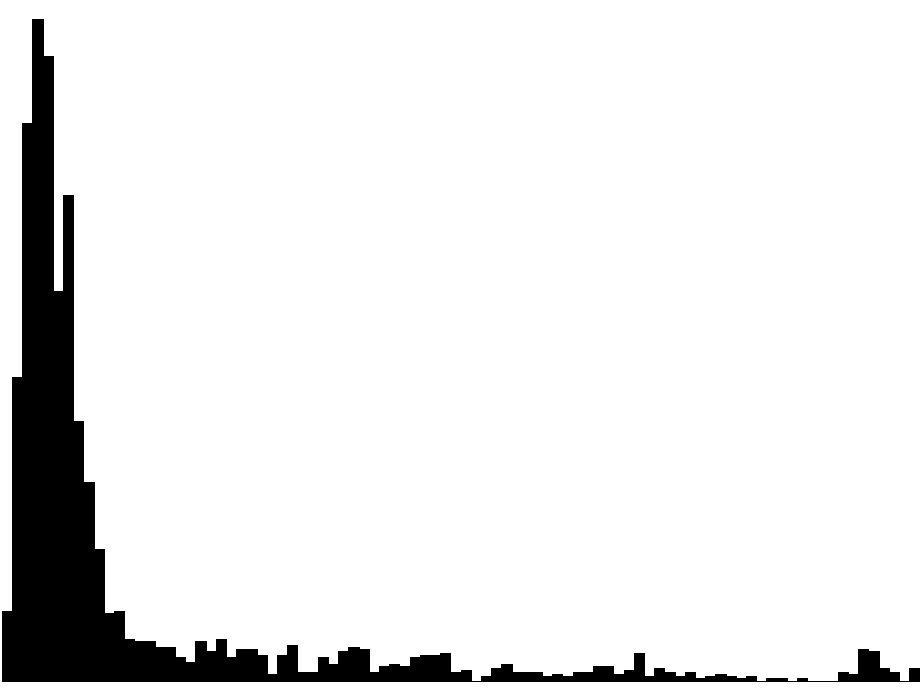}& \ih{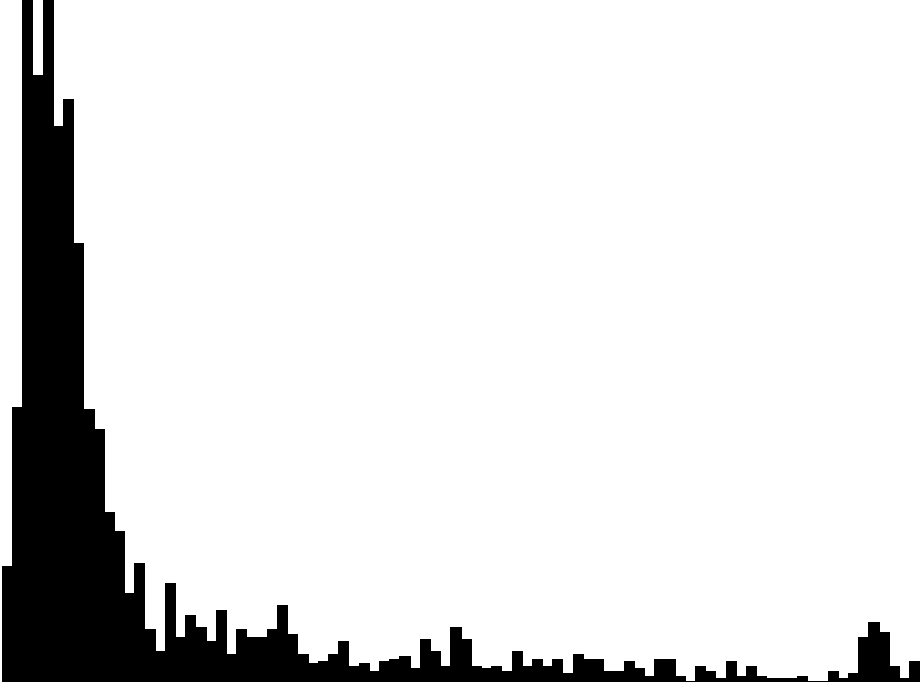} & \ih{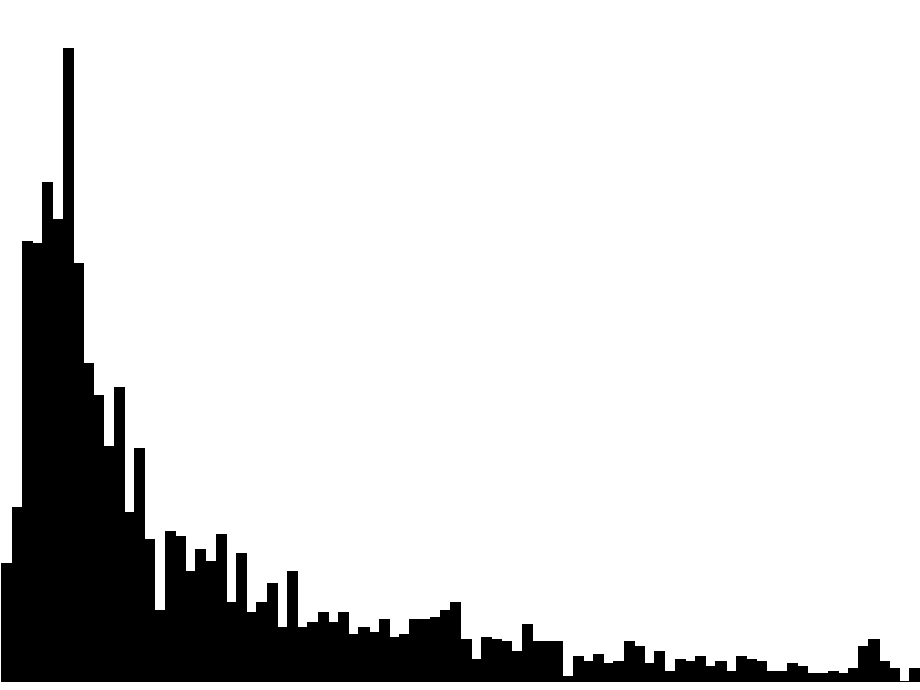} & \ih{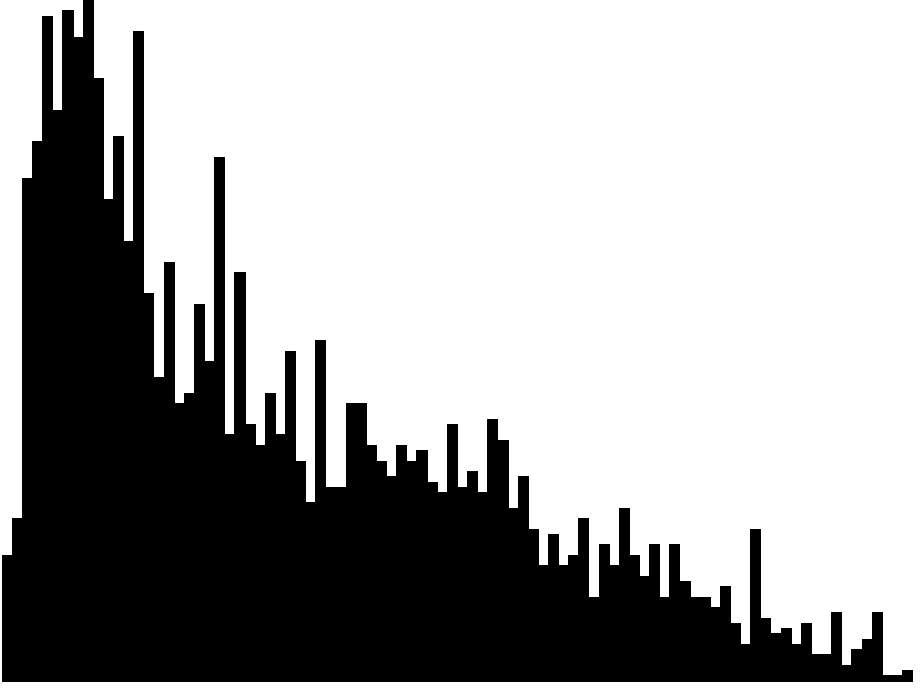} & \ih{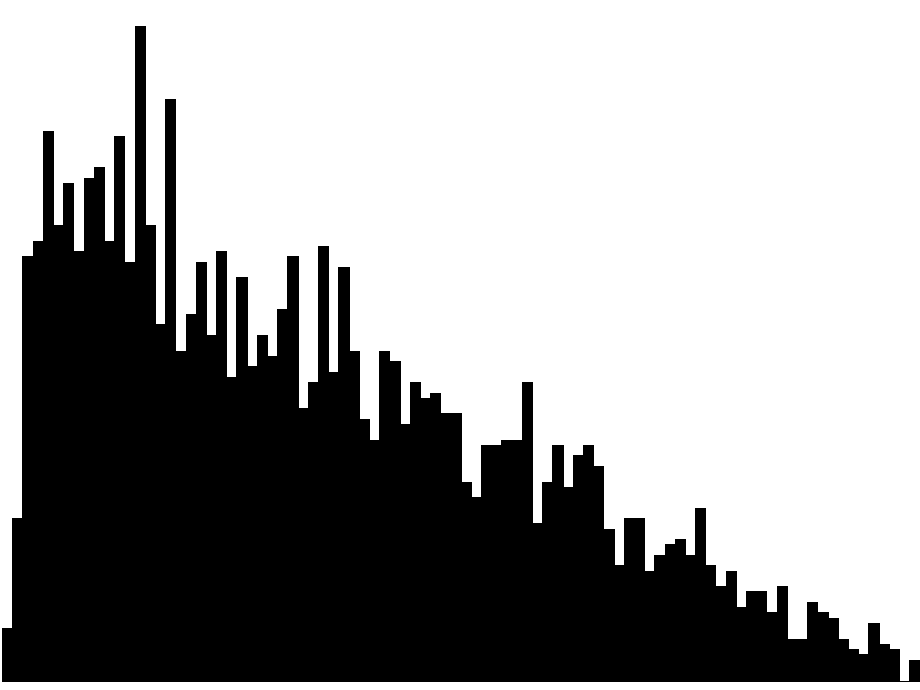} &
\ih{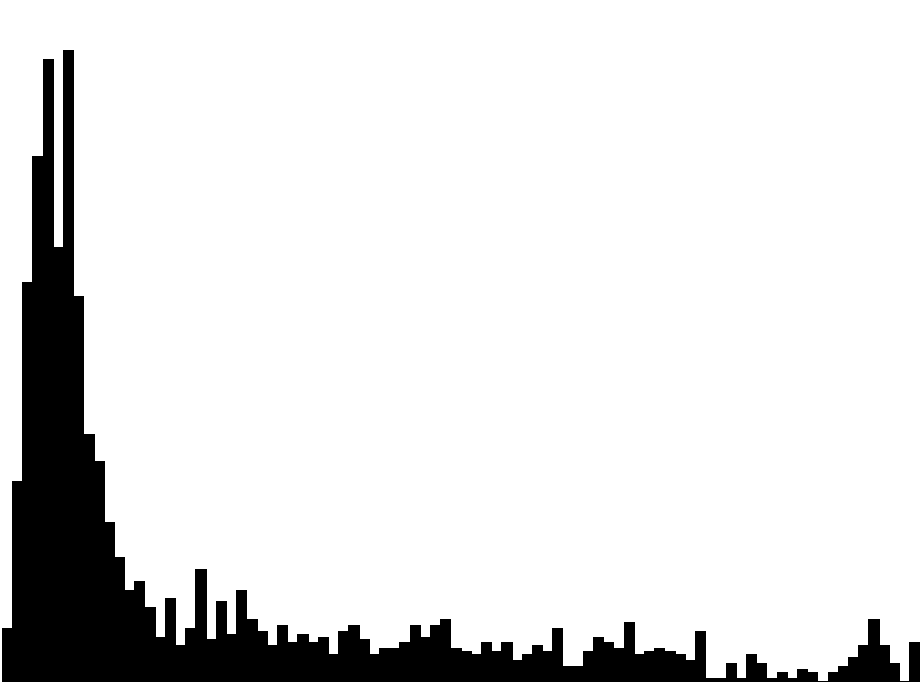} & \ih{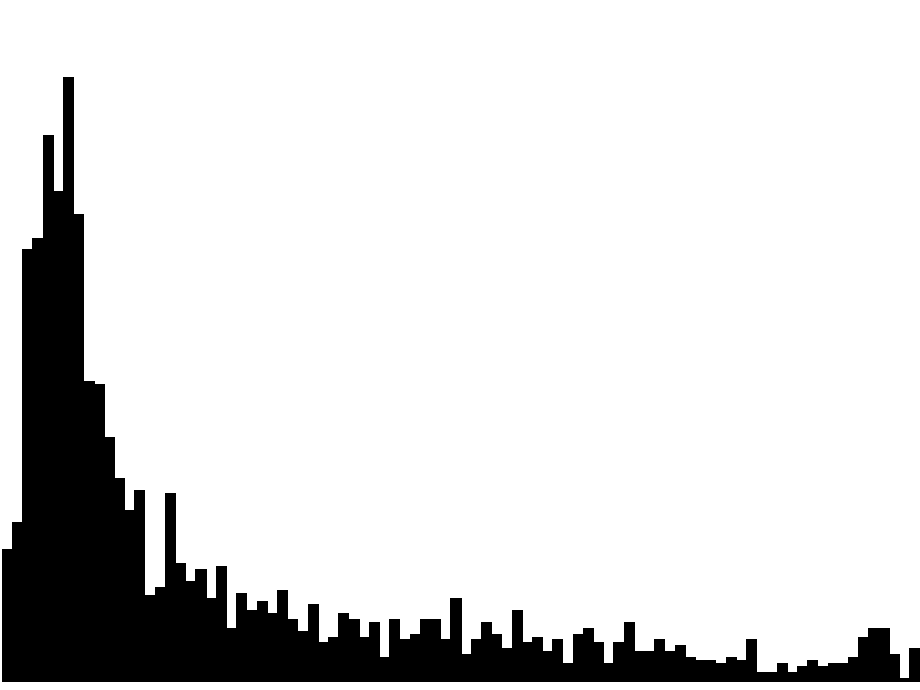} & \ih{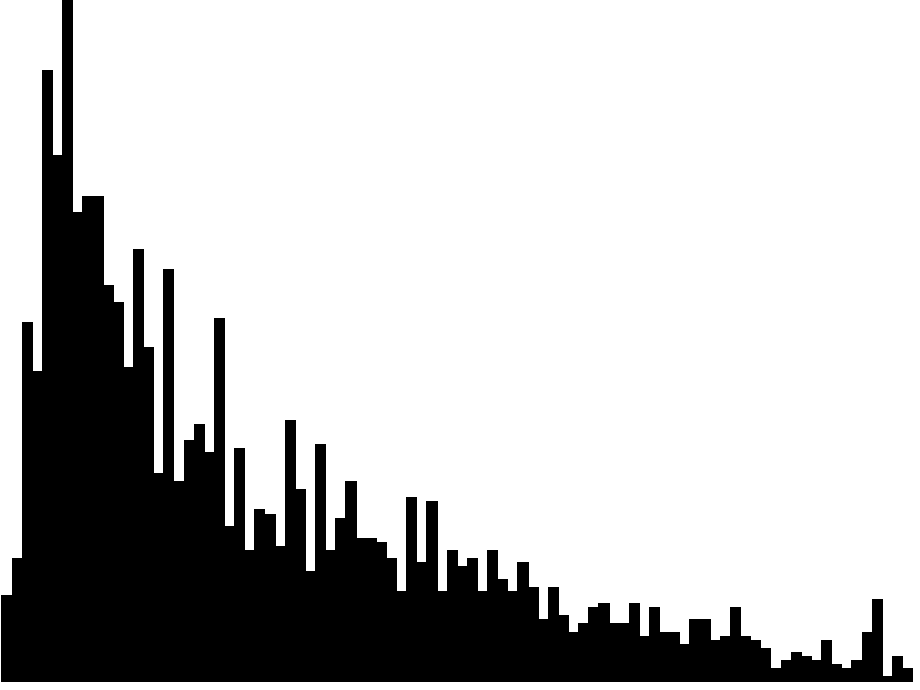} & \ih{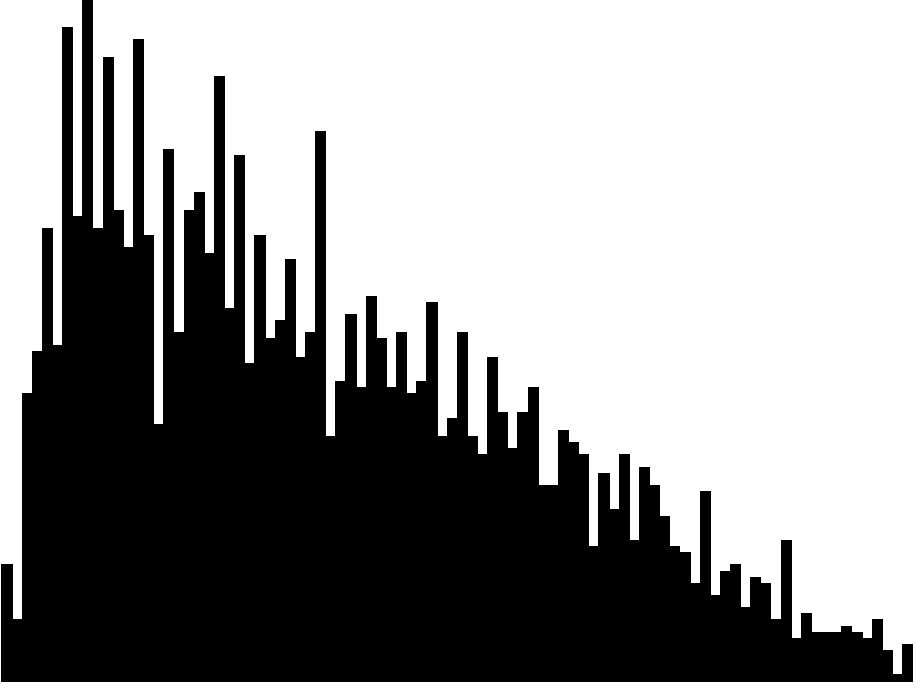} & \ih{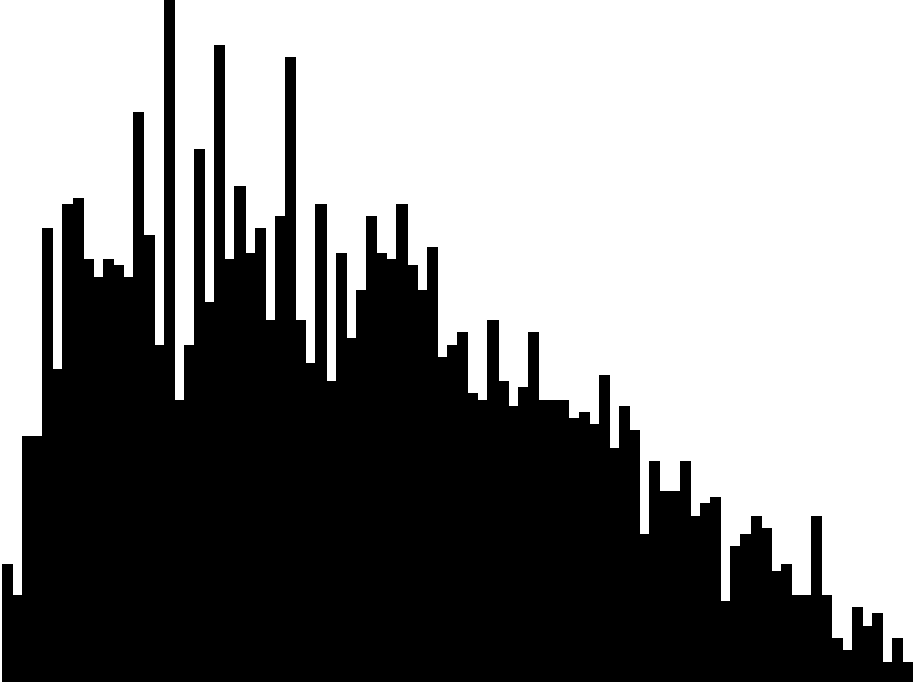} & \ih{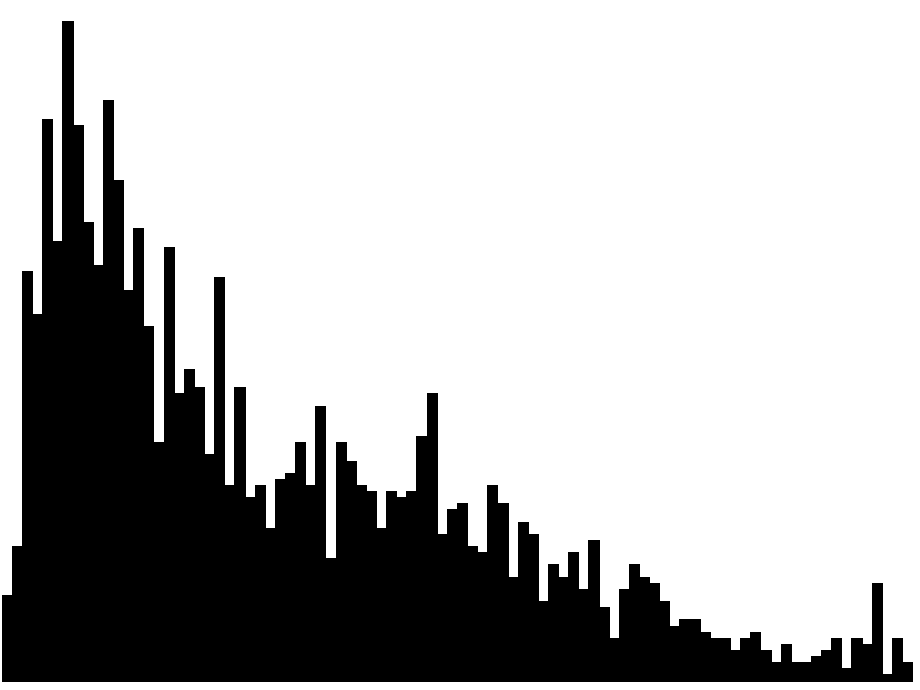} &
\ih{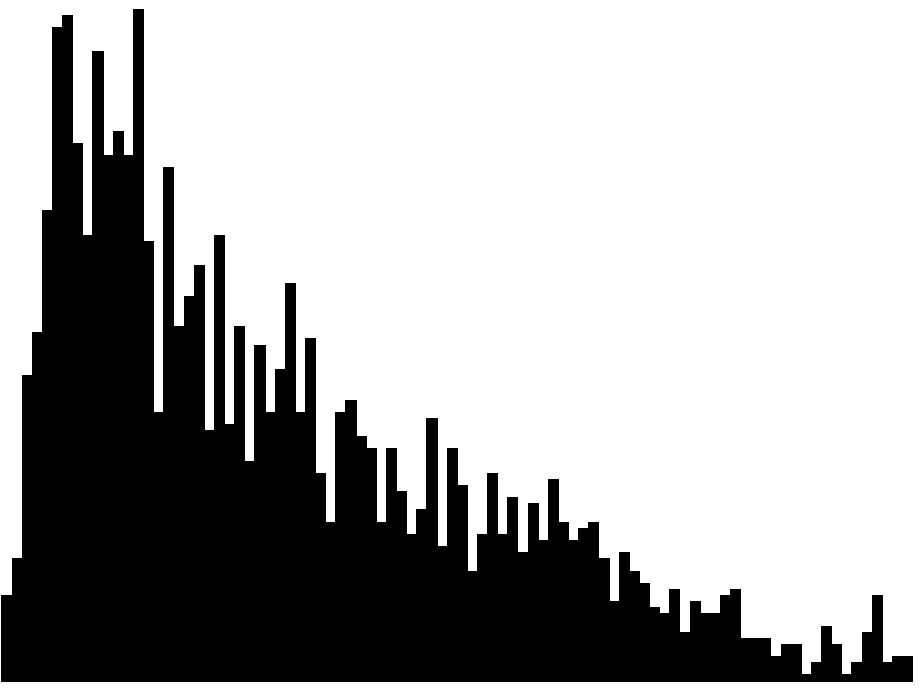} & \ih{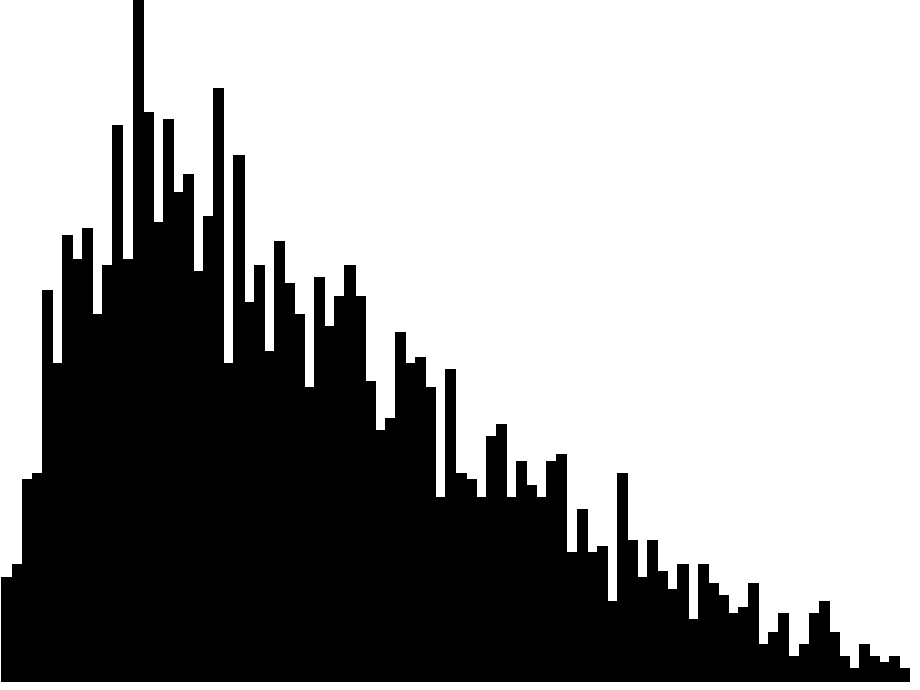} & \ih{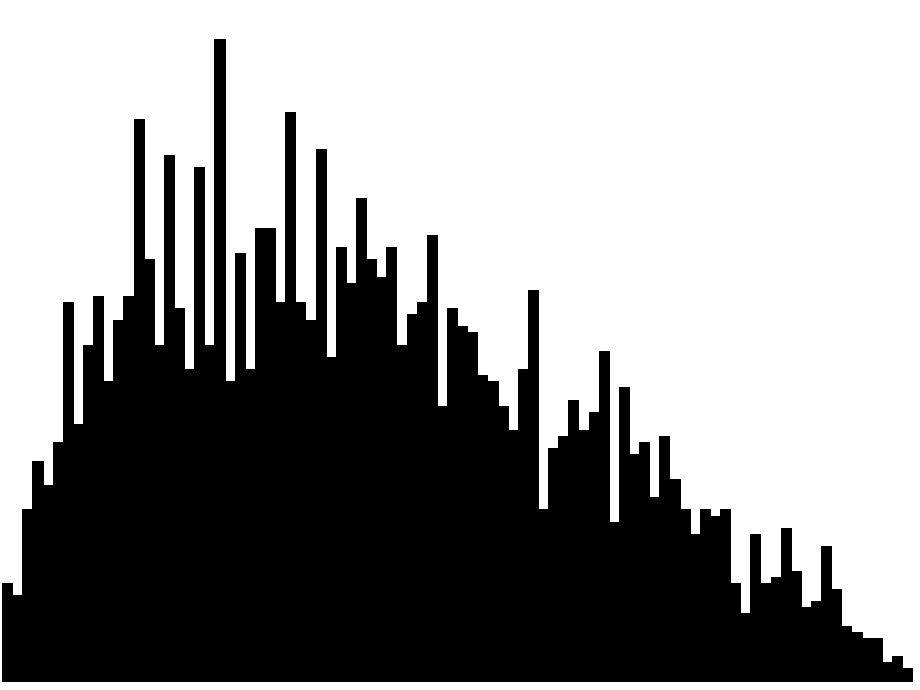} & \ih{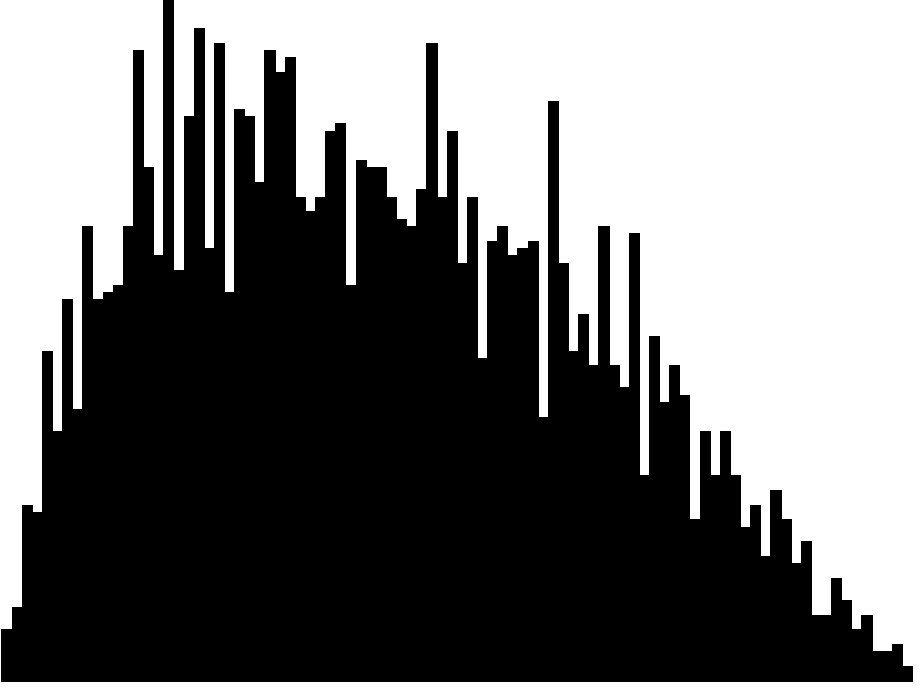} & \ih{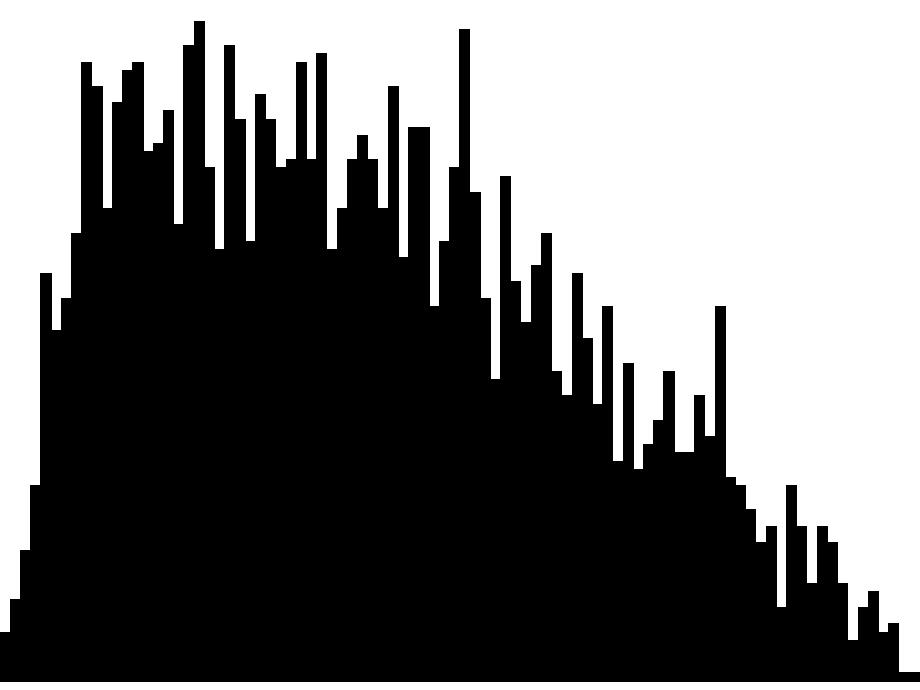} \\
\bottomrule
\end{tabular}}
\end{table*}

Roughly speaking, multichannel algorithms yield better results than multilateration. Among the algorithms belonging to
the latter class, \textit{n-mult}, \textit{n-mult} , and \textit{n-mult}  (jointly referred to as \textit{x-mult})
perform better than \textit{pi} in low-reverberant environments, independently of the level of noise. However, in high
reverberant environments, \textit{x-mult} performs slightly worse than \textit{pi}. It is worth noticing that, the value
of the parameter $r$ (in the constrained optimization formulation \refeq{eq:final_canclini}) has a small impact onto the
method's performance. Indeed, the variation of the inlier percentage is not greater than $2\%$ and the variation of the
mean and standard deviation is less than $1^\circ$.

All the multichannel TDE algorithms perform as expected with respect to the environmental parameters: The performance
decreases as $T_{60}$ is increased. However, the $SNR$ and $T_{60}$ have different effects on the objective function,
$J$. On one side, the sensor noise decorrelates the microphone signals leading to many randomly spread local minima and
increasing the value of the true minimum. If this effect is extreme, the hope for a good estimate decreases fast. On the
other side, the reverberations produce only a few strong local minima. This perturbation is systematic given the source
position in the room. Hence, there is hope to learn the effect of such reverberations in order to improve the quality of
the estimates. Clearly, these perturbation types (noise and reverberations) have different effects on the results.

Concerning the methods themselves, we noticed that \textit{unc} achieves a very low percentage of inliers. This fully
justifies the need of the geometric constraint introduced in this paper. In other words, the cost function \refeq{eq:J}
suffers from a lot of local minima outside the feasible domain. Thus, the naive idea of estimating the time delays
without adding information about the geometry of the microphone array, does not really work. We also remark that, except
for the two first cases (the easiest ones), the \textit{d-lb} method outperforms the \textit{dm} method, since the
former carries out a local minimization. Regarding the two easiest cases, it is not clear which of the two methods shows
a better performance. On one side, \textit{dm} captures more inliers. On the other side, both the mean angular error and
standard deviation over the set of inliers are significantly lower with \textit{d-lb} than with \textit{dm}. The sparse
initialization strategy does not show a remarkable performance. Indeed, the localization quality is comparable to
\textit{d-lb}, but the percentage of inliers is much lower. Thus, a method able to deal with large amounts of outliers
should be added in order to clean up the localization results provided by \textit{s-lb}. More importantly, we highlight
the performance of \textit{b\&b}. This method yields the highest percentage of inliers in all the tests. Moreover, the
quality of the localization is comparable, if not better, with \textit{d-lb} or with \textit{x-mult}.
\note{r1:3}{Regarding the percentage of outliers and the standard deviation, the \textit{b\&b} methods proves to
obtain the best results. Notice that multilateration methods often yield the best mean angular error of inliers.
However, these methods must be provided an estimate of the distance to the source: \textit{t-mult} was provided with
the true array-to-source distance. We conclude that \textit{b\&b} is the method of choice in the presence noise and
outliers.}

The behaviour of the methods described above on simulated data is similar on real data. The proposed multichannel
method outperforms the state-of-the-art on both multichannel TDE and multilateration. Among the multichannel algorithms,
\textit{unc} and \textit{s-lb} show very bad performance. Even if \textit{dm}, and specially \textit{d-lb} prove to work
to some extent, the best method is \textit{b\&b} since it has the highest percentage of inliers and the localization
quality is comparable to the one of \textit{d-lb} and of \textit{x-mult}. Finally, we noticed that the results on real
data roughly correspond to the simulated case with $T_{60}=0.6$~s and $SNR=-5$~dB, which is a very challenging scenario.

\note{r3:histograms}{The results of \reftab{tab:results} are complemented error histograms displayed in
\reftab{tab:histograms}. This table has a row-column structure that is strictly identical with \reftab{tab:results}.
The histograms count all localization trials, contrary to \reftab{tab:results} where only the inliers were used to
compute the mean and the standard deviation. \reftab{tab:histograms} is meant to provide a quantitative evaluation of
the error distribution. An ideal localization situation would show a histogram with a full leftmost bin ($0^\circ$
error) all the other bins being empty. One can see that the results of both tables are correlated. On one hand, good
localization results correspond to histograms whose mass is mainly concentrated to the left. On the other hand, bad
localization results correspond to histograms whose mass is evenly distributed over the histogram bins. However, we
observe three different histogram patterns among the methods reporting low performance. First, in some cases a very
rough estimate of the position could be extracted, namely the mass is concentrated on the left half of the histogram
(but not close to $0^\circ$), e.g., (\textit{b\&b},$-10$,$0.2$), (\textit{dm},$-5$,$0.4$) or (\textit{pi},$0$,$0.6$).
Second, in some cases a fairly accurate estimate is obtained, but only in 50\% of the trials. This translates into a
histogram with two large peaks, on of them close to $0^\circ$ and the other far apart. Consequently, some rough estimate
of the source position would allow to discard most of the outliers. Examples of this are the \textit{x-mult} methods.
Third, the worst case scenario, in which the error is uniformly distributed, do not provide any meaningful result,
namely for $SNR=-10$ dB and $T_{60}=0.6$ s.}

\begin{table}
 \caption{Average execution times and standard deviations for each method on 400 trials. All quantities are expressed
in seconds, except for those with *, which are expressed in milliseconds.}
 \label{tab:time}
\centering
 \begin{tabular}{c|c||cccc||cc}
\toprule
    Method & \textit{b\&b} & \textit{unc} & \textit{d-lb} & \textit{s-lb} & \textit{dm} & \textit{x-mult} &
\textit{pi} \\
  \midrule
   Mean & 9.17 & 25.13 & 10.34 & 0.958 & 0.103 & 2.02 & 13.5* \\
   Std & 1.082 & 2.759 & 1.316 & 0.810 & 3.64* & 0.265 & 0.310* \\
  \bottomrule
 \end{tabular}
\end{table}

Finally, we performed a statistical analysis of the execution times. \reftab{tab:time} shows the average and standard
deviation of the execution times for each one of the tested methods and  \note{r3:5}{on 400 trials}. All the methods
were implemented in MATLAB and the code was run on the very same computer. The methods \textit{n-mult}, \textit{t-mult}
and \textit{f-mult} were not evaluated separately because the parameter $r$ does not have any effect on the execution
time. We first observe that methods with low computational complexity correspond to methods that are not robust
(\textit{pi}, \textit{x-mult}, \textit{dm} and \textit{s-lb}). There are some methods that are neither robust nor fast
(\textit{unc}). A couple of methods present high robustness but high complexity (\textit{d-lb} and \textit{b\&b}).
However, \note{r:complexity}{\st{because our current software is an unoptimized implementation, we believe that the
execution time of the \textit{b\&b} method could be divided by a factor of $100$ and hence reach execution times of the
order of 0.5 seconds. Indeed, dedicated code-}} optimization techniques and smart approximations will lead to
\textit{b\&b}-based localization algorithms that are both efficient and robust. \note{r2:nao}{Indeed, platform-dedicated
algorithm optimization will reduce the computational time of the proposed sound source localization procedure. Moreover,
the accuracy of the localization results may be adjusted to the desired application/environment. For example, we could
use the proposed framework to obtain a rough estimate of the sound source position. The semantic context of the ongoing
social interplay will help us select the location of interest among the rough estimates. This coarse location of
interest could be then refined with the very same algorithm, but with a much smaller search space.}

\section{Conclusions and Future Work}
\label{sec:conclusions}
In this paper we addressed the problem of sound source localization from time delay estimates using non-coplanar
microphone arrays. Starting from the direct path propagation model, we derived the full geometric analysis associated
with an arbitrarily shaped non-coplanar microphone array. The necessary and sufficient conditions for the time delays to
correspond to a position in the source space are expressed by means of two feasibility conditions. If they are
satisfied, the position of the sound source can be recovered in closed-form, from the TDEs. Remarkably, the only
knowledge required to build the feasibility conditions and the localization mapping is the microphones' position. A
multichannel criterion for TDE allows us to cast the problem into an optimization task, which is
constrained by the feasibility conditions. A branch-and-bound global optimization technique is proposed to solve the
programming task, and hence to estimate the time delays and to localize the sound source. An extensive set of experiments
is performed on simulated and real data. The experiments clearly show that the global optimization technique that we
proposed outperforms existing methods in both the multilateration and the multichannel SSL literatures.

This work could be extended in several ways. First of all, considering the multiple source case. This could be achieved
using a frequency filter bank, that would also discard empty frequency bands as in \cite{Valin2006}. Second, a
different set of experiments could be performed on distributed microphone arrays, to evaluate the behaviour of the
proposed methods in such settings. Third, the method could also be used in calibration applications. Indeed, the
positions of the microphones could be estimated if they were free parameters in our current formulation. In that case,
measures from many different source positions would certainly be required, e.g., \cite{KFH13}. Fourth, by testing the proposed model and algorithms
in the case of dynamic sources, and subsequently extending the framework to perform tracking. Finally, experiments with
higher number of microphones should be performed, and the influence of the microphones' positions should be evaluated.

\appendices

\section{The Derivatives of the Cost Function}
\label{sec:criterion_grad_hess}
The log-barrier algorithm relies on the use of the gradient and the Hessian of both, the
objective function and the constraint(s). Providing the analytic expression for them would lead to a much more efficient
and precise algorithm than estimating them using finite differences. Hence, this Section is devoted to the derivation of
both the gradient and the Hessian of $J$, the cost function of \refeq{eq:const-optim}.

We will use three formulas from matrix calculus. Let $\Ymat:\mathbb{R}\rightarrow\mathbb{R}^{M\times M}$, be
a matrix function depending on $y$, the following formulas hold:
\begin{itemize}
 \item $\displaystyle \frac{\partial\, \det\left( \Ymat \right)} {\partial y} = \det(\Ymat) \trace \left( \Ymat^{-1}
\frac{\partial \Ymat}{\partial y} \right) $ \vspace{0.2cm}
  \item $\displaystyle\frac{\partial\, \trace\left( \Ymat \right)} {\partial y} = \trace\left( \frac{\partial\, \Ymat} {\partial y} \right) $ \vspace{0.2cm}
 \item $\displaystyle \frac{\partial\, \Ymat^{-1}}{\partial y} = - \Ymat^{-1}\frac{\partial \Ymat}{\partial y}\Ymat^{-1}$
\end{itemize}

Recall that the function we want to derivative is $J = \det\left(\Rmat\right)$. From the rules of matrix
calculus we have:
\begin{equation}
 \frac{\partial J}{\partial t_{1,m}} = \frac{\partial\,\det\left(\Rmat\right)}{\partial t_{1,m}} = 
\det\left(\Rmat\right)\trace\left(\Rmat^{-1}\frac{\partial\,\Rmat}{\partial t_{1,m}}\right).
\end{equation}

In addition we can compute the second derivative:
\begin{equation}
 \frac{\partial^2 J}{\partial t_{1,n}\partial t_{1,m}} =
\frac{\partial}{\partial t_{1,n}}\left[\det\left(\Rmat\right)\trace\left(\Rmat^{-1}\frac{\partial\,\Rmat}{\partial
t_{1,m}}\right)\right],
\end{equation}
whose full expression \ifCLASSOPTIONtwocolumn can be found in \refeq{eq:Jd2}.\else is:\fi
\ifCLASSOPTIONtwocolumn
\begin{figure*}
 \begin{equation}
   \frac{\partial^2\,\tilde{J}}{\partial t_{1,j}\partial t_{1,k}} =
 \det\left(\Rmat\right)
\left[\trace\left(\Rmat^{-1}\frac{\partial\,\Rmat}{\partial
t_{1,j}}\right)\trace\left(\Rmat^{-1}\frac{\partial\,\Rmat}{\partial
t_{1,k}}\right) + \trace\left( -\Rmat^{-1}\frac{\partial\,\Rmat}{\partial
t_{1,j}}\Rmat^{-1}\frac{\partial\,\Rmat}{\partial t_{1,k}} +
\Rmat^{-1}\frac{\partial^2\,\Rmat}{\partial t_{1,j}\partial t_{1,k}}\right)\right].
\label{eq:Jd2}
 \end{equation}
\end{figure*}
\else
 \begin{equation}
   \frac{\partial^2\,\tilde{J}}{\partial t_{1,j}\partial t_{1,k}} =
 \det\left(\Rmat\right)
\left[\trace\left(\Rmat^{-1}\frac{\partial\,\Rmat}{\partial
t_{1,j}}\right)\trace\left(\Rmat^{-1}\frac{\partial\,\Rmat}{\partial
t_{1,k}}\right) + \trace\left( -\Rmat^{-1}\frac{\partial\,\Rmat}{\partial
t_{1,j}}\Rmat^{-1}\frac{\partial\,\Rmat}{\partial t_{1,k}} +
\Rmat^{-1}\frac{\partial^2\,\Rmat}{\partial t_{1,j}\partial t_{1,k}}\right)\right].
\label{eq:Jd2}
 \end{equation}
\fi
Hence, in order to compute the first and second derivatives of the criterion $J$, we need the first and second
derivatives of the matrix $\Rmat$. We recall its expression:
{\small\[\Rmat =\left(
\begin{array}{cccc}
 1 & \rho_{1,2}(t_{1,2}) & \rho_{1,3}(t_{1,3}) & \cdots \\
 \rho_{1,2}(t_{1,2}) & 1 & \rho_{2,3}(-t_{1,2}+t_{1,3}) & \cdots \\
 \rho_{1,3}(t_{1,3}) & \rho_{2,3}(-t_{1,2}+t_{1,3}) & 1 & \cdots \\
 \vdots & \vdots & \vdots & \ddots \\
\end{array}
\right).\]}

We notice that only the $m-1\th$ row and column depend on $t_{1,m}$. Since $\Rmat$ is symmetric, we do not need to take
derivative of the $m-1\th$ row and column separately, but compute only the derivative of:
{\small\[\rvect_m =\left(\begin{array}{c}
\rho_{1,m}(t_{1,m}) \\ \rho_{2,m}(-t_{1,2}+t_{1,m}) \\ \vdots \\ \rho_{m-1,m}(-t_{1,m-1}+t_{1,m}) \\ 1
\\ \rho_{m,m+1}(-t_{1,m}+t_{1,m+1}) \\ \vdots \\ \rho_{m,M}(-t_{1,m}+t_{1,M}) \end{array}
\right),\]}
which is
{\small\[\frac{\partial\rvect_m}{\partial t_{1,m}} =\left(\begin{array}{c}
\rho_{1,m}'(t_{1,m}) \\ \rho_{2,m}'(-t_{1,2}+t_{1,m}) \\ \vdots \\ \rho_{m-1,m}'(-t_{1,m-1}+t_{1,m}) \\ 0
\\ -\rho_{m,m+1}'(-t_{1,m}+t_{1,m+1}) \\ \vdots \\ -\rho_{m,M}'(-t_{1,m}+t_{1,M}) \end{array}
\right).\]}

When computing the second derivative of $\Rmat$ with respect to $t_{1,m}$ and $t_{1,n}$ we need to differentiate two
cases:\\

\noindent\fbox{$m=n$} This fills the diagonal of the Hessian. Notice that:
{\small\[\frac{\partial^2\rvect_m}{\partial t_{1,m}^2} =\left(\begin{array}{c}
\rho_{1,m}''(t_{1,m}) \\ \rho_{2,m}''(-t_{1,2}+t_{1,m}) \\ \vdots \\ \rho_{m-1,m}''(-t_{1,m-1}+t_{1,m}) \\ 0
\\ \rho_{m,m+1}''(-t_{1,m}+t_{1,m+1}) \\ \vdots \\ \rho_{m,M}''(-t_{1,m}+t_{1,M}) \end{array}
\right).\]}

\noindent\fbox{$m>n$} This fills the lower triangular matrix of the Hessian (and the upper triangular part since the
Hessian is symmetric, i.e., that $\tilde{J}$ is twice continuously differentiable). Only the $n-1\th$ position of the
vector $\displaystyle\frac{\partial^2\rvect_m}{\partial t_{1,n}\partial t_{1,m}}$ is not null, taking the value:
$-\rho_{m,n}''(t_{1,m}-t_{1,n})$.

\section{The Derivatives of the Constraints}
\label{sec:constraint_grad_hess}
In this Section we compute the formulae for the first and the second derivatives of the non-linear constraint $\Delta$.
Recall the expression from \refeq{eq:delta}:
\[\Delta = \left<\Avect,\Bvect-\Mvect_1\right>^2 - \|\Bvect-\Mvect_1\|^2\left(\|\Avect\|^2-1\right),\]
where $\Avect=-\Mmat_L^{-1}\Pvect_L$ and $\Bvect = -\Mmat_L^{-1}\Qvect_L$. It is easy to show that:
{\small\begin{eqnarray*}
 \nonumber\nabla\Delta &=& 2\left(\left<\Avect,\Bvect-\Mvect_1\right>\left(\Jmat^\tr_{\Avect} (\Bvect-\Mvect_1) +
\Jmat^\tr_{\Bvect}\Avect\right)\right. -\\
 &-& \left.(\|\Avect\|^2-1)\Jmat^\tr_{\Bvect}(\Bvect-\Mvect_1)-
\|\Bvect-\Mvect_1\|^2\Jmat^\tr_{\Avect}\Avect\right) 
\end{eqnarray*}}
where $\Jmat_{\Avect} = -2\nu\Mmat_L^{-1}$ and $\Jmat_{\Bvect}=-2\nu^2\Mmat_L^{-1}\,\diag(\hat{\tvect})$.
We can also compute the Hessian of $\Delta$:
{\small\begin{eqnarray*}
 \nonumber\Hmat\Delta &=& 2\left(\left(\Jmat_{\Avect}^\tr(\Bvect-\Mvect_1) + \Jmat_{\Bvect}^\tr\Avect\right)
\left(\Jmat_{\Avect}^\tr(\Bvect-\Mvect_1) + \Jmat_{\Bvect}^\tr\Avect\right)^\tr + \right. \\
 \nonumber&+& \left<\Avect,\Bvect-\Mvect_1\right>\left(\Jmat_{\Avect}^\tr\Jmat_{\Bvect} + \Dmat +
\Jmat_{\Bvect}^\tr\Jmat_{\Avect}\right) - \\
 \nonumber&-&\left[ 2(\Jmat_{\Bvect}^\tr(\Bvect-\Mvect_1))(\Jmat_{\Avect}^\tr\Avect)^\tr + (\|\Avect\|^2-1)(\Emat +
\Jmat_{\Bvect}^\tr\Jmat_{\Bvect})\right. +  \\
 &+& \left.\left.2(\Jmat_{\Avect}^\tr\Avect)(\Jmat_{\Bvect}^\tr(\Bvect-\Mvect_1))^\tr +
\|\Bvect-\Mvect_1\|^2\Jmat_{\Avect}^\tr\Jmat_{\Avect}\right]\right)
\end{eqnarray*}}
where $\Dmat = -2\nu^2\,\diag(\Mmat_L^{-1}\Avect)$ and $\Emat = -2\nu^2\,\diag(\Mmat_L^{-1}(\Bvect-\Mvect_1))$.

Similarly, we can compute the derivatives of the equality constraint ${\cal E}$.

\section{The Derivatives of the Multilateration Cost Function}
\label{sec:derivatives_H}
In this Section we provide the first and second derivatives of the cost function used by the methods \textit{n-mult},
\textit{t-mult} and \textit{f-mult}. Denoted by $H$, the cost function has the
following expression:
\begin{equation}
 H(\Svect) = \sum_{1\leq m<n\leq M}\left(h_{m,n}(\Svect)\right)^2,
\end{equation}
where $h_{m,n}$ are the equivalent of \refeq{eq:normalized_hyperboloid} with foci $\Mvect_m$ and $\Mvect_n$ and
differential value $\hat{t}_{m,n}$. In all, $h_{m,n}$ takes the following expression:
\begin{eqnarray*}h_{m,n}(\Svect) &=& q_{m,n}^2+4\left<\Svect,\Mvect_n-\Mvect_m\right>^2-\\&
&-4q_{m,n}\left<\Svect,\Mvect_n-\Mvect_m\right>-p_{m,n}^2\|\Svect-\Mvect_n\|^2.\end{eqnarray*}
The gradient of $H$ writes:
\begin{equation*}
 \nabla H = 2\sum_{1\leq m<n\leq M} h_{m,n}\nabla h_{m,n},
\end{equation*}
where
\begin{eqnarray*}
\nabla h_{m,n} &=& 8\left<\Svect,\Mvect_n-\Mvect_m\right>\left(\Mvect_n-\Mvect_m\right)-\\&
&-4q_{m,n}\left(\Mvect_n-\Mvect_m\right)-2p_{m,n}^2\left(\Svect-\Mvect_n\right).
\end{eqnarray*}
Similarly, the Hessian of $H$ can be computed as:
\begin{equation*}
 \Hmat H = 2\sum_{1\leq m<n\leq M} \nabla h_{m,n}\left(\nabla h_{m,n}\right)^\tr + h_{m,n}\Hmat h_{m,n},
\end{equation*}
where the Hessian of $h_{m,n}$ has the following expression:
\begin{eqnarray*}
\Hmat h_{m,n} &=& 8\left(\Mvect_n-\Mvect_m\right)\left(\Mvect_n-\Mvect_m\right)\tr-2p_{m,n}^2\Imat_N.
\end{eqnarray*}

\bibliographystyle{plain}
\bibliography{gtde}

%

\end{document}